%% file: main.tex
\documentclass[a4paper,10pt]{article}
\usepackage[utf8]{inputenc}
\usepackage{fancyhdr}
\usepackage{graphicx}
\usepackage{amsmath}
\usepackage{amssymb}
\usepackage{hyperref}
\usepackage[all]{hypcap}
\usepackage{pdfpages}
\usepackage{mathrsfs}
\usepackage{caption}
\usepackage{subcaption}
\usepackage{bm}
\usepackage[margin=0.95in]{geometry}
\newtheorem{theorem}{Theorem}[section]
\newtheorem{lemma}[theorem]{Lemma}
\newtheorem{proposition}[theorem]{Proposition}

\newenvironment{proof}[1][Proof]{\begin{trivlist}
\item[\hskip \labelsep {\bfseries #1}]}{\end{trivlist}}
\newenvironment{definition}[1][Definition]{\begin{trivlist}
\item[\hskip \labelsep {\bfseries #1}]}{\end{trivlist}}

\newenvironment{remark}[1][Remark]{\begin{trivlist}
\item[\hskip \labelsep {\bfseries #1}]}{\end{trivlist}}

\newcommand{\qed}{\nobreak \ifvmode \relax \else
      \ifdim\lastskip<1.5em \hskip-\lastskip
      \hskip1.5em plus0em minus0.5em \fi \nobreak
      \vrule height0.75em width0.5em depth0.25em\fi}
\title{Stochastic Recursive Inclusions in two timescales with non-addtive iterate dependent Markov noise}
\author{Vinayaka G. Yaji and Shalabh Bhatnagar,\\ Department of Computer Science and Automation,
        \\ Indian Institute of Science, Bangalore.\\
        vgyaji@gmail.com, shalabh@csa.iisc.ernet.in}
\begin{document}
\maketitle
\input{abstract_intro}

\input{background}

\input{recursion_and_assumptions}

\input{lim_di_and_their_prop}

\input{recursion_analysis}

\input{application}

\input{conc_and_direc_for_future_work}
\input{apnd}

\bibliographystyle{IEEEtran}
\bibliography{Ref}
\end{document}

%% file: abstract_intro.tex
\begin{abstract}
 In this paper we study the asymptotic behavior of a stochastic approximation scheme on two timescales with set-valued drift functions and 
 in the presence of non-additive iterate-dependent Markov noise. It is shown that the recursion on each timescale tracks the flow of a 
 differential inclusion obtained by averaging the set-valued drift function in the recursion with respect to a set of measures which take into 
 account both the averaging with respect to the stationary distributions of the Markov noise terms and the interdependence between the two 
 recursions on different timescales. The framework studied in this paper builds on the works of \it{A. Ramaswamy et al. }\rm by allowing for the 
 presence of non-additive iterate-dependent Markov noise. As an application, we consider the problem of computing the optimum in a constrained 
 convex optimization problem where the objective function and the constraints are averaged with respect to the stationary distribution of an 
 underlying Markov chain. Further the proposed scheme neither requires the differentiability of the objective function nor the knowledge of the 
 averaging measure. 
\end{abstract}

\section{Introduction}
\label{intro}
Consider the standard two timescale stochastic approximation scheme given by,
\begin{subequations}\label{brktwo}
 \begin{align}\label{sloweg}
  Y_{n+1}-Y_{n}-b(n)M^{(2)}_{n+1}&=b(n)h_2(X_n,Y_n),\\
  \label{fasteg}
  X_{n+1}-X_{n}-a(n)M^{(1)}_{n+1}&=a(n)h_1(X_n,Y_n),
 \end{align}
\end{subequations}
where $n\geq0$ denotes the iteration index, $\{X_n\}_{n\geq0}$ is a sequence of $\mathbb{R}^{d_1}$-valued random variables, $\{Y_n\}_{n\geq0}$ is 
a sequence of $\mathbb{R}^{d_2}$-valued random variables, for any $i\in\{1,2\}$, $h_i:\mathbb{R}^{d_1+d_2}\rightarrow\mathbb{R}^{d_i}$ is a Lipschitz 
continuous function and $\{M^{(i)}_n\}_{n\geq1}$ is sequence of $\mathbb{R}^{d_i}$-valued square integrable martingale difference sequence. The 
step size sequences, $\{a(n)\}_{n\geq0}$ and $\{b(n)\}_{n\geq0}$ are sequences of positive real numbers chosen such that they satisfy 
$\lim_{n\to\infty}\frac{b(n)}{a(n)}=0$ in addition to the Monte Carlo step size conditions. The condition $\lim_{n\to\infty}\frac{b(n)}{a(n)}=0$, 
ensures that after large number of iterations the time step of recursion \eqref{sloweg} is much smaller than that of \eqref{fasteg}. Thus the 
recursion \eqref{sloweg} appears to be static with respect to the recursion \eqref{fasteg}. In \cite{borkartwo}, using the dynamical systems 
approach studied in \cite{benaim2}, the above intuition was shown to hold. More precisely, the faster timescale recursion \eqref{fasteg}, was 
shown to track the ordinary differential equation (o.d.e.) given by,
\begin{equation}\label{fastegode}
 \frac{dx}{dt}=h_1(x,y_0),
\end{equation}
for some $y_0\in\mathbb{R}^{d_2}$ and assuming that for every $y\in\mathbb{R}^{d_2}$, o.d.e. \eqref{fastegode} admits a unique globally 
asymptotically stable equilibrium point, say $\lambda(y)$, the slower timescale recursion \eqref{sloweg} was shown to track the o.d.e. given by,
\begin{equation}\label{slowegode}
 \frac{dy}{dt}=h_2(\lambda(y),y).
\end{equation}
Further, the map $y\rightarrow\lambda(y)$ was assumed to be Lipschitz continuous. 

An important application of the above stochastic approximation scheme is in the computation of a saddle point of a function. Given a function 
$f:\mathbb{R}^{d_1}\times\mathbb{R}^{d_2}\rightarrow\mathbb{R}$, $(x^*,y^*)\in\mathbb{R}^{d_1+d_2}$ ($x^*\in\mathbb{R}^{d_1}$ and 
$y^*\in\mathbb{R}^{d_2}$ respectively) is a saddle point of the function $f(\cdot)$ if, 
\begin{equation}
 \inf_{x\in\mathbb{R}^{d_1}}\sup_{y\in\mathbb{R}^{d_2}}f(x,y)=\sup_{y\in\mathbb{R}^{d_2}}\inf_{x\in\mathbb{R}^{d_1}}f(x,y)=f(x^*,y^*).
\end{equation}
From \cite[Prop.~5.5.6]{bertsekas}, we know that the function $f(\cdot)$ admits a saddle point if for every $(x,y)\in\mathbb{R}^d$,
\begin{itemize}
 \item [(1)] $-f(x,\cdot)$ and $f(\cdot,y)$ are convex functions,
 \item [(2)] the sub level sets of functions $x\rightarrow\sup_{y\in\mathbb{R}^{d_2}}f(x,y)$ and $y\rightarrow-\inf_{x\in\mathbb{R}^{d_1}}f(x,y)$ 
 are compact sets.
\end{itemize}
Over the years significant effort has been devoted for developing algorithms to compute such points (see \cite{nedic,benzi} and references 
therein). Most of the solutions proposed in literature require the computation of partial derivatives of the function $f(\cdot)$. However in practice
the closed form expressions of the partial derivatives are often not known or are expensive to compute and in such cases one often estimates the 
partial derivatives using values of the objective function (see \cite{spall} for one such estimation method). The two timescale stochastic 
approximation scheme can be used to compute a saddle point with noisy partial derivative values by setting $h_1(\cdot):=-\nabla_{x}f(\cdot)$ 
and $h_2(\cdot):=\nabla_{y}f(\cdot)$ where $\nabla_{x}$ and $\nabla_y$ denote the partial derivative operators with respect to $x$ and $y$ 
respectively. In this setting, the sequences $\{M_n^{(1)}\}_{n\geq1}$ and $\{M_n^{(2)}\}_{n\geq1}$ denote the partial derivative estimation 
errors and the map $\lambda(\cdot)$ denotes correspondence between $y\in\mathbb{R}^{d_2}$ and the minimum of the function $f(\cdot,y)$. 
The vector field associated with o.d.e. \eqref{slowegode} is now given by $\nabla_{y}f(x,y)|_{x=\lambda(y)}$ which can be shown to be 
the same as $\nabla_{y}f(\lambda(y),y)$ under some conditions known as \it{envelope theorem }\rm in mathematical economics (see \cite{milgrom}). Thus the 
slower timescale maximizes the function $y\rightarrow\inf_{x\in\mathbb{R}^{d_1}}f(x,y)=f(\lambda(y),y)$, there by in the limit the iterates of 
recursion \eqref{brktwo} converge to a saddle point of the function $f(\cdot)$. 

In some cases the function whose saddle point needs to be computed is itself averaged with respect to a certain probability measure. For example 
consider the function $f:\mathbb{R}^{d_1+d_2}\times\mathcal{S}\rightarrow\mathbb{R}$ where $\mathcal{S}$ is a compact metric space and for some 
probability measure $\mu$ on $\mathcal{S}$, one wishes to compute the saddle point of the function 
$f_{\mu}:\mathbb{R}^{d_1+d_2}\rightarrow\mathbb{R}$ where for every $(x,y)\in\mathbb{R}^{d_1+d_2}$, 
$f_{\mu}(x,y):=\int_{\mathcal{S}}f(x,y,s)\mu(ds)$. If one has access to i.i.d. samples with probability measure $\mu$, then the saddle point 
problem above can be solved using recursion \eqref{brktwo}. But if access to such samples are not available and one uses Markov chain 
Monte Carlo methods to sample from the measure $\mu$, then the recursion \eqref{brktwo} has a non-additive iterate-dependent Markov noise 
component. The recursion now takes the form:
\begin{subequations}\label{pktwo}
 \begin{align}\label{slowpk}
  Y_{n+1}-Y_{n}-b(n)M^{(2)}_{n+1}&=b(n)h_2(X_n,Y_n,S^{(2)}_n),\\
  \label{fastpk}
  X_{n+1}-X_{n}-a(n)M^{(1)}_{n+1}&=a(n)h_1(X_n,Y_n,S^{(1)}_n),
 \end{align}
\end{subequations}
where $\{S^{(1)}_n\}_{n\geq1}$ and $\{S^{(2)}_n\}_{n\geq1}$ denote the Markov noise terms taking values in an appropriate state space. The 
recursion \eqref{pktwo} was studied in \cite{prasan}, under assumptions similar to those in \cite{borkartwo} which include the Lipschitz 
continuity of the maps $h_1(\cdot),\ h_2(\cdot)$ and $\lambda(\cdot)$.  

Often, the maps $h_1(\cdot)$ and $h_2(\cdot)$ in recursion \eqref{brktwo} are not Lipschitz continuous and the map $\lambda(\cdot)$ is not even 
single valued (that is the o.d.e. \eqref{fastegode} has a globally asymptotically stable equilibrium set). This motivates one to study the two 
timescale recursion with set-valued drift functions. The recursion now takes the form:
\begin{subequations}\label{aruntwo}
 \begin{align}\label{slowarun}
  Y_{n+1}-Y_{n}-b(n)M^{(2)}_{n+1}&\in b(n)H_2(X_n,Y_n),\\
  \label{fastarun}
  X_{n+1}-X_{n}-a(n)M^{(1)}_{n+1}&\in a(n)H_1(X_n,Y_n),
 \end{align}
\end{subequations}
where $H_1(\cdot)$ and $H_2(\cdot)$ are set-valued maps and other quantities have similar interpretation to those in recursion \eqref{brktwo}. 
The above recursion was studied in \cite{arun2t} and further the map $\lambda(\cdot)$ was allowed to be set-valued and upper semicontinuous.

\subsection{Contributions of this paper and comparisons with the state of the art}

In this paper we study the asymptotic behavior of the recursion given by,
\begin{subequations}\label{vgytwo}
 \begin{align}\label{slowvgy}
  Y_{n+1}-Y_{n}-b(n)M^{(2)}_{n+1}&\in b(n)H_2(X_n,Y_n,S^{(2)}_n),\\
  \label{fastvgy}
  X_{n+1}-X_{n}-a(n)M^{(1)}_{n+1}&\in a(n)H_1(X_n,Y_n,S^{(1)}_n),
 \end{align}
\end{subequations}
where $H_1(\cdot)$ and $H_2(\cdot)$ are set-valued maps and $\{S^{(1)}_n\}_{n\geq0}$ and $\{S^{(2)}_n\}_{n\geq0}$ are the Markov noise terms 
taking values in compact metric spaces $\mathcal{S}^{(1)}$ and $\mathcal{S}^{(2)}$ respectively. We show that the fast timescale recursion 
\eqref{fastvgy}, tracks the flow of the differential inclusion (DI) given by,
\begin{equation}
 \frac{dx}{dt}\in \cup_{\mu\in D^{(1)}(x,y_0)} \int_{\mathcal{S}^{(1)}}H_1(x,y_0,s^{(1)})\mu(ds^{(1)}),
\end{equation}
for some $y_0\in\mathbb{R}^{d_2}$, where $D^{(1)}(x,y)$ denotes the set of stationary distributions of the Markov noise terms 
$\{S^{(1)}_n\}_{n\geq0}$ for every $(x,y)\in\mathbb{R}^{d}$ and the integral above denotes the integral of a set-valued map with respect 
to measure $\mu$. Further we assume that for every $y\in\mathbb{R}^{d_2}$, the above DI admits a unique globally attracting set $\lambda(y)$.
The map $y\rightarrow\lambda(y)$ is also assumed to be upper semicontinuous. The slower timescale recursion \eqref{slowvgy}, is show to track the 
flow of the DI given by,
\begin{equation}
\label{slowdivgy}
 \frac{dy}{dt}\in \cup_{\mu\in D(y)}\int_{\mathbb{R}^{d_1}\times\mathcal{S}^{(2)}}H_2(x,y,s^{(2)})\mu(dx,ds^{(2)}),
\end{equation}
where $y\rightarrow D(y)$ denotes a set-valued map taking values in the space of probability measures on $\mathcal{S}^{2}$ and the map $D(\cdot)$
is defined such that it captures both the equilibration of the fast timescale iterates to $\lambda(\cdot)$ and the averaging due to the Markov 
noise terms $\{S^{(2)}_n\}_{n\geq0}$. 

In comparison with the two timescale framework studied in \cite{prasan}, our work allows for the drift functions (that is $H_1(\cdot)$ 
and $H_2(\cdot)$) to be set-valued and further the map $\lambda(\cdot)$ is allowed to be set-valued and upper semicontinuous which is much weaker 
than the requirement of single valued and Lipschitz continuity imposed in \cite{prasan}. The generalization to the set-valued case allows one 
to analyze recursion \eqref{pktwo} when the drift functions $h_1(\cdot)$ and $h_2(\cdot)$ are single valued and are just measurable, since 
graph of such a map can be embedded in the graph of a upper semicontinuous set-valued map as in \cite[ch.~5.3(iv)]{borkartxt}. We refer the reader 
to \cite[ch.~5.3]{borkartxt} for several other scenarios where the study of stochastic approximation scheme with set-valued maps becomes 
essential.

Our work further generalizes the two timescale framework studied in \cite{arun2t} by allowing for the presence of Markov noise terms. The 
analysis in this paper does not extend in an straight forward manner from those in \cite{arun2t} and requires results from set-valued map 
approximation, parametrization, integration and the use of probability measure valued functions. However the method of analysis adopted in this paper 
can be adapted appropriately to obtain the same convergence guarantees as in \cite{arun2t} when the Markov noise terms are absent.

\subsection{Overview of the analysis and organization of the paper}
It is known that continuous, convex and compact set-valued maps taking values in a finite dimensional space admit a continuous single-valued 
parametrization. The properties of the set-valued drift function ensure that the drift functions $H_1(\cdot)$ and $H_2(\cdot)$ are 
convex and compact set-valued maps and is upper semicontinuous. However such maps do not admit a continuous parametrization. We can work around 
this problem by enlarging the graph of the drift function since the graph of drift function can be embedded in the graph of a continuous, convex 
and compact set-valued map which admit a continuous single-valued parametrization. Thus a sequence of continuous single-valued maps can be 
obtained which approximate the set-valued drift function from above. This enables us to write the inclusion \eqref{vgytwo} in the form of 
recursion \eqref{pktwo} with an additional parameter. The results needed to accomplish the above are stated in section \ref{ussmata}.

Before proceeding further one needs to identify the mean fields that the recursion \eqref{vgytwo} is expected to track. To this end we need some 
results from the theory of integration of set-valued maps which are reviewed in section \ref{msmai}. Further the measurablility and integrability 
properties of the drift functions of the recursion are investigated and the characterization of the integral of a continuous set-valued map
in terms of its parametrization is established.

In section \ref{diatls} we compile some definitions and results from the theory of differential inclusions which are needed later to characterize the 
asymptotic behavior of recursion \eqref{vgytwo}. Further in section \ref{stsri} we state the assumptions and the main result of the analysis of 
single timescale stochastic recursive inclusions with non-additive iterate dependent Markov noise from \cite{vin1t} and in section \ref{sopmvf} 
we define and compile some results needed from the space of probability measure valued functions.

In section \ref{recass} we state and motivate the assumptions under which the recursion \eqref{vgytwo} is analyzed. Using the results from 
integration of set-valued maps reviewed in section \ref{msmai} the mean fields are defined and the main convergence result is stated. The mean 
fields defined in section \ref{recass} possess some properties which ensure existence of solutions (of their associated differential inclusions).
These properties are established in section \ref{limdiprop}. In section \ref{limdiprop} it is also shown that appropriate modifications of the 
continuous set-valued maps which approximate the drift functions (obtained in section \ref{ussmata}) also approximate the mean fields which 
play an important role in the analysis later.

The analysis of recursion \eqref{vgytwo} consists of two parts. In section \ref{fstranal}, the recursion \eqref{vgytwo} is analyzed along the 
faster timescale. The recursion \eqref{vgytwo} when viewed along the faster timescale appears to be a single timescale stochastic recursive 
inclusion with non-additive iterate dependent Markov noise. In section \ref{fstranal}, we show that recursion \eqref{vgytwo} viewed along the 
faster timescale satisfies all the assumptions associated with the single timescale recursion presented in section \ref{stsri}. Applying the 
main result of single timescale analysis we conclude that the faster timescale iterates converge to $\lambda(\cdot)$ for some $y\in\mathbb{R}^{d_2}$.
In section \ref{stsranal}, the slower timescale recursion is analyzed. It is shown that the linearly interpolated sample path of the slower 
timescale iterates (defined in section \ref{stsranalprelim}) tracks an appropriate DI. Continuous functions tracking the flow of a dynamical 
system are known as \it{asymptotic pseudotrajectories }\rm (see \cite{benaim1} for definition and related results). The asymptotic pseudotrajectory 
argument in this paper presented in section \ref{stsranalapt} comprises of the following steps:
\begin{itemize}
 \item [(1)] First step is to get rid of the additive noise terms, $\{M^{(2)}_n\}_{n\geq1}$. This involves defining an o.d.e. with an appropriate 
 piecewise constant vector field and showing that the limit points of the shifted linearly interpolated trajectory of the slower timescale 
 iterates coincide with the limit points of the solutions of this o.d.e. in the space of continuous functions on $[0,\infty)$ taking values 
 in $\mathbb{R}^{d_2}$. Further a simple argument gives us that the set of limit points of the shifted linearly interpolated trajectories of the 
 slower timescale iterates is non-empty.
 \item [(2)] The second step is to show that the limit point obtained in the first step is in fact a solution of DI \eqref{slowdivgy}. This is 
 accomplished using probability measure valued functions reviewed in section \ref{sopmvf}. This method has also been used in analyzing stochastic 
 approximation schemes such as recursion \eqref{pktwo} in \cite{prasan} and in \cite{borkarmark}. But the analysis in these references made 
 explicit use of the Lipschitz property of the underlying drift functions. We observe that continuity is sufficient to carry out this analysis. 
 This is also where our analysis significantly differs from that in \cite{arun2t}. The equilibration of the faster timescale is also accomplished 
 using probability measures which simplifies the proof compared to that in \cite{arun2t}. 
\end{itemize}
In section \ref{stsranalcls}, the limit sets of the slower timescale iterates are characterized in terms of the dynamics of DI \eqref{slowdivgy}. 
In addition to the above, using the convergence of the faster timescale iterates to $\lambda(\cdot)$ obtained in section \ref{fstranal}, we obtain the 
main convergence result of this paper. 

In section \ref{appl}, as an application, we propose an algorithm to compute a solution of a constrained convex optimization problem. The objective 
function and constraints are assumed to be convex and affine respectively. Further the optimization problem is obtained by averaging the 
quantities involved with respect to the stationary distribution of an underlying Markov chain. Such problems arise in optimal control where the 
controller must find an optimum parameter where the changes in state of the underlying system can be modeled by a Markov chain. The cost function 
and system constraints are dependent on the state of the system and the controller seeks to find the optimum of the long run average of cost 
function while satisfying the long run average constraints. In such applications the stationary distribution of the system states are not known, 
but one has access to a sample path of system state changes. We propose a two timescale scheme which performs primal ascent along the 
faster timescale and dual descent along the slower timescale with the knowledge of the current state at a given iteration. Using the theory 
presented in this paper, it is shown that the limit set of the iterates of the proposed two timescale scheme are contained in the set of 
Lagrangian saddle points of the underlying averaged constrained convex optimization problem. Further the algorithm does not assume the 
differentiability of the objective function and requires only a noisy estimate of the subgradient.

In section \ref{cadffw},  we conclude by providing a few directions for future research and outline certain extensions where we believe the analysis 
remains the same.

%% file: background.tex
\section{Background}
In this section we shall briefly review some results needed from the theory of set-valued maps and differential inclusions, present a brief outline of the 
analysis of the single timescale version of stochastic recursive inclusions with non-additive iterate-dependent Markov noise and define the space 
of probability measure valued functions with a metrizable topology which are needed later in the analysis of the two timescale recursion.

Throughout this paper $\mathcal{S}$ denotes a compact metric space and the metric on $\mathcal{S}$ is denoted by $d_{\mathcal{S}}$. Further 
let $1\leq d_1\in\mathbb{Z}$, $1\leq d_2\in\mathbb{Z}$, $d:=d_1+d_2$ and $(x,y)$ denotes a generic element in $\mathbb{R}^d$ where 
$x\in\mathbb{R}^{d_1}$ and $y\in\mathbb{R}^{d_2}$.

\input{upper_semicontinuous_setvalued_maps_and_their_approximation}

\input{measurable_set_valued_map_and_integration}

\input{differential_inclusions_and_limit_sets}

\input{single_time_scale_sri_with_mn}

\input{space_of_prob_meas_valued_func}

%% file: upper_semicontinuous_setvalued_maps_and_their_approximation.tex
\subsection{Upper semicontinuous set-valued maps and their approximation}
\label{ussmata}
First we shall recall the notions of upper semicontinuity, lower semicontinuity and continuity of set-valued maps. These notions are 
taken from \cite[ch.~1.1]{aubindi}.
\begin{definition}
 A set valued map $F:\mathbb{R}^d\times\mathcal{S}\rightarrow\left\{\text{subsets of }\mathbb{R}^k\right\}$ is,
 \begin{itemize}
  \item \it{Upper semicontinuous }\rm(u.s.c.) if, for every $(x_0,y_0,s_0)\in\mathbb{R}^d\times\mathcal{S}$, for every $\epsilon>0$, there 
  exists $\delta>0$ (depending on $(x_0,y_0,s_0)$ and $\epsilon$) such that,
  \begin{equation*}
   \parallel (x,y)-(x_0,y_0)\parallel<\delta,\ d_{\mathcal{S}}(s,s_0)<\delta\implies F(x,y,s)\subseteq F(x_0,y_0,s_0)+\epsilon U,
  \end{equation*}
  where $U$ denotes the closed unit ball in $\mathbb{R}^k$.
  \item \it{Lower semicontinuous }\rm(l.s.c)  if, for every $(x_0,y_0,s_0)\in\mathbb{R}^d\times\mathcal{S}$, for every $z_0\in F(x_0,y_0,s_0)$, 
  for every sequence $\left\{\left(x_n,y_n,s_n\right)\right\}_{n\geq1}$ converging to $(x_0,y_0,s_0)$, there exists a sequence 
  $\left\{z_n\in F(x_n,y_n,s_n)\right\}$ converging to $z_0$.
  \item \it{Continuous }\rm if, it is both u.s.c. and l.s.c.
 \end{itemize}
\end{definition}
For set valued maps taking compact set values we have the above mentioned notion of u.s.c. to be equivalent to the standard notion of u.s.c.
(see \cite[pg.~45]{aubindi}). In this paper we shall encounter set valued maps which are compact set valued and hence we have chosen to state 
the above as the definition of upper semicontinuity.

Set-valued maps studied later satisfy certain properties under which we will be able to approximate them with a family of 
continuous single-valued maps with an additional parameter. These properties are natural extensions of the properties imposed 
on maps studied in \cite{benaim1,arun2t} to the case of stochastic recursive inclusions with Markov noise and we choose to call 
such maps \it{stochastic approximation maps }\rm(SAM). The definition of SAM is stated below. 
\begin{definition}\emph{[SAM]}
 \label{sam}
  A set-valued map $F:\mathbb{R}^{d}\times\mathcal{S}\rightarrow\left\{\text{subsets of }\mathbb{R}^{k}\right\}$ is a 
 stochastic approximation map if,
 \begin{itemize}
  \item [(a)]for every $\left(x,y,s\right)\in\mathbb{R}^{d}\times\mathcal{S}$, $F\left(x,y,s\right)$ is a convex and compact subset of $\mathbb{R}^{k}$,
  \item [(b)] for every $(x_0,y_0,s_0)\in\mathbb{R}^d\times\mathcal{S}$, for every $\mathbb{R}^{d}\times\mathcal{S}$ sequence, say 
  $\left\{\left(x_n,y_n,s_n\right)\right\}_{n\geq1}$ converging to $\left(x_0,y_0,s_0\right)$ and a sequence $\left\{z_n\in F(x_n,y_n,s_n)\right\}_{n\geq0}$ 
  converging to $z\in \mathbb{R}^{k}$, we have that $z\in F(x_0,y_0,s_0)$, 
  \item [(c)] there exists $K>0$ such that for every $(x,y,s)\in\mathbb{R}^{d}\times\mathcal{S}$, 
  $\sup_{z\in F(x,y,s)}\left\|z\right\|\leq K\left(1+\left\|(x,y)\right\|\right)$.
 \end{itemize}
\end{definition}

For SAM appearing in two-timescale stochastic recursive inclusions the condition $(c)$ stated above is replaced by an equivalent condition,
\begin{itemize}
 \item [(c)$'$] there exists $K>0$ such that for every $(x,y,s)\in\mathbb{R}^{d}\times\mathcal{S}$, 
  $\sup_{z\in F(x,y,s)}\left\|z\right\|\leq K\left(1+\left\|x\right\|+\left\|y\right\|\right)$.
\end{itemize}

The condition $(b)$ in the definition of SAM tells us that the graph of the set-valued map $F$, $\mathscr{G}(F)$, defined as
\begin{equation*}
 \mathscr{G}(F):=\left\{\left(x,y,s,z\right): z\in F(x,y,s),\ (x,y,s)\in\mathbb{R}^{d}\times\mathcal{S}\right\}\subseteq\mathbb{R}^d\times\mathcal{S}
 \times\mathbb{R}^k,
\end{equation*}
is closed and hence the said condition is known as the \it{closed graph property}\rm. The condition $(c)$ (or $(c)'$) is known as the 
\it{point-wise boundedness condition }\rm and it makes sure that the \lq size\rq\ of the sets grow linearly with the distance from the origin. 
This is the only condition where we differ from the conditions imposed in \cite{benaim1,arun2t}. It is easy to show that, when the 
Markov noise component is absent, condition $(c)$ (or $(c)'$) imposed in this paper is the same as the one in \cite{benaim1} (\cite{arun2t}). 

As a consequence of the properties possessed by a SAM, $F$, one can show that the map $F$ is u.s.c. This claim follows from arguments similar to 
that in \cite[ch.~1.1, Cor.~1]{aubindi} and is stated as a lemma below.

\begin{lemma}\emph{[u.s.c.]}\label{usc}
 A set-valued map $F$ which is a SAM is u.s.c.
\end{lemma}

The graph of a convex and compact u.s.c. set-valued map can be embedded in the graph of a sequence of decreasing continuous 
set-valued maps. The following statement is made precise in the following lemma. 

\begin{lemma}\emph{[continuous embedding]}\label{ctem}
 For any set-valued map $F$, a SAM, there exists a sequence of set-valued maps $\left\{F^{(l)}\right\}_{l\geq1}$ such that for every $l\geq1$, 
 $F^{(l)}:\mathbb{R}^d\times\mathcal{S}\rightarrow\left\{\text{subsets of }\mathbb{R}^k\right\}$ is continuous and satisfies the following.
 \begin{itemize}
  \item [(i)] For every $(x,y,s)\in\mathbb{R}^d\times\mathcal{S}$, $F^{(l)}(x,y,s)$ is a convex and compact subset of $\mathbb{R}^k$.
  \item [(ii)] For every $(x,y,s)\in\mathbb{R}^d\times\mathcal{S}$, $F(x,y,s)\subseteq F^{(l+1)}(x,y,s)\subseteq F^{(l)}(x,y,s)$.
  \item [(iii)] There exists $K^{(l)}>0$ such that for every $(x,y,s)\in\mathbb{R}^d\times\mathcal{S}$, 
  $\sup_{z\in F^{(l)}(x,y,s)}\left\|z\right\|\leq K^{(l)}(1+\left\|(x,y)\right\|)$. (If the set-valued map $F$ satisfies condition $(c)'$ instead of 
  $(c)$ in the definition of SAM, we have $\sup_{z\in F^{(l)}(x,y,s)}\left\|z\right\|\leq K^{(l)}(1+\left\|x\right\|+\left\|y\right\|)$).
 \end{itemize}
Furthermore, for every $(x,y,s)\in\mathbb{R}^d\times\mathcal{S}$, $\cap_{l\geq1}F^{(l)}(x,y,s)=F(x,y,s)$.
\end{lemma}

The statement of the above lemma can be found in \cite[pg.~39]{aubindi} and the proof is similar to the proof of \cite[ch.~1.13, Thm.~1]{aubindi} 
(a brief outline can be found in \cite[Appendix~A]{vin1t}). The following are some useful observations from the proof of Lemma \ref{ctem}.
\begin{itemize}
 \item [(1)] $\sup_{l\geq1}K^{(l)}$ is finite and let $\tilde{K}:=\sup_{l\geq1}K^{(l)}$,
 \item [(2)] for every $(x,y,s)\in\mathbb{R}^{d}\times\mathcal{S}$, for every $\epsilon>0$, there exists $L$ (depending on $\epsilon$ and 
 $(x,y,s)$), such that for every $l\geq L$, $F^{(l)}(x,y,s)\subseteq F(x,y,s)+\epsilon U$ where $U$ denotes the closed unit ball in $\mathbb{R}^k$.
\end{itemize}

Continuous set-valued maps admit a parametrization by which we mean that a 
continuous single-valued map can be obtained which represents the set-valued map in the sense made precise in the lemma below which follows from \cite[ch.~1.7, Thm.~2]{aubindi}

\begin{lemma}\emph{[parametrization]}\label{param}
 Let $F$ be a SAM and for every $l\geq1$, the set-valued map $F^{(l)}$ be as in Lemma \ref{ctem}. Then for every $l\geq1$ there exists a 
 continuous single-valued map $f^{(l)}:\mathbb{R}^d\times\mathcal{S}\times U\rightarrow\mathbb{R}^k$ where $U$ denotes the closed unit ball in 
 $\mathbb{R}^k$, such that,
 \begin{itemize}
  \item [(i)] for every $(x,y,s)\in\mathbb{R}^d\times\mathcal{S}$, $F^{(l)}(x,y,s)=f^{(l)}(x,y,s,U)$ where $f^{(l)}(x,y,s,U)=
  \left\{f^{(l)}(x,y,s,u):u\in U\right\}$,
  \item [(ii)] for $K^{(l)}$ as in Lemma \ref{ctem}$(iii)$, for every $(x,y,s,u)\in\mathbb{R}^d\times\mathcal{S}\times U$, we have that 
  $\left\|f^{(l)}(x,y,s,u)\right\|\leq K^{(l)}(1+\left\|(x,y)\right\|)$ (If the set-valued map $F$ satisfies condition $(c)'$ instead of 
  $(c)$ in the definition of SAM, we have $\left\|f^{(l)}(x,y,s,u)\right\|\leq K^{(l)}(1+\left\|x\right\|+\left\|y\right\|)$).
 \end{itemize}
\end{lemma}

Throughout this paper we shall use $U$ to denote the closed unit ball in $\mathbb{R}^k$ where the dimension $k$ will be made clear by the 
context. 

Combining Lemma \ref{ctem} and Lemma \ref{param} we obtain the approximation theorem stated below.

\begin{theorem}\emph{[approximation]}\label{approx}
 For any SAM $F$, there exists a sequence of continuous functions $\left\{f^{(l)}\right\}_{l\geq1}$ such that for every $l\geq1$, 
 $f^{(l)}:\mathbb{R}^d\times\mathcal{S}\times U\rightarrow\mathbb{R}^k$ is such that,
 \begin{itemize}
  \item [(i)] for every $(x,y,s)\in\mathbb{R}^d\times\mathcal{S}$, $F(x,y,s)\subseteq f^{(l+1)}(x,y,s,U)\subseteq f^{(l)}(x,y,s,U)$ and 
  $f^{(l)}(x,y,s,U)$ is a convex and compact subset of $\mathbb{R}^k$, 
  \item [(ii)] for $K^{(l)}$ as in Lemma \ref{ctem}$(iii)$, for every $(x,y,s,u)\in\mathbb{R}^d\times\mathcal{S}\times U$, we have that 
  $\left\|f^{(l)}(x,y,s,u)\right\|\leq K^{(l)}(1+\left\|(x,y)\right\|)$ (If the set-valued map $F$ satisfies condition $(c)'$ instead of 
  $(c)$ in the definition of SAM, we have $\left\|f^{(l)}(x,y,s,u)\right\|\leq K^{(l)}(1+\left\|x\right\|+\left\|y\right\|)$).
 \end{itemize}
Furthermore, for every $(x,y,s)\in\mathbb{R}^d\times\mathcal{S}$, $F(x,y,s)=\cap_{l\geq1}f^{(l)}(x,y,s,U)$. 
\end{theorem}

%% file: measurable_set_valued_map_and_integration.tex
\subsection{Measurable set-valued maps and integration}
\label{msmai}
In this section we shall review concepts of measurability and integration of set-valued maps. These concepts will be needed to define the 
limiting differential inclusion which the recursion studied later in this paper is expected to track. 

Let $(\mathcal{W},\mathscr{F}_{\mathcal{W}})$ denote a measurable space and $F:\mathcal{W}\rightarrow\left\{\text{subsets of }\mathbb{R}^k\right\}$
be a set-valued map such that, for every $w\in\mathcal{W}$, $F(w)$ is a non-empty closed subset of $\mathbb{R}^k$. Throughout this 
subsection $F$ refers to the set-valued map as defined above.

\begin{definition}
\emph{[measurable set-valued map]} A set-valued map $F$ is measurable if for every $C\subseteq \mathbb{R}^k$, closed,
\begin{equation*}
 F^{-1}(C):=\left\{w\in\mathcal{W}:F(w)\cap C\neq\emptyset\right\}\in \mathscr{F}_{\mathcal{W}}.
\end{equation*}
\end{definition}
We refer the reader to \cite[Thm.~1.2.3]{shoumei} for other notions of measurability and their relation to the definition above.

\begin{definition}
 \emph{[measurable selection]} A function $f:\mathcal{W}\rightarrow\mathbb{R}^k$ is a measurable selection of a set-valued map $F$ if, 
 $f$ is measurable and for every $w\in\mathcal{W}$, $f(w)\in F(w)$.
\end{definition}
For a set-valued map $F$ let, $\mathscr{S}(F)$ denote the set of all measurable selections. The next lemma summarizes some standard
results about measurable set-valued maps and their measurable selections.

\begin{lemma}\label{msel}
 For any measurable set-valued map $F$,
 \begin{itemize}
  \item [(i)] $\mathscr{S}(F)\neq \emptyset$.
  \item [(ii)] \it{(Castaing representation) }\rm there exists $\left\{f_n\right\}_{n\geq1}\subseteq\mathscr{S}(F)$ such that, for every 
  $w\in\mathcal{W}$, $F(w)=cl(\left\{f_n(w)\right\}_{n\geq1})$, where $cl(\cdot)$ denotes the closure of a set.
 \end{itemize}
\end{lemma}
We refer the reader to \cite[Thm.~1.2.6]{shoumei} and \cite[Thm.~1.2.7]{shoumei} for the proofs of Lemma \ref{msel}$(i)$ and $(ii)$ 
respectively.

\begin{definition}
 \emph{[$\mu$-integrable set-valued map]}
 Let $\mu$ be a probability measure on $(\mathcal{W},\mathscr{F}_{\mathcal{W}})$. A measurable set-valued map $F$ is said to be 
 $\mu$-integrable if, there exists $f\in\mathscr{S}(F)$ which is $\mu$-integrable.
\end{definition}

\begin{definition}
 \emph{[Aumann's integral]} Let $\mu$ be a probability measure on $(\mathcal{W},\mathscr{F}_{\mathcal{W}})$. The integral of a 
 $\mu$-integrable set-valued map $F$ is defined as,
 \begin{equation*}
  \int_{\mathcal{W}} F(w)\mu(dw):=\left\{\int_{\mathcal{W}} f(w)\mu(dw):\ f\in \mathscr{S}(F),\ f\ is\ \mu-integrable\right\}.
 \end{equation*}
\end{definition}

The next lemma states a useful result on the properties of the integral of a set-valued map which is convex and compact set valued.
\begin{lemma}\label{cldint}
 Let $\mu$ be a probability measure on $(\mathcal{W},\mathscr{F}_{\mathcal{W}})$ and $F$ a $\mu$-integrable set-valued map such that, 
 for every $w\in\mathcal{W}$, $F(w)$ is convex and compact. Then, $\int_{\mathcal{W}} F(w)\mu(dw)$ is a convex and closed subset of $\mathbb{R}^k$.
\end{lemma}
For a proof of the above lemma we refer the reader to \cite[Thm.~2.2.2]{shoumei}. 

Now we shall briefly investigate the measurability properties of a SAM. First we shall define slices of a SAM, $F$, for each 
$(x,y)\in\mathbb{R}^d$ and for each $y\in\mathbb{R}^{d_2}$. As shown in Lemma \ref{ctem}, when $F$ is a SAM there exists 
$\left\{F^{(l)}\right\}_{l\geq1}$ a sequence of continuous set-valued maps which approximate $F$ and for every $l\geq1$, the set-valued map 
$F^{(l)}$ can be parametrized with single-valued maps $f^{(l)}$ as in Lemma \ref{param}. We shall define similar slices of $F^{(l)}$ and 
$f^{(l)}$ as well.
 
\begin{definition}\label{slices} Let $F:\mathbb{R}^d\times\mathcal{S}\rightarrow\left\{\text{subsets of }\mathbb{R}^k\right\}$ be a SAM. Let  
$\left\{F^{(l)}\right\}_{l\geq1}$ and $\left\{f^{(l)}\right\}_{l\geq1}$ be as in Lemma \ref{ctem} and Lemma \ref{param} respectively.
 \begin{itemize}
  \item [(i)]For every $(x,y)\in\mathbb{R}^d$, define $F_{(x,y)}:\mathcal{S}\rightarrow\left\{\text{subsets of }\mathbb{R}^k\right\}$ such that for every 
  $s\in \mathcal{S}$, $F_{(x,y)}(s):=F(x,y,s)$.
  \item [(ii)]For every $l\geq1$, for every $(x,y)\in\mathbb{R}^d$, define $F^{(l)}_{(x,y)}:\mathcal{S}\rightarrow\left\{\text{subsets of }\mathbb{R}^k\right\}$ 
  such that for every $s\in\mathcal{S}$, $F^{(l)}_{(x,y)}(s):=F^{(l)}(x,y,s)$.
  \item [(iii)]For every $l\geq1$, for every $(x,y)\in\mathbb{R}^d$, define $f^{(l)}_{(x,y)}:\mathcal{S}\times U\rightarrow\mathbb{R}^k$ such that for every
  $(s,u)\in\mathcal{S}\times U$, $f^{(l)}_{(x,y)}(s,u):=f^{(l)}(x,y,s,u)$.
  \item [(iv)] For every $y\in\mathbb{R}^{d_2}$, define $F_{y}:\mathbb{R}^{d_1}\times\mathcal{S}\rightarrow\left
  \{\text{subsets of }\mathbb{R}^k\right\}$ such that for every $(x,s)\in\mathbb{R}^{d_1}\times\mathcal{S}$, $F_{y}(x,s):=F(x,y,s)$.
  \item [(v)] For every $l\geq1$, for every $y\in\mathbb{R}^{d_2}$, define $F^{(l)}_{y}:\mathbb{R}^{d_1}\times\mathcal{S}\rightarrow\left
  \{\text{subsets of }\mathbb{R}^k\right\}$ such that for every $(x,s)\in\mathbb{R}^{d_1}\times\mathcal{S}$, $F^{(l)}_{y}(x,s):=F^{(l)}(x,y,s)$.
  \item [(vi)]For every $l\geq1$, for every $y\in\mathbb{R}^{d_2}$, define $f^{(l)}_{y}:\mathbb{R}^{d_1}\times\mathcal{S}\times U\rightarrow
  \mathbb{R}^k$ such that for every $(x,s,u)\in\mathbb{R}^{d_1}\times\mathcal{S}\times U$, $f^{(l)}_{y}(x,s,u):=f^{(l)}(x,y,s,u)$.
 \end{itemize}
\end{definition}

The next two lemmas summarize properties that the slices inherit from the maps $F,\ F^{(l)}$ and $f^{(l)}$. Let $\mathscr{B}(\mathcal{S})$ denote the Borel 
sigma algebra associated with the metric space $(\mathcal{S},d_\mathcal{S})$.

\begin{lemma}\label{msble1}
Let $F:\mathbb{R}^d\times\mathcal{S}\rightarrow\left\{\text{subsets of }\mathbb{R}^k\right\}$ be a SAM. Let  
$\left\{F^{(l)}\right\}_{l\geq1}$ and $\left\{f^{(l)}\right\}_{l\geq1}$ be as in Lemma \ref{ctem} and Lemma \ref{param} respectively. For every 
$(x,y)\in\mathbb{R}^d$, let $F_{(x,y)},\ F_{(x,y)}^{(l)}$ and $f^{(l)}_{(x,y)}$ denote the slices as in Definition \ref{slices}. Then for 
every $(x,y)\in\mathbb{R}^d$,
\begin{itemize}
 \item [(i)] $F_{(x,y)}$ is a measurable set-valued map and for every $s\in\mathcal{S}$, $F_{(x,y)}(s)$ is a convex and compact subset of $\mathbb{R}^k$.
 Further, there exists $C_{(x,y)}=K(1+\left\|(x,y)\right\|)>0$ such that for every $s\in\mathcal{S}$, $\sup_{z\in F_{(x,y)}(s)}\left\|z\right\|\leq C_{(x,y)}$.
 (If $F$ satisfies condition $(c)'$ instead of condition $(c)$ in the definition of SAM, we have $C_{(x,y)}=K(1+\left\|x\right\|+\left\|y\right\|)$).
 \item [(ii)] for every $l\geq1$, $F^{(l)}_{(x,y)}$ is a measurable set-valued map and for every $s\in\mathcal{S}$, $F^{(l)}_{(x,y)}(s)$ is a convex and 
 compact subset of $\mathbb{R}^k$. Further, there exists $C^{(l)}_{(x,y)}=K^{(l)}(1+\left\|(x,y)\right\|)>0$ such that for every $s\in\mathcal{S}$, 
 $\sup_{z\in F^{(l)}_{(x,y)}(s)}\left\|z\right\|\leq C^{(l)}_{(x,y)}$. (If $F$ satisfies condition $(c)'$ instead of condition $(c)$ in the definition of SAM,
 we have $C^{(l)}_{(x,y)}=K^{(l)}(1+\left\|x\right\|+\left\|y\right\|)$).
 \item [(iii)] for any probability measure $\mu$ on $(\mathcal{S},\mathscr{B}(S))$, every measurable selection of $F_{(x,y)}$ is $\mu$-integrable and 
 hence $F_{(x,y)}$ is $\mu$-integrable.
 \item [(iv)] for every $l\geq1$, for any probability measure $\mu$ on $(\mathcal{S},\mathscr{B}(S))$, every measurable selection of $F^{(l)}_{(x,y)}$ 
 is $\mu$-integrable and hence $F^{(l)}_{(x,y)}$ is $\mu$-integrable.
 \item [(v)] for every $l\geq1$, $f^{(l)}_{(x,y)}$ is continuous and for every $s\in\mathcal{S}$, $f^{(l)}_{(x,y)}(s,U)=F_{(x,y)}^{(l)}(s)$ and 
 $\sup_{u\in U}\left\|f^{(l)}_{(x,y)}(s,u)\right\|\leq C^{(l)}_{(x,y)}$ where $C^{(l)}_{(x,y)}$ is as in part $(ii)$ of this lemma. 
\end{itemize}
\end{lemma}

The proof of the above lemma is similar to that of \cite[Lemma~4.1]{vin1t} and we shall provide a brief outline here. Fix $(x,y)\in\mathbb{R}^d$. In order to 
show that $F_{(x,y)}$ is measurable, one needs to establish that $F_{(x,y)}^{-1}(C)\in\mathscr{B}(\mathcal{S})$ for any $C\subseteq\mathbb{R}^k$ closed. 
Using the closed graph property of $F$ one can show that $F_{(x,y)}^{-1}(C)$ is closed subset of $\mathcal{S}$ and hence is in 
$\mathscr{B}(\mathcal{S})$. The bound $C_{(x,y)}$ and the claim that $F_{(x,y)}(s)$ is convex and compact for every $s\in\mathcal{S}$ follows from 
conditions $(c)$ (or $(c)'$) and $(a)$ in the definition of SAM respectively. Since all measurable selections of $F_{(x,y)}$ are bounded, they are 
$\mu$-integrable for any probability measure $\mu$ on $(\mathcal{S},\mathscr{B}(\mathcal{S}))$. The arguments are exactly same for the claims 
associated with the slices of approximating maps $F^{(l)}$, for every $l\geq1$. Finally the part $(v)$ of the above lemma follows from the 
properties of maps $f^{(l)}$ stated in Lemma \ref{param}.

Let $\mu$ be a probability measure on $(\mathbb{R}^{d_1}\times\mathcal{S},\mathscr{B}(\mathbb{R}^{d_1}\times\mathcal{S}))$ where 
$\mathscr{B}(\mathbb{R}^{d_1}\times\mathcal{S})$ denotes the Borel sigma algebra on metric space $\mathbb{R}^{d_1}\times\mathcal{S}$ with metric 
$\max\left\{\left\|x-x'\right\|,d_{\mathcal{S}}(s,s')\right\}$ for every $(x,s),\ (x',s')\in\mathbb{R}^{d_1}\times\mathcal{S}$ (in fact 
$\mathscr{B}(\mathbb{R}^{d_1}\times\mathcal{S})$ is the same as the product sigma algebra $\mathscr{B}(\mathbb{R}^{d_1})\otimes\mathcal{S}$). The 
\it{support }\rm of the measure $\mu$ denoted by $\mathrm{supp}(\mu)$ is defined as a closed subset of $\mathbb{R}^{d_1}\times\mathcal{S}$ such that,
\begin{itemize}
 \item [(1)] $\mu(\mathrm{supp}(\mu))=1$,
 \item [(2)] for any other closed set $A\subseteq\mathbb{R}^{d_1}\times\mathcal{S}$ such that $\mu(A)=1$, we have $\mathrm{supp}(\mu)\subseteq A$.
\end{itemize}
For any probability measure $\mu$ on $\mathbb{R}^{d_1}\times\mathcal{S}$ such a set always exists and is unique (see \cite[ch.~2, Thm.~2.1]{parth}). 

\begin{lemma}\label{msble2}Let $F:\mathbb{R}^d\times\mathcal{S}\rightarrow\left\{\text{subsets of }\mathbb{R}^k\right\}$ be a SAM satisfying 
condition $(c)'$ instead of condition $(c)$ in the definition of SAM. Let $\left\{F^{(l)}\right\}_{l\geq1}$ and $\left\{f^{(l)}\right\}_{l\geq1}$
be as in Lemma \ref{ctem} and Lemma \ref{param} respectively. For every $(x,y)\in\mathbb{R}^d$, let $F_{y},\ F_{y}^{(l)}$ and $f^{(l)}_{y}$ 
denote the slices as in Definition \ref{slices}. Then, for every $y\in\mathbb{R}^{d_2}$,
\begin{itemize}
 \item [(i)]$F_{y}$ is a measurable set-valued map and for every $(x,s)\in\mathbb{R}^{d_1}\times\mathcal{S}$, $F_{y}(x,s)$ is a convex 
 and compact subset of $\mathbb{R}^k$. Further for every $(x,s)\in\mathbb{R}^{d_1}\times\mathcal{S}$, $\sup_{z\in F_{y}(x,s)}
 \left\|z\right\|\leq K_{y}(1+\left\|x\right\|)$ where $K_{y}:=\max\left\{K,K\left\|y\right\|\right\}$ and $K$ is as in condition $(c)'$ 
 in the definition of SAM.
 \item [(ii)] for every $l\geq1$, $F^{(l)}_{y}$ is a measurable set-valued map and for every $(x,s)\in\mathbb{R}^{d_1}\times\mathcal{S}$, 
 $F^{(l)}_{y}(x,s)$ is a convex and compact subset of $\mathbb{R}^k$. Further for every $(x,s)\in\mathbb{R}^{d_1}\times\mathcal{S}$, 
 $\sup_{z\in F_{y}(x,s)}\left\|z\right\|\leq K^{(l)}_{y}(1+\left\|x\right\|)$ where $K^{(l)}_{y}:=\max\left\{K^{(l)},K^{(l)}
 \left\|y\right\|\right\}$ and $K^{(l)}$ is as in Lemma \ref{ctem}$(iii)$.
 \item [(iii)] for every probability measure $\mu$ on $(\mathbb{R}^{d_1}\times\mathcal{S},\mathscr{B}(\mathbb{R}^{d_1}\times\mathcal{S}))$ such 
 that $\mathrm{supp}(\mu)$ is a compact subset of $\mathbb{R}^{d_1}\times\mathcal{S}$, every measurable selection of $F_{y}$ is $\mu$-integrable and 
 hence $F_{y}$ is $\mu$-integrable.
 \item [(iv)] for every $l\geq1$, for every probability measure $\mu$ on $(\mathbb{R}^{d_1}\times\mathcal{S},\mathscr{B}(\mathbb{R}^{d_1}\times
 \mathcal{S}))$ such that $\mathrm{supp}(\mu)$ is a compact subset of $\mathbb{R}^{d_1}\times\mathcal{S}$, every measurable selection of $F^{(l)}_{y}$ 
 is $\mu$-integrable and hence $F^{(l)}_{y}$ is $\mu$-integrable.
 \item [(v)] for every $l\geq1$, $f^{(l)}_{y}$ is continuous and for every $(x,s)\in\mathbb{R}^{d_1}\times\mathcal{S}$, 
 $f^{(l)}_{y}(x,s,U)=F_{y}^{(l)}(x,s)$ and $\sup_{u\in U}\left\|f^{(l)}_{y}(x,s,u)\right\|\leq K^{(l)}_{y}(1+\left\|x\right\|)$ 
 where $K^{(l)}_{y}$ is as in part $(ii)$ of this lemma. 
\end{itemize}
\end{lemma}
The proof of parts $(i),\ (ii)$ and $(v)$ of the above lemma are similar to the corresponding in Lemma \ref{msble1}. We shall provide a proof of 
part $(iii)$ and the proof of part $(iv)$ is exactly the same.
\begin{proof}Fix $y\in\mathbb{R}^{d_2}$.
 \begin{itemize}
  \item [(iii)] Consider $f \in \mathscr{S}(F_{y})$. By part $(i)$ of this lemma we have that $\left\|f(x,s)\right\|\leq K_{y}
  (1+\left\|x\right\|)$. Since $\mathrm{supp}(\mu)$ is a compact subset of $\mathbb{R}^{d_1}\times\mathcal{S}$, there exists $M>0$ such that 
  for every $x\in\mathbb{R}^{d_1}$ for which there exists $s\in\mathcal{S}$ satisfying $(x,s)\in \mathrm{supp}(\mu)$, we have $\left\|x\right\|\leq M$.
  Hence $\left\|\int_{\mathbb{R}^{d_1}\times\mathcal{S}}f(x,s)\mu(dx,ds)\right\|=\left\|\int_{\mathrm{supp}(\mu)}f(x,s)
  \mu(dx,ds)\right\|\leq\int_{\mathrm{supp}(\mu)}\left\|f(x,s)\right\|\mu(dx,ds)\leq\int_{\mathrm{supp}(\mu)}K_{y}(1+\left\|x\right\|)\mu(dx,ds)\leq 
  K_{y}(1+M)$. Therefore every measurable selection $f\in\mathscr{S}(F_{y})$ is $\mu$-integrable and hence $F_{y}$ is $\mu$-integrable.\qed
 \end{itemize}
\end{proof}

By Lemma \ref{msble1}$(iv)$ and $(v)$ we know that $F^{(l)}_{(x,y)}$ is a $\mu$-integrable set-valued map for any probability measure $\mu$ on 
$(\mathcal{S},\mathscr{B}(\mathcal{S}))$ and $f^{(l)}_{(x,y)}$ is a continuous parametrization of $F^{(l)}_{(x,y)}$ for every $l\geq1$ and for every 
$(x,y)\in\mathbb{R}^{d}$. Similarly, by Lemma \ref{msble2}$(iv)$ and $(v)$ we know that $F^{(l)}_{y}$ is $\mu$-integrable for any probability 
measure $\mu$ on $(\mathbb{R}^{d_1}\times\mathcal{S},\mathscr{B}(\mathbb{R}^{d_1}\times\mathcal{S}))$ with compact support and $f^{(l)}_{y}$ is 
a continuous parametrization of $F^{(l)}_{y}$ for every $l\geq1$ and for every $y\in\mathbb{R}^{d_2}$. A natural question to ask is about the 
relation between integral of map $F^{(l)}_{(x,y)}$ (or $F^{(l)}_{y}$) and the integral of its parametrization $f^{(l)}_{(x,y)}$ (or $f^{(l)}_{y}$). The 
next lemma answers this question. Before stating the lemma we introduce the following notation which will be used throughout this paper.

Let $\mathcal{P}(\cdots)$ denote the space of probability measures on a Polish space $\lq\cdots\rq$ with the Prohorov topology (also known 
as the topology of convergence in distribution,see \cite[ch.~2]{borkarap} for details). For any probability measure $\nu\in 
\mathcal{P}(\mathcal{S}\times U)$, let $\nu_{\mathcal{S}}\in\mathcal{P}(\mathcal{S})$ denote the image of measure $\nu$ under the projection 
$\mathcal{S}\times U\rightarrow\mathcal{S}$ (that is for any $A\in\mathscr{B}(\mathcal{S})$, $\nu_{\mathcal{S}}(A)=\int_{A\times U}\mu(ds,du)$).
Similarly, for any probability measure $\nu\in\mathcal{P}(\mathbb{R}^{d_1}\times\mathcal{S}\times U)$, let 
$\nu_{\mathbb{R}^{d_1}\times\mathcal{S}},\ \nu_{\mathcal{S}}$ and $\nu_{\mathbb{R}^{d_1}}$ belonging to $\mathcal{P}(\mathbb{R}^{d_1}
\times\mathcal{S}),\ \mathcal{P}(\mathcal{S})$ and $\mathcal{P}(\mathbb{R}^{d_1})$ respectively denote the image of measure $\nu$ under the 
projections $\mathbb{R}^{d_1}\times\mathcal{S}\times U\rightarrow\mathbb{R}^{d_1}\times\mathcal{S},\ \mathbb{R}^{d_1}\times\mathcal{S}\times 
U\rightarrow\mathcal{S}$ and $\mathbb{R}^{d_1}\times\mathcal{S}\times U\rightarrow\mathbb{R}^{d_1}$ respectively.

\begin{lemma}\label{chint}
Let $F:\mathbb{R}^d\times\mathcal{S}\rightarrow\left\{\text{subsets of }\mathbb{R}^k\right\}$ be a SAM. Let $\left\{F^{(l)}\right\}_{l\geq1}$ 
and $\left\{f^{(l)}\right\}_{l\geq1}$ be as in Lemma \ref{ctem} and Lemma \ref{param} respectively. For every $l\geq1$, for every 
$(x,y)\in\mathbb{R}^d$, let $F_{(x,y)}^{(l)},\ f^{(l)}_{(x,y)}$ and for every $y\in\mathbb{R}^{d_2}$ let $F^{(l)}_{y},\ f^{(l)}_{y}$ denote the 
slices as in Definition \ref{slices}.
\begin{itemize}
 \item [(i)] For every $l\geq1$, for every $(x,y)\in\mathbb{R}^d$, for any probability measure $\mu\in\mathcal{P}(\mathcal{S})$,
 \begin{equation*}
  \int_{\mathcal{S}}F^{(l)}_{(x,y)}(s)\mu(ds)=\left\{\int_{\mathcal{S}\times U}f^{(l)}_{(x,y)}(s,u)\nu(ds,du):\nu\in\mathcal{P}(\mathcal{S}\times U),\  
  \nu_{\mathcal{S}}=\mu\right\}.
 \end{equation*}
 \item [(ii)] Suppose $F$ satisfies condition $(c)'$ instead of condition $(c)$ in the definition of SAM. Then for every $l\geq1$, for every $y\in\mathbb{R}^{d_2}$, 
 for any probability measure $\mu\in\mathcal{P}(\mathbb{R}^{d_1}\times\mathcal{S})$ with compact support,
 \begin{small}
 \begin{equation*}
  \int_{\mathbb{R}^{d_1}\times\mathcal{S}}F^{(l)}_{y}(dx,ds)\mu(dx,ds)=\left\{\int_{\mathbb{R}^{d_1}\times\mathcal{S}\times U}f^{(l)}_{y}
  (x,s,u)\nu(dx,ds,du):\nu\in\mathcal{P}(\mathbb{R}^{d_1}\times\mathcal{S}\times U),\ \nu_{\mathbb{R}^{d_1}\times\mathcal{S}}=\mu\right\}.
 \end{equation*}
 \end{small}
\end{itemize}
\end{lemma}

\begin{remark}
For any $\nu\in\mathcal{P}(\mathbb{R}^{d_1}\times\mathcal{S}\times U)$ with $\nu_{\mathbb{R}^{d_1}\times\mathcal{S}}=\mu$, the support 
of the measure $\nu$, is contained in $\mathrm{supp}(\mu)\times U$, since by \cite[ch.~3, Cor.3.1.2]{borkarap} there exists a $\mu$ a.s. unique measurable 
map $q:\mathbb{R}^{d_1}\times\mathcal{S}\rightarrow\mathcal{P}(U)$ such that, $\nu(dx,ds,du)=q(x,s,du)\mu(dx,ds)$ and 
$1=\nu(\mathbb{R}^{d_1}\times\mathcal{S}\times U)=\int_{\mathbb{R}^{d_1}\times\mathcal{S}\times U}\nu(dx,ds,du)=\int_{\mathbb{R}^{d_1}\times\
\mathcal{S}}\left[\int_{U}q(x,s,du)\right]\mu(dx,ds)=\int_{\mathrm{supp}(\mu)}\left[\int_{U}q(x,s,du)\right]\mu(dx,ds)=\nu(\mathrm{supp}(\mu)\times U)$.
Therefore when $\mathrm{supp}(\mu)$ is a compact set the support of measure $\nu$ is also compact and by Lemma \ref{msble2}$(v)$ it is easy to deduce that 
for all measures $\nu\in\mathcal{P}(\mathbb{R}^{d_1}\times\mathcal{S}\times U)$ with compact support, for all $y\in\mathbb{R}^{d_2}$, 
$f^{(l)}_{y}$ is $\nu$-integrable for all $l\geq1$. 
\end{remark}
The proof of part $(i)$ of the above lemma is exactly same as \cite[Lemma~4.2]{vin1t}. The proof of part $(ii)$ is similar but with minor 
technical modifications and is presented below. 
\begin{proof}
\begin{itemize}
 \item [(ii)] Fix $y\in\mathbb{R}^{d_2}$, $l\geq1$ and $\mu\in\mathcal{P}(\mathbb{R}^{d_1}\times\mathcal{S})$ with compact support.
 
 Consider $z\in\int_{\mathbb{R}^{d_1}\times\mathcal{S}}F^{(l)}_{y}(x,s)\mu(dx,ds)$. Then there exists $f\in\mathscr{S}(F^{(l)}_{y})$ such 
 that $z=\int_{\mathbb{R}^{d_1}\times\mathcal{S}}f(x,s)\mu(dx,ds)$. Let $G:\mathbb{R}^{d_1}\times\mathcal{S}\rightarrow\left\{\text{subsets 
 of }U\right\}$ be such that for every $(x,s)\in\mathbb{R}^{d_1}\times\mathcal{S}$, \begin{small}$G(x,s)=\left\{u\in U: f(x,s)=f^{(l)}_{y}(x,s,u)
 \right\}$\end{small}. By the fact that $f^{(l)}_{y}(x,s,U)=F^{(l)}_{y}(x,s)$ and since $f(x,s)\in F^{(l)}_{y}(x,s)$ for every $(x,s)
 \in\mathbb{R}^{d_1}\times\mathcal{S}$ we have that $G(x,s)$ is nonempty. By the continuity of $f^{(l)}_{y}(x,s,\cdot)$ we have that 
 $G(x,s)$ is closed for every $(x,s)\in\mathbb{R}^{d_1}\times\mathcal{S}$. For any $C\subseteq U$ closed, $G^{-1}(C)\in\mathscr{B}
 (\mathbb{R}^{d_1}\times\mathcal{S})$ (for a proof see \cite[Appendix~B]{vin1t}) and hence $G$ is measurable. Since $G$ is measurable, by Lemma \ref{msel}$(i)$ 
 we have that $\mathscr{S}(G)\neq\emptyset$. Let $g\in\mathscr{S}(G)$ and let $\hat{g}:\mathbb{R}^{d_1}\times\mathcal{S}\rightarrow
 \mathbb{R}^{d_1}\times\mathcal{S}\times U$ be such that for every $(x,s)\in\mathbb{R}^{d_1}\times\mathcal{S}$, 
 $\hat{g}(x,s):=(x,s,g(x,s))$. Let $\nu=\mu\hat{g}^{-1}$ (push-forward measure). Clearly $\nu_{\mathbb{R}^{d_1}\times U}=\mu$ and 
 $\int_{\mathbb{R}^{d_1}\times\mathcal{S}\times U}f^{(l)}_{y}(x,s,u)\nu(dx,ds,du)=\int_{\mathbb{R}^{d_1}\times\mathcal{S}\times U}f^{(l)}
 _{y}(x,s,u)\mu\hat{g}^{-1}(dx,ds,du)=\int_{\mathbb{R}^{d_1}\times\mathcal{S}}f^{(l)}_{y}(x,s,g(x,s))\mu(dx,ds)=
 \int_{\mathbb{R}^{d_1}\times\mathcal{S}}f(x,s)\mu(dx,ds)=z$. Therefore L.H.S. is contained in R.H.S.

 Let $\nu\in\mathcal{P}(\mathbb{R}^{d_1}\times\mathcal{S}\times U)$ with $\nu_{\mathbb{R}^{d_1}\times\mathcal{S}}=\mu$. By 
 \cite[Cor.~3.1.2]{borkarap}, there exists a $\mu$ a.s. unique measurable map $q:\mathbb{R}^{d_1}\times\mathcal{S}\rightarrow\mathcal{P}(U)$ such 
 that $\nu(dx,ds,du)=q(x,s,du)\mu(dx,ds)$. Since $\mu$ has compact support, $\nu$ has compact support and hence $f^{(l)}_{y}$ is $\nu$-
 integrable (see remark following Lemma \ref{chint}). Therefore $\int_{\mathbb{R}^{d_1}\times\mathcal{S}\times U}f^{(l)}_{y}(x,s,u)\nu(dx,
 ds,du)=\int_{\mathbb{R}^{d_1}\times\mathcal{S}}\left[\int_{U}f^{(l)}_{y}(x,s,u)q(x,s,du)\right]\mu(dx,ds)$. By Lemma \ref{msble2}$(v)$ 
 we know that for every $(x,s)\in\mathbb{R}^{d_1}\times\mathcal{S}$, $f^{(l)}_{y}(x,s,U)=F^{(l)}_{y}(x,s)$ and hence $f^{(l)}(x,s,U)$ 
 is convex and compact subset of $\mathbb{R}^k$. Therefore for every $(x,s)\in\mathbb{R}^{d_1}\times\mathcal{S}$, $\int_{U}f^{(l)}_{y}(x,
 s,u)q(x,s,du)\in F^{(l)}_{y}(x,s)$. Let $f:\mathbb{R}^{d_1}\times\mathcal{S}\rightarrow\mathbb{R}^k$ be such that for every 
 $(x,s)\in\mathbb{R}^{d_1}\times\mathcal{S}$, $f(x,s):=\int_{U}f^{(l)}(x,s,u)q(x,s,du)$. Then clearly, $f$ is measurable and 
 $f\in\mathscr{S}(F^{(l)}_{y})$. Therefore, $\int_{\mathbb{R}^{d_1}\times\mathcal{S}\times U}f^{(l)}_{y}(x,s,u)\nu(dx,ds,du)=\\
 \int_{\mathbb{R}^{d_1}\times\mathcal{S}}\left[\int_{U}f^{(l)}_{y}(x,s,u)q(x,s,du)\right]\mu(dx,ds)=\int_{\mathbb{R}^{d_1}\times
 \mathcal{S}}f(x,s)\mu(dx,ds)\in\int_{\mathbb{R}^{d_1}\times\mathcal{S}}F^{(l)}_{y}(x,s)\mu(dx,ds)$. The above gives us that R.H.S. is contained in L.H.S\qed
\end{itemize}
\end{proof}

%% file: differential_inclusions_and_limit_sets.tex
\subsection{Differential inclusions and their limit sets}
\label{diatls}
In this section we shall review results from the theory of differential inclusions and state definitions of limit sets associated with such 
dynamical systems which are used later in the paper. Most of the results in this section are taken from \cite{benaim1}.

First we shall define a set-valued map whose associated differential inclusion (DI) is known to admit at-least one solution through every initial 
condition. Such set-valued maps are called Marchaud maps and the definition of such a map is stated below.

\begin{definition}\emph{[Marchaud map]}\label{marmap}
$F:\mathbb{R}^k\rightarrow\left\{\text{subsets of }\mathbb{R}^k\right\}$ is a Marchaud map if, 
\begin{itemize}
 \item [(i)] for every $z\in\mathbb{R}^k$, $F(z)$ is a convex and compact subset of $\mathbb{R}^k$,
 \item [(ii)] there exists $K>0$ such that for every $z\in\mathbb{R}^k$, $\sup_{z'\in F(z)}\left\|z'\right\|\leq K(1+\left\|z\right\|)$,
 \item [(iii)] for every $z\in\mathbb{R}^k$, for every $\mathbb{R}^k$-valued sequence, $\left\{z_n\right\}_{n\geq0}$ converging to $z\in\mathbb{R}^k$, 
 for every sequence $\left\{z'_n\in F(z_n)\right\}_{n\geq0}$ converging to $z'\in\mathbb{R}^k$, we have that $z'\in F(z)$.
\end{itemize}
\end{definition}

Let $F$ be a Marchaud map. Then the DI associated with the map $F$ is given by,
\begin{equation}\label{diex}
 \frac{dz}{dt}\in F(z).
\end{equation}
Since $F$ is a Marchaud map, it is known that the DI \eqref{diex}, admits at-least one solution through every initial condition (see 
\cite[sec.~1.2]{benaim1}). By a solution of DI \eqref{diex} with initial condition $z\in\mathbb{R}^k$, we mean a function 
$\bm{\mathrm{z}}:\mathbb{R}\rightarrow\mathbb{R}^k$ such that $\bm{\mathrm{z}}(\cdot)$ is absolutely continuous, $\bm{\mathrm{z}}(0)=z$ and for $a.e.$ 
$t\in\mathbb{R}$, $\frac{d\bm{\mathrm{z}}(t)}{dt}\in F(\bm{\mathrm{z}}(t))$.

Now we shall recall the notions of flow, invariant sets, attracting sets, attractors, basin of attraction and internally chain transitive sets.
All of these notions are taken from \cite{benaim1}.

The \bf{flow }\rm of DI \eqref{diex} is given by the set-valued map $\Phi:\mathbb{R}^k\times\mathbb{R}\rightarrow\left\{\text{subsets of }
\mathbb{R}^k\right\}$ such that for every $(z,t)\in\mathbb{R}^k\times\mathbb{R}$,
\begin{equation*}
 \Phi(z,t):=\left\{\bm{\mathrm{z}}(t): \bm{\mathrm{z}}\ is\ a\ solution\ of\ DI\ \eqref{diex}\ with\ \bm{\mathrm{z}}(0)=z\right\}.
\end{equation*}
For any set $A\subseteq\mathbb{R}^k$, let $\Phi(A,t):=\cup_{z\in A}\Phi(z,t)$.

A closed set $A\subseteq\mathbb{R}^k$ is \bf{invariant }\rm for the flow $\Phi$ of DI \eqref{diex} if for every $z\in A$, there exists a solution 
$\bf{z}\rm(\cdot)$ of DI \eqref{diex} such that, $\bm{\mathrm{z}}(0)=z$ and for every $t\in\mathbb{R}$, $\bm{\mathrm{z}}(t)\in A$.

A compact set $A\subseteq\mathbb{R}^k$ is an \bf{attracting }\rm set for the flow $\Phi$ of DI \eqref{diex}, if there exists an open 
neighborhood of $A$, say $\mathcal{O}$, with the property that for every $\epsilon>0$, there exists $T>0$ (depending only on $\epsilon>0$) such 
that for every $t\geq T$, $\Phi(\mathcal{O},t)\subseteq N^{\epsilon}(A)$, where $N^{\epsilon}(A)$ stands for the $\epsilon$-neighborhood of $A$.

A compact set $A\subseteq\mathbb{R}^k$ is an \bf{attractor }\rm for the flow $\Phi$ of DI \eqref{diex}, if $A$ is an attracting set and is 
invariant for the flow $\Phi$ of DI \eqref{diex}.

For any $z\in\mathbb{R}^k$, $\omega_{\Phi}(z):=\cap_{t\geq0}\overline{\Phi(z,[t,\infty))}$ where $\Phi(z,[t,\infty)):=\cup_{q\geq t}\Phi(z,q)$.
For any set $A\subseteq\mathbb{R}^k$, the \bf{basin of attraction }\rm of set $A$ is denoted by $B(A)$ and is defined as,
\begin{equation*}
 B(A):=\left\{z\in\mathbb{R}^k:\omega_{\Phi}(z)\subseteq A\right\}.
\end{equation*}

If $A\subseteq\mathbb{R}^k$ is an attractor whose basin of attraction is the whole of $\mathbb{R}^k$ (i.e. $B(A)=\mathbb{R}^k$) then $A$ is 
called a \bf{global attractor}\rm.

Given a set $A\subseteq\mathbb{R}^k$ and $z,z'\in A$, for any $\epsilon>0$ and $T>0$ there exists an $(\epsilon,T)$ chain from $z$ to $z'$ for 
DI\eqref{diex} if there exists an integer $n\in\mathbb{N}$, solutions $\bf{z}\rm_1,\dots,\bf{z}\rm_n$ to DI \eqref{diex} and real numbers 
$t_1,\dots,t_n$ greater than $T$ such that
\begin{itemize}
 \item for all $i\in\left\{1,\dots,n\right\}$ and for all $q\in[0,t_i]$, $\bm{\mathrm{z}}_i(q)\in A$,
 \item for all $i\in\left\{1,\dots,n\right\}$, $\parallel\bm{\mathrm{z}}_i(t_i)-\bm{\mathrm{z}}_{i+1}(0)\parallel\leq\epsilon$,
 \item $\parallel \bm{\mathrm{z}}_1(0)-z\parallel\leq\epsilon$ and $\parallel\bm{\mathrm{z}}_n(t_n)-z'\parallel\leq\epsilon$.
\end{itemize}
A compact set $A\subseteq\mathbb{R}^d$ is said to be \bf{internally chain transitive }\rm if for every $z,z'\in A$, for every $\epsilon>0$ and 
for every $T>0$, there exists $(\epsilon,T)$ chain from $z$ to $z'$ for DI \eqref{diex}.

Suppose $L\subseteq\mathbb{R}^k$ is an invariant set. Then the flow of DI \eqref{diex} restricted to the invariant set $L$ is a set-valued map, 
$\Phi^{L}:L\times\mathbb{R}\rightarrow\left\{\text{subsets of }\mathbb{R}^k\right\}$ such that for every $(z,t)\in L\times\mathbb{R}$,
\begin{equation}
 \Phi^{L}(z,t):=\left\{\bm{\mathrm{z}}(t):\bm{\mathrm{z}}(\cdot)\ is\ a\ solution\ of\ DI\ \eqref{diex}\ with\ \bm{\mathrm{z}}(0)=z\ and\ for\ every\ t\in\mathbb{R},\ \bm{\mathrm{z}}(t)
 \in L\right\}.
\end{equation}

%% file: single_time_scale_sri_with_mn.tex
\subsection{Single timescale stochastic recursive inclusions with non-additive iterate-dependent Markov noise}
\label{stsri}
In this section we review results from the analysis of single timescale stochastic recursive inclusions with non-additive iterate dependent 
Markov noise. All of the results presented here can be found in \cite{vin1t}. 

Let $(\Omega,\mathscr{F},\mathbb{P})$ be a probability space and $\left\{Z_n\right\}_{n\geq0}$ be a sequence of $\mathbb{R}^{d}$-valued 
random variables satisfying
\begin{equation}\label{str}
 Z_{n+1}-Z_{n}-a(n)M_{n+1}\in a(n)F(Z_n,S_n), 
\end{equation}
where the following assumptions hold:
\begin{itemize}
 \item [S(A1)] the map $F:\mathbb{R}^{d}\times\mathcal{S}\rightarrow\left\{\text{subsets of }\mathbb{R}^d\right\}$ where 
 $(\mathcal{S},d_{\mathcal{S}})$ is a compact metric space is such that,
 \begin{itemize}
 \item [(i)] for every $(z,s)\in\mathbb{R}^{d}\times\mathcal{S}$, $F(z,s)$ is a convex and compact subset of $\mathbb{R}^{d}$,
 \item [(ii)] there exists $K>0$ such that for every $(z,s)\in\mathbb{R}^{d}\times\mathcal{S}$, $\sup_{z'\in F(z,s)}\left\|z'\right\|\leq K(1+
 \left\|z\right\|)$,
 \item [(iii)] for every $z\in\mathbb{R}^d$, for every $\mathbb{R}^{d}\times\mathcal{S}$ valued sequence, say $\left\{(z_n,s_n)\right\}_{n\geq1}$ 
 converging to $(z,s)$ and for any sequence $\left\{z'_n\in F(z_n,s_n)\right\}_{n\geq1}$ converging to $z'$ we have $z'\in F(z,s)$. 
 \end{itemize}
 \item [S(A2)] $\{S_n\}_{n\geq0}$ is a sequence of $\mathcal{S}$-valued measurable functions on $\Omega$ such that for every $n\geq0$, for every 
 $A\in\mathscr{B}(\mathcal{S})$, $\mathbb{P}(S_{n+1}\in A|S_m,Z_m,m\leq n)=\mathbb{P}(S_{n+1}\in A|S_n,Z_n)=\Pi(Z_n,S_n)(A)$ $a.s.$, where 
 $\Pi:\mathbb{R}^{d}\times\mathcal{S}\rightarrow\mathcal{P}(\mathcal{S})$ is continuous.
 \item [S(A3)] $\left\{a(n)\right\}_{n\geq0}$ is a sequence of positive real numbers satisfying,
 \begin{itemize}
 \item [(i)] $a(0)\leq1$ and for every $n\geq0$, $a(n)\geq a(n+1)$,
 \item [(ii)] $\sum_{n=0}^{\infty}a(n)=\infty$ and $\sum_{n=0}^{\infty}(a(n))^{2}<\infty$,
 \end{itemize}
 \item [S(A4)] $\{M_n\}_{n\geq1}$ is a sequence of $\mathbb{R}^d$-valued random variables on $\Omega$ such that for $a.s.(\omega)$, for any $T>0$, 
 $\\\lim_{n\to\infty}\sup_{n\leq k\leq \tau(n,T)}\left\|\sum_{m=n}^{k}a(m)M_{m+1}(\omega)\right\|=0$ where $\tau(n,T):=\min\left\{m>n:\sum_{k=n}^{m-1}
 a(k)\geq T\right\}$.
 \item [S(A5)] $\mathbb{P}(\sup_{n\geq0}\left\|X_n\right\|<\infty)=1$.
\end{itemize}
 
A detailed motivation for each of these assumptions can be found in \cite{vin1t}. We shall briefly explain them and their consequences.

Assumption $S(A1)$ ensures that the set-valued map $F$ is a SAM and assumption $S(A2)$ is the iterate-dependent Markov noise assumption. As 
a consequence of assumption $S(A2)$, for every $z\in\mathbb{R}^{d}$ we know that the Markov chain defined by the transition kernel 
$\Pi(z,\cdot)(\cdot)$, possesses the weak Feller property (see \cite{meyntweed}). In addition to the above since the state space is compact, 
the set of stationary distributions for the Markov chain whose transition probability is given by $\Pi(z,\cdot)(\cdot)$ is non-empty for every 
$z\in\mathbb{R}^d$. Let $D(z)\subseteq\mathcal{P}(\mathcal{S})$ denote the set of stationary distributions of the Markov chain whose transition 
kernel is $\Pi(z,\cdot)(\cdot)$ (for any $z\in\mathbb{R}^d$, $\mu\in D(z)$ if and only if for every $A\in\mathscr{B}(\mathcal{S})$, $\mu(A)=
\int_{\mathcal{S}}\Pi(z,s)(A)\mu(ds)$). We also know that for every $z\in\mathbb{R}^d$, $D(z)$ is a convex and compact subset of 
$\mathcal{P}(\mathcal{S})$ and the map $z\rightarrow D(z)$ has closed graph (see \cite{vin1t} and references therein). Assumption $S(A3)$ is 
the standard step-size assumption and assumption $S(A4)$ is the general additive noise assumption which ensures that the contribution of the 
additive noise is eventually negligible (for various noise models satisfying $S(A4)$ see \cite{benaim1}). Assumption $S(A5)$ is the stability 
assumption on the iterate sequence. 

The set-valued map, $\hat{F}:\mathbb{R}^{d}\rightarrow\left\{\text{subsets of }\mathbb{R}^d\right\}$ that serves as the vector filed for the differential inclusion 
(DI) that the iterates are expected to track  is defined as,
\begin{equation*}
 \hat{F}(z):=\cup_{\mu\in D(z)}\int_{\mathcal{S}}F_{z}(s)\mu(ds),
\end{equation*}
for every $z\in\mathbb{R}^d$ where  for every $z\in\mathbb{R}^d$, $F_z$ denotes the slice as in Definition \ref{slices}$(i)$ of the set-valued map 
$F$ appearing in recursion $\eqref{str}$. The set-valued map $\hat{F}$, is a Marchaud map (see \cite[Lemma~4.7]{vin1t}) and the associated DI 
given by,
\begin{equation}\label{strdi}
 \frac{dz}{dt}\in\hat{F}(z)
\end{equation}
admits at-least one solution through every initial condition (see \cite[sec.~1.2]{benaim1}). Let $\Sigma(z_0)$ denote the set of solutions
of DI $\eqref{strdi}$ with initial condition $z_0\in\mathbb{R}^d$ and $\Sigma:=\cup_{z_0\in\mathbb{R}^d}\Sigma(z_0)$ (the set of all possible 
solutions). For every $z_0\in\mathbb{R}^d$, $\Sigma(z_0)$ is a subset of $\mathcal{C}(\mathbb{R},\mathbb{R}^{d})$, the set of all 
$\mathbb{R}^d$-valued continuous functions on $\mathbb{R}$. The set $\mathcal{C}(\mathbb{R},\mathbb{R}^{d})$ is a complete metric space 
for the metric $\bf{D}\rm$ defined by,
\begin{equation*}
 \bm{\mathrm{D}}(\bm{\mathrm{z}},\bm{\mathrm{z'}}):=\sum_{k=1}^{\infty}\frac{1}{2^k}\min\left(\left\|\bm{\mathrm{z}}-\bm{\mathrm{z'}}\right\|_{[-k,k]},1\right),
\end{equation*}
where $\left\|\bf{z}\rm-\bf{z'}\rm\right\|_{[-k,k]}:=\sup_{t\in[-k,k]}\left\|\bf{z}\rm(t)-\bf{z'}\rm(t)\right\|$. As a consequence of 
\cite[Lemma~3.1]{benaim1}, we have that $\Sigma$ and for every $z_0\in\mathbb{R}^d$, $\Sigma(z_0)$ are closed and compact subsets of 
$\mathcal{C}(\mathbb{R},\mathbb{R}^{d})$ respectively. 

Let $t_0:=0$ and for every $n\geq1$, $t(n):=\sum_{k=0}^{n-1}a(k)$. Define the stochastic process with continuous sample paths, $\bar{Z}:\Omega\times
\mathbb{R}\rightarrow\mathbb{R}^d$ as,
\begin{equation*}
 \bar{Z}(\omega,t):=\left(\frac{t-t(n)}{t(n+1)-t(n)}\right)Z_{n+1}(\omega)+\left(\frac{t(n+1)-t}{t(n+1)-t(n)}\right)Z_{n}(\omega),
\end{equation*}
for every $(\omega,t)\in\Omega\times[0,\infty)$ where $n$ is such that $t\in[t(n),t(n+1))$ and for every $(\omega,t)\in\Omega\times(-\infty,0]$, 
let $\bar{Z}(\omega,t):=Z_{0}(\omega)$. Then the main result from the analysis of recursion $\eqref{str}$ in \cite{vin1t} is as follows.

\begin{theorem}\label{strapt}
 Under assumptions $S(A1)-S(A5)$, for almost every $\omega\in\Omega$,
 \begin{itemize}
  \item [(i)] the family of functions $\left\{\bar{Z}(\omega,\cdot+t)\right\}_{t\geq0}$ is relatively compact in 
  $\mathcal{C}(\mathbb{R},\mathbb{R}^{d})$,
  \item [(ii)] every limit point of $\left\{\bar{Z}(\omega,\cdot+t)\right\}_{t\geq0}$ in $\mathcal{C}(\mathbb{R},\mathbb{R}^d)$ is a solution of 
  DI $\eqref{strdi}$, more formally,
  \begin{equation*}
   \lim_{t\to\infty}\bm{\mathrm{D}}(\bar{Z}(\omega,\cdot+t),\Sigma)=0,
  \end{equation*}
  \item [(iii)] the limit set denoted by $L(\bar{Z}(\omega,\cdot))$ defined as 
  \begin{equation*}
  L(\bar{Z}(\omega,\cdot)):=\cap_{t\geq0}\overline{\left\{\bar{Z}(\omega,q+t):q\geq0\right\}},
  \end{equation*}
  is non-empty, compact and internally chain transitive for the flow of DI $\eqref{strdi}$.
 \end{itemize}
\end{theorem}

For a proof of the above theorem see \cite[Thm.~6.6~\&~6.7]{vin1t}.

%% file: space_of_prob_meas_valued_func.tex
\subsection{Space of probability measure valued functions}
\label{sopmvf}
In this section we shall define the space of probability measure valued measurable functions on $[0,\infty)$. We shall introduce an appropriate 
topology on this space and show that such a space is compact metrizable. These spaces are used in the theory of optimal control of diffusions (see \cite{borkaropt})
and also in analyzing stochastic approximation schemes (see \cite{vin1t,borkarmark}). Quantities defined in this section will serve as tools in analyzing the stochastic recursions 
later.

Throughout this section $U$ will denote the closed unit ball in $\mathbb{R}^{d_2}$ and for any $r>0$, $B_r$ denotes the closed ball of 
radius $r$ in $\mathbb{R}^{d_1}$ centered at the origin. For every $r>0$, let $\mathcal{M}(U\times B_r\times\mathcal{S})$ denote the set of all functions 
$\gamma(\cdot)$ on $[0,\infty)$ taking values in $\mathcal{P}(U\times B_r\times\mathcal{S})$ (space of probability measures on $U\times B_r
\times\mathcal{S}$ equipped with the Prohorov topology), such that $\gamma(\cdot)$ is measurable. Formally,
\begin{equation*}
 \mathcal{M}(U\times B_r\times\mathcal{S}):=\left\{\gamma:[0,\infty)\rightarrow\mathcal{P}(U\times B_r\times\mathcal{S}): \gamma(\cdot)\ is\ 
 measurable\right\}.
\end{equation*}

Similarly for every $r>0$, $\mathcal{M}(B_r\times\mathcal{S})$ (or $\mathcal{M}(B_r)$) denotes the set of all functions $\gamma(\cdot)$ on 
$[0,\infty)$ taking values in $\mathcal{P}(B_r\times\mathcal{S})$ (or $\mathcal{P}(B_r)$) such that $\gamma(\cdot)$ is measurable. Formally,
\begin{align*}
 \mathcal{M}(B_r\times\mathcal{S})&:=\left\{\gamma:[0,\infty)\rightarrow\mathcal{P}(B_r\times\mathcal{S}): \gamma(\cdot)\ is\ 
 measurable\right\},\\
 \mathcal{M}(B_r)&:=\left\{\gamma:[0,\infty)\rightarrow\mathcal{P}(B_r): \gamma(\cdot)\ is\ measurable\right\}.
\end{align*}

For every $r>0$, let $\tau_{U\times B_r\times\mathcal{S}}$ denote the coarsest topology on $\mathcal{M}(U\times B_r\times\mathcal{S})$ which 
renders continuous the functions, $\gamma(\cdot)\rightarrow\int_{0}^{T}g(t)\left[\int_{U\times B_r\times\mathcal{S}}f(u,x,s)\gamma(t)
(du,dx,ds)\right]dt$ for every $f\in\mathcal{C}(U\times B_r\times\mathcal{S},\mathbb{R})$, for every $g\in L_{2}([0,T],\mathbb{R})$ and for 
every $T>0$. 

Similarly, for every $r>0$, let $\tau_{B_r\times\mathcal{S}}$ denote the coarsest topology on $\mathcal{M}(B_r\times\mathcal{S})$ which renders 
continuous the functions, $\gamma(\cdot)\rightarrow\int_{0}^{T}g(t)\left[\int_{B_r\times\mathcal{S}}f(x,s)\gamma(t)(dx,ds)\right]dt$
for every $f\in\mathcal{C}(B_r\times\mathcal{S},\mathbb{R})$, for every $g\in L_{2}([0,T],\mathbb{R})$ and for every $T>0$. 

Finally, for every $r>0$, let $\tau_{B_r}$ denote the coarsest topology on $\mathcal{M}(B_r)$ which renders continuous the functions, 
$\gamma(\cdot)\rightarrow\int_{0}^{T}g(t)\left[\int_{B_r}f(x)\gamma(t)(dx)\right]dt$ for every $f\in\mathcal{C}(B_r,\mathbb{R})$, for every 
$g\in L_{2}([0,T],\mathbb{R})$ and for every $T>0$. 

The following is a well known metrization lemma for the topological spaces defined above. 

\begin{lemma}\emph{[metrization]}\label{mtriz}
\begin{itemize}
 \item [(i)] For every $r>0$, the topological space $(\mathcal{M}(U\times B_r\times\mathcal{S}),\tau_{U\times B_r\times\mathcal{S}})$ is compact 
 metrizable.
 \item [(ii)] For every $r>0$, the topological space $(\mathcal{M}(B_r\times\mathcal{S}),\tau_{B_r\times\mathcal{S}})$ is compact metrizable.
 \item [(iii)] For every $r>0$, the topological space $(\mathcal{M}(B_r),\tau_{B_r})$ is compact metrizable.
 \end{itemize}
\end{lemma}
We refer the reader to \cite[Lemma~2.1]{borkarmark} for the proof of the above metrization lemma. The next lemma provides continuous functions 
between the above defined metric spaces which are used later. The proof of the lemma below is an extention of \cite[Lemma~5.2]{vin1t} to the 
above defined metric spaces. Recall that for any probability measure $\nu\in\mathcal{P}(U\times B_r\times\mathcal{S})$, 
$\nu_{B_r\times\mathcal{S}}\in\mathcal{P}(B_r\times\mathcal{S})$ denotes the image of the measure $\nu$ under the projection 
$U\times B_r\times\mathcal{S}\rightarrow B_r\times\mathcal{S}$ (that is, for every $A\in \mathscr{B}(B_r\times\mathcal{S})$, 
$\nu_{B_r\times\mathcal{S}}(A)=\int_{U\times A}\nu(du,dx,ds)$). Similarly, $\nu_{B_r}\in\mathcal{P}(B_r)$ denotes the image of 
measure $\nu$ under the projection $U\times B_r\times\mathcal{S}\rightarrow B_r$ (that is, for every $A\in\mathscr{B}(B_r)$, $\nu_{B_r}(A)=
\int_{U\times A\times \mathcal{S}}\nu(du,dx,ds)$). It is easy to see that $\nu_{B_r}$ is also the image of $\nu_{B_r\times\mathcal{S}}$ under 
the projection $B_r\times\mathcal{S}\rightarrow B_r$.

\begin{lemma}\label{contfun}
 For every $r>0$,
 \begin{itemize}
  \item [(i)] the map $\theta_1:\mathcal{P}(U\times B_r\times\mathcal{S})\rightarrow\mathcal{P}(B_r\times\mathcal{S})$ such that for every 
  $\nu\in\mathcal{P}(U\times B_r\times\mathcal{S})$, $\theta_1(\nu):=\nu_{B_r\times\mathcal{S}}$ is continuous.
  \item [(ii)] the map $\theta_2:\mathcal{P}(B_r\times\mathcal{S})\rightarrow\mathcal{P}(B_r)$ such that for every $\nu\in\mathcal{P}(B_r\times
  \mathcal{S})$, $\theta_2(\nu):=\nu_{B_r}$ is continuous.
  \item [(iii)] for any $\gamma\in\mathcal{M}(U\times B_r\times\mathcal{S})$, we have that $\theta_1\circ\gamma\in\mathcal{M}(B_r\times
  \mathcal{S})$ where for every $t\geq0$, $(\theta_1\circ\gamma)(t)=\theta_1(\gamma(t))$.
  \item [(iv)] for any $\gamma\in\mathcal{M}(B_r\times\mathcal{S})$, we have that $\theta_2\circ\gamma\in\mathcal{M}(B_r)$ where for every 
  $t\geq0$, $(\theta_2\circ\gamma)(t)=\theta_2(\gamma(t))$.
  \item [(v)]  the map $\Theta_1:\mathcal{M}(U\times B_r\times\mathcal{S})\rightarrow\mathcal{M}(B_r\times\mathcal{S})$ such that for every 
  $\gamma\in\mathcal{M}(U\times B_r\times\mathcal{S})$, $\Theta_1(\gamma):=\theta_1\circ\gamma$ is continuous.
  \item [(vi)] the map $\Theta_2:\mathcal{M}(B_r\times\mathcal{S})\rightarrow\mathcal{M}(B_r)$ such that for every $\gamma\in\mathcal{M}(B_r
  \times\mathcal{S})$, $\Theta_2(\gamma):=\theta_2\circ\gamma$ is continuous.
 \end{itemize}
\end{lemma}
\begin{proof}
 Fix $r>0$.
 \begin{itemize}
  \item  [(i)] Let $\left\{\nu^n\right\}_{n\geq1}$ be a sequence in $\mathcal{P}(U\times B_r\times\mathcal{S})$ converging to $\nu\in\mathcal{P}
  (U\times B_r\times\mathcal{S})$ as $n\to\infty$ and let $\pi:U\times B_r\times\mathcal{S}\rightarrow B_r\times\mathcal{S}$ denote the projection map such that 
  for every $(u,x,s)\in U\times B_r\times\mathcal{S}$, $\pi(u,x,s)=(x,s)$. Clearly $\pi$ is continuous and for any continuous function $f
  \in\mathcal{C}(B_r\times\mathcal{S},\mathbb{R})$, $f\circ\pi$ is continuous. Since $U\times B_r\times\mathcal{S}$ is a compact metric space, 
  from \cite[Thm.~2.1.1(ii)]{borkarap}, we get that for every $f\in\mathcal{C}(B_r\times\mathcal{S},\mathbb{R})$, $\int_{U\times B_r\times\mathcal{S}}\left
  (f\circ\pi\right)(u,x,s)\nu^n(du,dx,ds)\to\int_{U\times B_r\times\mathcal{S}}\left(f\circ\pi\right)\nu(du,dx,ds)$ as $n\to\infty$. By 
  definition, we have that for every $n\geq0$, $\nu^n_{B_r\times\mathcal{S}}=\nu^n\pi^{-1}$ (the push-forward measure) and $\nu_{B_r\times
  \mathcal{S}}=\nu\pi^{-1}$. Therefore for every 
  $f\in\mathcal{C}(B_r\times\mathcal{S},\mathbb{R})$, $\int_{B_r\times\mathcal{S}}f(x,s)\nu^n_{B_r\times\mathcal{S}}(dx,ds)\to\int_{B_r
  \times\mathcal{S}}f(x,s)\nu_{B_r\times\mathcal{S}}(dx,ds)$. Hence by \cite[Thm.~2.1.1]{borkarap} we get that $\nu^n_{B_r\times\mathcal{S}}
  \to\nu_{B_r\times\mathcal{S}}$ as $n\to\infty$ in $\mathcal{P}(B_r\times\mathcal{S})$ which gives us continuity of $\theta_1(\cdot)$.
  \item [(ii)] Similar to part $(i)$ of this lemma.
  \item [(iii) \& (iv)] Composition of measurable functions is measurable.
  \item [(v)] Let $\left\{\gamma_n\right\}_{n\geq1}$ be a sequence in $\mathcal{M}(U\times\ B_r\times\mathcal{S})$ converging to $\gamma\in
  \mathcal{M}(U\times B_r\times\mathcal{S})$ as $n\to\infty$. Then we know that for every $f\in\mathcal{C}(U\times B_r\times\mathcal{S},\mathbb{R})$, for 
  every $T>0$ and for every $g\in L_2([0,T],\mathbb{R})$, $\int_{0}^{T}g(t)\left[\int_{U\times B_r\times\mathcal{S}}f(u,x,s)\gamma_n(t)(du,dx
  ,ds)\right]dt\to\int_{0}^{T}g(t)\left[\int_{U\times B_r\times\mathcal{S}}f(u,x,s)\gamma(t)(du,dx,ds)\right]dt$ as $n\to\infty$. Let $\pi$ 
  denote the projection map as in part $(i)$ of this lemma and we know that for any $f\in\mathcal{C}(B_r\times\mathcal{S},\mathbb{R})$, $f\circ\pi\in
  \mathcal{C}(U\times B_r\times\mathcal{S},\mathbb{R})$. Then we have that for every $f\in\mathcal{C}(B_r\times\mathcal{S},\mathbb{R})$, for every $T>0$ and for every 
  $g\in L_2([0,T],\mathbb{R})$,  $\int_{0}^{T}g(t)\left[\int_{U\times B_r\times\mathcal{S}}\left(f\circ\pi\right)(u,x,s)\gamma_n(t)(du,dx,ds)
  \right]dt\to\int_{0}^{T}g(t)\left[\int_{U\times B_r\times\mathcal{S}}\left(f\circ\pi\right)(u,x,s)\gamma(t)(du,dx,ds)\right]dt$ as 
  $n\to\infty$. By arguments similar to part $(i)$ of this lemma, we have that for every $f\in\mathcal{C}(B_r\times\mathcal{S},\mathbb{R})$, for every $T>0$ 
  and for every $g\in L_2([0,T],\mathbb{R})$,  $\int_{0}^{T}g(t)\left[\int_{U\times B_r\times\mathcal{S}}f(x,s)\left(\theta_1\circ\gamma_n
  \right)(t)(dx,ds)\right]dt\to\int_{0}^{T}g(t)\left[\int_{U\times B_r\times\mathcal{S}}f(x,s)\left(\theta_1\circ\gamma\right)(t)(dx,ds)
  \right]dt$ as $n\to\infty$. Therefore $\Theta_1(\gamma_n)\to\Theta_1(\gamma)$  in $\mathcal{M}(B_r\times\mathcal{S})$ as $n\to\infty$ which gives 
  us continuity of $\Theta_1(\cdot)$.
  \item [(vi)] Similar to part $(v)$ of this lemma.\qed
 \end{itemize}
\end{proof}

%% file: recursion_and_assumptions.tex
\section{Recursion and assumptions}
\label{recass}
In this section we shall formally define the two timescale recursion as well as state and motivate the assumptions imposed ((A1)-(A10)). 

Let $(\Omega,\mathscr{F},\mathbb{P})$ denote a probability space, $\left\{X_n\right\}_{n\geq0}$ be a sequence of $\mathbb{R}^{d_1}$-valued 
random variables on $\Omega$ and $\left\{Y_n\right\}_{n\geq0}$ be a sequence of $\mathbb{R}^{d_2}$-valued random variables on $\Omega$ which satisfy 
for every $n\geq0$,
\begin{subequations}\label{2tr}
\begin{align}
\label{stsr}
 Y_{n+1}-Y_{n}-b(n)M^{(2)}_{n+1}&\in b(n)H_2(X_n,Y_n,S^{(2)}_n),\\
 \label{fstr}
 X_{n+1}-X_{n}-a(n)M^{(1)}_{n+1}&\in a(n)H_1(X_n,Y_n,S^{(1)}_n),
\end{align}
\end{subequations}
where,
\begin{itemize}
 \item [(A1)] the map $H_1:\mathbb{R}^d\times\mathcal{S}^{(1)}\rightarrow\left\{\text{subsets of }\mathbb{R}^{d_1}\right\}$ with $\mathcal{S}^{(1)}$ a compact metric 
 space with metric $d_{\mathcal{S}^{(1)}}$, is such that,
 \begin{itemize}
 \item [(i)] for every $(x,y,s^{(1)})\in\mathbb{R}^{d}\times\mathcal{S}^{(1)}$, $H_1(x,y,s^{(1)})$ is a convex and compact subset of $\mathbb{R}^{d_1}$,
 \item [(ii)] there exists $K>0$, such that, for every $(x,y,s^{(1)})\in\mathbb{R}^{d}\times\mathcal{S}^{(1)}$, $\sup_{x'\in H_1(x,y,s^{(1)})}\left\|x'\right
 \|\leq K(1+\left\|x\right\|+\left\|y\right\|)$,
 \item [(iii)] for every $(x,y,s^{(1)})\in\mathbb{R}^d\times\mathcal{S}^{(1)}$, for every $\left(\mathbb{R}^d\times\mathcal{S}^{(1)}\right)$-valued sequence, $\left\{(x_n,y_n,s^{(1)}_n)\right\}_{n\geq1}$ converging to $(x,y,s^{(1)})\in
 \mathbb{R}^d\times\mathcal{S}^{(1)}$, for every sequence $\left\{x'_n\in H_1(x_n,y_n,s^{(1)}_n)\right\}_{n\geq1}$ converging to $x'\in\mathbb{R}^{d_1}$, 
 we have that $x'\in H_1(x,y,s^{(1)})$.
 \end{itemize}
 \item [(A2)] the map $H_2:\mathbb{R}^d\times\mathcal{S}^{(2)}\rightarrow\left\{\text{subsets of }\mathbb{R}^{d_2}\right\}$ with $\mathcal{S}^{(2)}$ a compact metric 
 space with metric $d_{\mathcal{S}^{(2)}}$, is such that,
 \begin{itemize}
 \item [(i)] for every $(x,y,s^{(2)})\in\mathbb{R}^{d}\times\mathcal{S}^{(2)}$, $H_2(x,y,s^{(2)})$ is a convex and compact subset of $\mathbb{R}^{d_2}$,
 \item [(ii)] there exists $K>0$, such that, for every $(x,y,s^{(2)})\in\mathbb{R}^{d}\times\mathcal{S}^{(2)}$, $\sup_{y'\in H_2(x,y,s^{(2)})}\left\|y'\right
 \|\leq K(1+\left\|x\right\|+\left\|y\right\|)$,
 \item [(iii)] for every $(x,y,s^{(2)})\in\mathbb{R}^d\times\mathcal{S}^{(2)}$, for every $\left(\mathbb{R}^d\times\mathcal{S}^{(2)}\right)$-valued sequence, $\left\{(x_n,y_n,s^{(2)}_n)\right\}_{n\geq1}$ converging to $(x,y,s^{(2)})\in
 \mathbb{R}^d\times\mathcal{S}^{(2)}$, for every sequence $\left\{y'_n\in H_2(x_n,y_n,s^{(2)}_n)\right\}_{n\geq1}$ converging to $y'\in\mathbb{R}^{d_2}$, 
 we have that $y'\in H_2(x,y,s^{(2)})$.
 \end{itemize}
 \item [(A3)] $\left\{S^{(1)}_n\right\}_{n\geq0}$ is a sequence of $\mathcal{S}^{(1)}$-valued random variables on $\Omega$, such that for every $n\geq0$,
 for every $A\in\mathscr{B}(\mathcal{S}^{(1)})$, $\mathbb{P}(S^{(1)}_{n+1}\in A|S^1_m,X_m,Y_m,\ 0\leq m\leq n)=\mathbb{P}(S^{(1)}_{n+1}\in A|S^{(1)}_n,X_n,Y_n)=
 \Pi^{(1)}(X_n,Y_n,S^{(1)}_n)(A)$ $a.s.$, where $\Pi^{(1)}:\mathbb{R}^d\times\mathcal{S}^{(1)}\rightarrow\mathcal{P}(\mathcal{S}^{(1)})$ is continuous.
 \item [(A4)] $\left\{S^{(2)}_n\right\}_{n\geq0}$ is a sequence of $\mathcal{S}^{(2)}$-valued random variables on $\Omega$, such that for every $n\geq0$,
 for every $A\in\mathscr{B}(\mathcal{S}^{(2)})$, $\mathbb{P}(S^{(2)}_{n+1}\in A|S^2_m,X_m,Y_m,\ 0\leq m\leq n)=\mathbb{P}(S^{(2)}_{n+1}\in A|S^{(2)}_n,X_n,Y_n)=
 \Pi^{(2)}(X_n,Y_n,S^{(2)}_n)(A)$ $a.s.$, where $\Pi^{(2)}:\mathbb{R}^d\times\mathcal{S}^{(2)}\rightarrow\mathcal{P}(\mathcal{S}^{(2)})$ is continuous.
 \item [(A5)] $\left\{a(n)\right\}_{n\geq0}$ and $\left\{b(n)\right\}_{n\geq0}$ are two sequences of positive real numbers satisfying,
 \begin{itemize}
 \item [(i)] $a(0)\leq1$ and for every $n\geq0$, $a(n)\geq a(n+1)$,
 \item [(ii)] $b(0)\leq1$ and for every $n\geq0$, $b(n)\geq b(n+1)$,
 \item [(iii)] $\lim_{n\to\infty}\frac{b(n)}{a(n)}=0$,
 \item [(iv)] $\sum_{n=0}^{\infty}a(n)=\sum_{n=0}^{\infty}b(n)=\infty$ and $\sum_{n=0}^{\infty}\left((a(n))^2+(b(n))^2\right)<\infty$.
 \end{itemize}
 \item [(A6)] $\{M^{(1)}_{n}\}_{n\geq1}$ is a sequence of $\mathbb{R}^{d_1}$-valued random variables on $\Omega$ such that for $a.e.(\omega)$, for any $T>0$, 
 $\\\lim_{n\to\infty}\sup_{n\leq k\leq \tau^1(n,T)}\left\|\sum_{m=n}^{k}a(m)M^{(1)}_{m+1}(\omega)\right\|=0$ where $\tau^1(n,T):=\min\left\{m>n:\sum_{k=n}^{m-1}
 a(k)\geq T\right\}$.
 \item [(A7)] $\{M^{(2)}_{n}\}_{n\geq1}$ is a sequence of $\mathbb{R}^{d_2}$-valued random variables on $\Omega$ such that for $a.e.(\omega)$, for any $T>0$, 
 $\\\lim_{n\to\infty}\sup_{n\leq k\leq \tau^2(n,T)}\left\|\sum_{m=n}^{k}b(m)M^{(2)}_{m+1}(\omega)\right\|=0$ where $\tau^2(n,T):=\min\left\{m>n:\sum_{k=n}^{m-1}
 b(k)\geq T\right\}$.
 \item [(A8)] $\mathbb{P}\left(\sup_{n\geq0}\left(\left\|X_n\right\|+\left\|Y_n\right\|\right)<\infty\right)=1$.
\end{itemize}
 
Assumptions $(A1)$ and $(A2)$ ensure that $H_1$ and $H_2$ are SAMs. Assumptions $(A3)$ and $(A4)$ are the iterate-dependent Markov noise 
assumptions. Under $(A3)$, for every $(x,y)\in\mathbb{R}^d$, the Markov 
chain associated with the transition kernel given by $\Pi^{(1)}(x,y,\cdot)(\cdot)$ possesses the weak Feller property (see \cite{meyntweed}). In 
addition to the above since $\mathcal{S}^{(1)}$ is a compact metric space, the Markov chain associated with the transition kernel $\Pi^{(1)}(x,y,\cdot)
(\cdot)$ has at least one stationary distribution for every $(x,y)\in\mathbb{R}^d$ ($\mu\in\mathcal{P}(\mathcal{S}^{(1)})$ is stationary for the 
Markov chain associated with the transition kernel $\Pi^{(1)}(x,y,\cdot)(\cdot)$ if for every $A\in\mathscr{B}(\mathcal{S}^{(1)})$, $\mu(A)=\int_
{\mathcal{S}^{(1)}}\Pi^{(1)}(x,y,s^{(1)})(A)\mu(ds^{(1)})$). For every $(x,y)\in\mathbb{R}^d$, let $D^{(1)}(x,y)\subseteq\mathcal{P}(\mathcal{S}^{(1)})$ denote the set 
of stationary distributions of the Markov chain associated with the transition kernel $\Pi^{(1)}(x,y,\cdot)(\cdot)$. It can easily be shown that,
\begin{itemize}
 \item [(i)] for every $(x,y)\in\mathbb{R}^d$, $D^{(1)}(x,y)$ is a convex and compact subset of $\mathcal{P}(\mathcal{S}^{(1)})$,
 \item [(ii)] graph of the map $(x,y)\rightarrow D^{(1)}(x,y)$ is closed, that is, the set
 \begin{equation*}
  \mathcal{G}(D^{(1)}):=\left\{(x,y,\mu)\in\mathbb{R}^d\times\mathcal{P}(\mathcal{S}^{(1)}): (x,y)\in\mathbb{R}^d,\ \mu\in D^{(1)}(x,y)\right\},
 \end{equation*}
is a closed subset of $\mathbb{R}^d\times\mathcal{P}(\mathcal{S}^{(1)})$.
\end{itemize}
The proofs of the above two statements are similar to that in \cite[pg.~69]{borkartxt}.
Similarly under assumption $(A4)$, for every $(x,y)\in\mathbb{R}^d$ the set of stationary distributions (denoted by $D^{(2)}(x,y)$) associated with the
Markov chain defined by the transition kernel $\Pi^{(2)}(x,y,\cdot)(\cdot)$ is a non-empty, convex and compact subset of $\mathcal{P}(\mathcal{S}^{(2)})$ 
and further the map $(x,y)\rightarrow D^{(2)}(x,y)$ has a closed graph (that is the set $\mathcal{G}(D^{(2)})$ defined in an analogous manner as $\mathcal{G}(D^{(1)})$ is 
a closed subset of $\mathbb{R}^d\times\mathcal{P}(\mathcal{S}^{(2)})$).

Assumption $(A5)$ is the standard two timescale step size assumption. Assumption $(A5)(iii)$ tells that eventually the time step taken by 
recursion $\eqref{stsr}$ is smaller than the time step taken by recursion $\eqref{fstr}$. Hence recursion $\eqref{stsr}$ is called the slower 
timescale recursion and the recursion $\eqref{fstr}$ is called the faster timescale recursion. Assumptions $(A6)$ and $(A7)$ are the conditions 
that the additive noise terms satisfy. These guarantee that the contribution of additive noise terms is eventually negligible. For various noise 
models where these additive noise assumptions are satisfied we refer the reader to \cite{benaim1}.

Assumption $(A8)$ is the stability assumption which ensures that the iterates remain within a bounded set. While this is a standard requirement in the study of such recursions, it is highly 
nontrivial. An important future direction would be to provide sufficient conditions for verification of $(A8)$.

The Markov noise terms in the faster timescale, in limit will average the drift function $H_1$ w.r.t. the stationary distributions given by 
the map $(x,y)\rightarrow D^{(1)}(x,y)$. The appropriate set-valued map whose associated DI the faster timescale recursion is expected to 
track is given by, 
\begin{equation}\label{fstrmap}
 \hat{H}_1(x,y):=\cup_{\mu\in D^{(1)}(x,y)}\int_{\mathcal{S}^{(1)}}H_{1,(x,y)}(s^{(1)})\mu(ds^{(1)}),
\end{equation}
for every $(x,y)\in\mathbb{R}^d$ where for every $(x,y)\in\mathbb{R}^d$, $H_{1,(x,y)}$ denotes the slice as in Definition\ref{slices}$(i)$ of the 
set-valued map $H_1$ in the recursion \eqref{fstr}. As a consequence of the step size assumption $(A5)$, with respect to the faster timescale \eqref{fstr},
the slower timescale recursion \eqref{stsr} appears to be static and one would expect that the family of DIs, 
\begin{equation}\label{fstrdi}
 \frac{dx}{dt}\in \hat{H}_1(x,y_0),
\end{equation}
obtained by fixing some $y_0\in\mathbb{R}^{d_2}$ to describe the behavior of the faster timescale recursion \eqref{fstr}. Before we proceed, we 
need to ensure that for every $y_0\in\mathbb{R}^{d_2}$, the DI \eqref{fstrdi} has solutions through every initial condition. The next lemma 
states the map $\hat{H}_1(\cdot,y_0)$ is a Marchaud map for every $y_0\in\mathbb{R}^{d_2}$, which ensures that the DI \eqref{fstrdi} has solutions.
\begin{lemma}\label{march1}
 For every $y_0\in\mathbb{R}^{d_2}$, the set-valued map $\hat{H}_1(\cdot,y_0):\mathbb{R}^{d_1}\rightarrow\left\{\text{subsets of }\mathbb{R}^{d_1}
 \right\}$ is a Marchaud map.
\end{lemma}
Proof of the above lemma is given in section \ref{limdiprop}. The next assumption will ensure that for every $y_0\in\mathbb{R}^{d_2}$, the DI \eqref{fstrdi} 
has a global attractor to which one expects the faster time scale iterates $\{X_n\}_{n\geq0}$ to converge. 

\begin{itemize}
 \item [(A9)] For every $y_0\in\mathbb{R}^{d_2}$, the DI \eqref{fstrdi} admits a globally attracting set, $A_{y_0}$. 
 The map $\lambda:\mathbb{R}_{d_2}\rightarrow\left\{\text{subsets of }\mathbb{R}^{d_1}\right\}$ where for every $y\in\mathbb{R}^{d_2}$, 
 $\lambda(y):=A_{y}$ is such that
 \begin{itemize}
 \item [(i)] for every $y\in\mathbb{R}^{d_2}$, $\sup_{x\in \lambda(y)}\left\|x\right\|\leq K(1+\left\|y\right\|)$,
 \item [(ii)] for every $y\in\mathbb{R}^{d_2}$, for every $\mathbb{R}^{d_2}$-valued sequence, $\left\{y_n\right\}_{n\geq1}$ converging to $y\in\mathbb{R}^{d_2}$, for every $\left
 \{x_n\in\lambda(y_n)\right\}_{n\geq0}$ converging $x\in\mathbb{R}^{d_1}$, we have $x\in\lambda(y)$.
 \end{itemize}
\end{itemize}

With respect to the slower timescale recursion \eqref{stsr}, the faster time scale recursion will appear to have equilibrated. Further the Markov 
noise terms average the set-valued drift function $H_2$ with respect to the stationary distributions. In what follows we construct the set-valued map 
that the slower timescale recursion is expected to track which captures both the equilibration of the faster timescale and the averaging by the 
Markov noise terms.

Before we proceed recall that $\mathcal{P}(\mathbb{R}^{d_1}\times\mathcal{S}^{(2)})$ denotes the set of probability measures on $\mathbb{R}^{d_1}\times
\mathcal{S}^{(2)}$ with the Prohorov topology. For any $\mu\in\mathcal{P}(\mathbb{R}^{d_1}\times\mathcal{S}^{(2)})$, $\mu_{\mathbb{R}^{d_1}}\in\mathcal{P}(
\mathbb{R}^{d_1})$ and $\mu_{\mathcal{S}^{(2)}}\in\mathcal{P}(\mathcal{S}^{(2)})$ denote the images of the probability measure $\mu$ under projections 
$\mathbb{R}^{d_1}\times\mathcal{S}^{(2)}\rightarrow\mathbb{R}^{d_1}$ and $\mathbb{R}^{d_1}\times\mathcal{S}^{(2)}\rightarrow\mathcal{S}^{(2)}$ respectively 
(for any $A\in\mathscr{B}(\mathbb{R}^{d_1})$, $\mu_{\mathbb{R}^{d_1}}(A):=\int_{A\times\mathcal{S}^{(2)}}\mu(dx,ds^{(2)})$ and similarly for every $A
\in\mathscr{B}(\mathcal{S}^{(2)})$, $\mu_{\mathcal{S}^{(2)}}(A):=\int_{\mathbb{R}^{d_1}\times A}\mu(dx,ds^{(2)})$).

Define the map $D:\mathbb{R}^{d_2}\rightarrow\left\{\text{subsets of }\mathcal{P}(\mathbb{R}^{d_1}\times\mathcal{S})\right\}$ such that
for every $y\in\mathbb{R}^{d_2}$ , 
\begin{small}
\begin{equation}\label{measmapdefn}
 D(y):=\!\!\left\{\mu\!\in\!\mathcal{P}(\mathbb{R}^{d_1}\times\mathcal{S}^{(2)}): \mathrm{supp}(\mu_{\mathbb{R}^{d_1}})\!\subseteq\!\lambda(y)\ and\ for\ every\ A\!\in\!
 \mathscr{B}(\mathcal{S}^{(2)}),\ \mu_{\mathcal{S}^{(2)}}(A)\!=\!\!\int_{\mathcal{S}^{(2)}}\!\!\!\!\!\Pi^{(2)}(x,y,s^{(2)})(A)\mu(dx,ds^{(2)})\right\},
\end{equation}
\end{small}
where $\mathrm{supp}(\mu_{\mathbb{R}^{d_1}})$ denotes the support of measure $\mu_{\mathbb{R}^{d_1}}$ (that is $\mathrm{supp}(\mu_{\mathbb{R}^{d_1}})\subseteq
\mathbb{R}^{d_1}$ is a closed set such that $\mu_{\mathbb{R}^{d_1}}(\mathrm{supp}(\mu_{\mathbb{R}^{d_1}}))=1$ and for every closed set 
$A\subseteq\mathbb{R}^{d_1}$, with $\mu_{\mathbb{R}^{d_1}}(A)=1$ we have $\mathrm{supp}(\mu_{\mathbb{R}^{d_1}})\subseteq A$). A natural question to 
ask is whether $D(y)$ is non-empty for every $y\in\mathbb{R}^{d_2}$ and if it is non-empty, what properties the map $D(\cdot)$ possesses and its
relation to the stationary distributions of the Markov noise terms $\left\{S^{(2)}_n\right\}_{n\geq0}$. The lemma below answers these questions.

\begin{lemma}\label{measmap}
 The map $D(\cdot)$ defined in \eqref{measmapdefn} satisfies,
 \begin{itemize}
  \item [(i)] for every $y\in\mathbb{R}^{d_2}$, $D(y)$ is non-empty, convex and compact subset of $\mathcal{P}(\mathbb{R}^{d_1}\times\mathcal{S}^{(2)})$,
  \item [(ii)] for every $y\in\mathbb{R}^{d_2}$, for every $\mathbb{R}^{d_2}$-valued sequence, $\left\{y_n\right\}_{n\geq1}$ converging to $y\in\mathbb{R}^{d_2}$, for every 
  $\mathcal{P}(\mathbb{R}^{d_1}\times\mathcal{S}^{(2)})$-valued sequence $\left\{\mu^n\in D(y_n)\right\}_{n\geq1}$ converging to $\mu\in\mathcal{P}
  (\mathbb{R}^{d_1}\times\mathcal{S}^{(2)})$, we have $\mu\in D(y)$.
  \item [(iii)] for every $y\in\mathbb{R}^{d_2}$, $\bar{co}\left(\left\{\delta_{x^*}\otimes\nu\in\mathcal{P}(\mathbb{R}^{d_1}\times\mathcal{S}^{(2)})
  :x^*\in\lambda(y),\ \nu\in D^{(2)}(x^*,y)\right\}\right)\subseteq D(y)$, where for any $x\in\mathbb{R}^{d_1}$, $\delta_x$ denotes the Dirac measure.
 \end{itemize}
\end{lemma}
\begin{proof}$\\$
\begin{itemize}
 \item [(i)] Fix $y\in\mathbb{R}^{d_2}$. Consider the product measure $\mu:=\delta_{x^*}\otimes\nu\in\mathcal{P}(\mathbb{R}^{d_1}\times
 \mathcal{S}^{(2)})$ where, $\delta_{x^*}\in\mathcal{P}(\mathbb{R}^{d_1})$ denotes the Dirac measure on some $x^*\in\lambda(y)$ (that is for every 
 $A\in\mathscr{B}(\mathbb{R}^{d_1})$, $\delta_{x^*}(A)=1$ if $x^*\in A$, $\delta_{x^*}(A)=0$ otherwise) and $\nu\in\mathcal{P}(\mathcal{S}^{(2)})$ is 
 such that $\nu\in D^{(2)}(x^*,y)$ (that is $\nu$ is a stationary measure of the Markov chain whose transition kernel is given by $\Pi^{(2)}(x^*,y,\cdot)
 (\cdot)$). Then $\mu_{\mathbb{R}^{d_1}}=\delta_{x^*}$ and since $x^*\in\lambda(y)$, $\mathrm{supp}(\mu_{\mathbb{R}^{d_1}})=\{x^*\}\subseteq\lambda(y)$. Further 
 $\mu_{\mathcal{S}^{(2)}}=\nu$ and for every $A\in\mathscr{B}(\mathcal{S}^{(2)})$, $\int_{\mathbb{R}^{d_1}\times\mathcal{S}^{(2)}}\Pi^{(2)}(x,y,s)(A)\mu(dx,ds)=
 \int_{\mathcal{S}^{(2)}}\left[\int_{\mathbb{R}^{d_1}}\Pi^{(2)}(x,y,s^{(2)})(A)\delta_{x^*}(dx)\right]\nu(ds^{(2)})=\int_{\mathcal{S}^{(2)}}\Pi^{(2)}(x^*,y,s^{(2)})(A)\nu
 (ds^{(2)})=\nu(A)$ where the last equality follows from the fact that $\nu\in D^{(2)}(x^*,y)$. Therefore $\delta_{x^*}\otimes\nu\in D(y)$ and hence 
 $D(y)\neq\emptyset$.
 
 Let $\mu^1,\mu^2\in D(y)$ and $\alpha\in (0,1)$. Consider the measure $\mu:=\alpha\mu^1+(1-\alpha)\mu^2$ (that is for any $A\in\mathscr{B}(\mathbb{R}
 ^{d_1}\times\mathcal{S}^{(2)})$, $\mu(A)=\alpha\mu^1(A)+(1-\alpha)\mu^2(A)$). Clearly $\mu\in\mathcal{P}(\mathbb{R}^{d_1}\times\mathcal{S}^{(2)})$, 
 $\mu_{\mathbb{R}^{d_1}}=\alpha\mu^1_{\mathbb{R}^{d_1}}+(1-\alpha)\mu^2_{\mathbb{R}^{d_1}}$ and $\mu_{\mathcal{S}^{(2)}}=\alpha\mu^1_{\mathcal{S}^{(2)}}
 +(1-\alpha)\mu^2_{\mathcal{S}^{(2)}}$. For $i\in{1,2}$, $\mathrm{supp}(\mu^i_{\mathbb{R}^{d_1}})\subseteq \lambda(y)$, from which we have $\mu^i_{\mathbb{R}^{d_1}}(
 \lambda(y))=1$ and hence $\mu_{\mathbb{R}^{d_1}}(\lambda(y))=\alpha\mu^1_{\mathbb{R}^{d_1}}(\lambda(y))+(1-\alpha)\mu^2_{\mathbb{R}^{d_1}}(
 \lambda(y))=1$. Therefore $\mathrm{supp}(\mu_{\mathbb{R}^{d_1}})\subseteq \lambda(y)$. For every $A\in\mathscr{B}(\mathcal{S}^{(2)})$, 
 $\int_{\mathbb{R}^{d_1}\times\mathcal{S}^{(2)}}\Pi^{(2)}(x,y,s^{(2)})(A)\mu(dx,ds^{(2)})=\alpha\int_{\mathbb{R}^{d_1}\times\mathcal{S}^{(2)}}\Pi^{(2)}(x,y,s^{(2)})(A)\mu^1
 (dx,ds^{(2)})+(1-\alpha)\int_{\mathbb{R}^{d_1}\times\mathcal{S}^{(2)}}\Pi^{(2)}(x,y,s^{(2)})(A)\mu^2(dx,ds^{(2)})=\alpha\mu^1_{\mathcal{S}^{(2)}}(A)+(1-\alpha)\mu^2_
 {\mathcal{S}^{(2)}}(A)=\mu_{\mathcal{S}^{(2)}}(A)$. Therefore $\mu:=\alpha\mu^1+(1-\alpha)\mu^2\in D(y)$ which gives us the convexity of $D(y)$.
 
 In order to show that $D(y)$ is compact, we will first show that the set $D(y)$ is a closed set. Consider $\{\mu^n\}_{n\geq1}$ such that for 
 every $n\geq1$, $\mu^n\in D(y)$ converging to $\mu\in\mathcal{P}(\mathbb{R}^d\times\mathcal{S}^{(2)})$. Clearly $\{\mu^n_{\mathbb{R}^{d_1}}\}_{n\geq1}$
 converges to $\mu_{\mathbb{R}^{d_1}}$ in $\mathcal{P}(\mathbb{R}^{d_1})$. Since for every $n\geq1$, since $\mathrm{supp}(\mu^n_{\mathbb{R}^{d_1}})\subseteq
 \lambda(y)$, we have $\mu^n(\lambda(y))=1$ for every $n\geq1$. By assumption $(A9)$, $\lambda(y)$ is a compact subset of $\mathbb{R}^{d_1}$ and 
 by \cite[Thm.~2.1.1(iv)]{borkarap}, we have $\limsup_{n\to\infty}\mu^n_{\mathbb{R}^{d_1}}(\lambda(y))\leq\mu_{\mathbb{R}^{d_1}}(\lambda(y))$. 
 Therefore $\mu_{\mathbb{R}^{d_1}}(\lambda(y))=1$ which gives us that $\mathrm{supp}(\mu_{\mathbb{R}^{d_1}})\subseteq\lambda(y)$. Clearly $\left\{\mu^n_{
 \mathcal{S}^{(2)}}\right\}_{n\geq1}$ converges to $\mu_{\mathcal{S}^{(2)}}$ in $\mathcal{P}(\mathcal{S}^{(2)})$. Since $\mathcal{S}^{(2)}$ is a compact metric 
 space, by \cite[Thm.~2.1.1(ii)]{borkarap} we know that for every $f\in\mathcal{C}(\mathcal{S}^{(2)},\mathbb{R})$, $\int_{\mathcal{S}^{(2)}}f(\tilde{s}^{(2)}
 )\mu^n_{\mathcal{S}^{(2)}}(d\tilde{s}^{(2)})\to\int_{\mathcal{S}^{(2)}}f(\tilde{s}^{(2)})\mu_{\mathcal{S}^{(2)}}(d\tilde{s}^{(2)})$ as $n\to\infty$. Let $\nu^n(d\tilde
 {s}^2):=\int_{\mathbb{R}^{d_1}\times\mathcal{S}^{(2)}}\Pi^{(2)}(x,y,s^{(2)})(d\tilde{s}^{(2)})\mu(dx,ds^{(2)})\in\mathcal{P}(\mathcal{S}^{(2)})$ for every $n\geq1$ and 
 $\nu(d\tilde{s}^{(2)}):=\int_{\mathbb{R}^{d_1}\times\mathcal{S}^{(2)}}\Pi^{(2)}(x,y,s^{(2)})(d\tilde{s}^{(2)})\mu(dx,ds^{(2)})$. It 
 is easy to see that for any $f\in\mathcal{C}(\mathcal{S}^{(2)},\mathbb{R})$, $\int_{\mathcal{S}^{(2)}}f(\tilde{s}^{(2)})\nu^n(d\tilde{s}^{(2)})=\int_{
 \mathbb{R}^{d_1}\times\mathcal{S}^{(2)}}\left[\int_{\mathcal{S}^{(2)}}f(\tilde{s}^{(2)})\Pi^{(2)}(x,y,s^{(2)})(d\tilde{s}^{(2)})\right]\mu^n(dx,ds^{(2)})$. By assumption 
 $(A4)$, $(x,s^{(2)})\rightarrow\int_{\mathcal{S}^{(2)}}f(\tilde{s}^{(2)})\Pi^{(2)}(x,y,s^{(2)})(d\tilde{s}^{(2)})$ is continuous for any $f\in\mathcal{C}
 (\mathcal{S}^{(2)},\mathbb{R})$. Therefore as $\mu^n\to\mu$ in $\mathcal{P}(\mathbb{R}^{d_1}\times\mathcal{S}^{(2)})$, we have 
 $\int_{\mathbb{R}^{d_1}\times\mathcal{S}^{(2)}}\left[\int_{\mathcal{S}^{(2)}}f(\tilde{s}^{(2)})\Pi^{(2)}(x,y,s^{(2)})(d\tilde{s}^{(2)})\right]\mu^n(dx,ds^{(2)})\to\int_{
 \mathbb{R}^{d_1}\times\mathcal{S}^{(2)}}\left[\int_{\mathcal{S}^{(2)}}f(\tilde{s}^{(2)})\Pi^{(2)}(x,y,s^{(2)})(d\tilde{s}^{(2)})\right]\mu(dx,ds^{(2)})$ or $\int_
 {\mathcal{S}^{(2)}}f(\tilde{s}^{(2)})\nu^n(d\tilde{s}^{(2)})\to\int_{\mathcal{S}^{(2)}}f(\tilde{s}^{(2)})\nu(d\tilde{s}^{(2)})$. Since for every $n\geq1$, 
 $\mu^n\in D(y)$, we have $\int_{\mathcal{S}^{(2)}}f(\tilde{s}^{(2)})\mu^n_{\mathcal{S}^{(2)}}(d\tilde{s}^{(2)})=\int_{\mathcal{S}^{(2)}}f(\tilde{s}^{(2)})\nu^n(d
 \tilde{s}^{(2)})$ for every $f\in\mathcal{C}(\mathcal{S}^{(2)},\mathbb{R})$. Thus for every $f\in\mathcal{C}(\mathcal{S}^{(2)},\mathbb{R})$, we have 
 $\int_{\mathcal{S}^{(2)}}f(\tilde{s}^{(2)})\mu_{\mathcal{S}^{(2)}}(d\tilde{s}^{(2)})=\int_{\mathcal{S}^{(2)}}f(\tilde{s}^{(2)})\nu(d\tilde{s}^{(2)})$. Therefore 
 $\mu_{\mathcal{S}^{(2)}}=\nu$ which establishes that $\mu\in D(y)$ and hence $D(y)$ is closed. To establish compactness of $D(y)$ it is now enough 
 to show that the set $D(y)$ is relatively compact in $\mathcal{P}(\mathbb{R}^{d_1}\times\mathcal{S}^{(2)})$. For any measure $\mu\in D(y)$, the 
 support of the measure $\mu$, denoted by $\mathrm{supp}(\mu)$ is contained in $\lambda(y)\times\mathcal{S}^{(2)}$ which is a compact set independent of $\mu
 \in D(y)$. Thus the family of measures $\left\{\mu:\mu\in D(y)\right\}$ is tight and by Prohorov's theorem (see \cite[Thm.~2.3.1]{borkarap}) 
 we have that the set of measures $D(y)$ is relatively compact in $\mathcal{P}(\mathbb{R}^{d_1}\times\mathcal{S}^{(2)})$. Therefore $D(y)$ is closed 
 and relatively compact and hence is compact.
 \item [(ii)] Let $y_n\to y$ in $\mathbb{R}^{d_2}$ and $\mu^n\in D(y_n)\to\mu$ in $\mathcal{P}(\mathbb{R}^{d_1}\times\mathcal{S}^{(2)})$ as 
 $n\to\infty$. Let $B_1$ denote the closed unit ball in $\mathbb{R}^{d_1}$. By assumption $(A9)$, we have that the set-valued map $y\rightarrow
 \lambda(y)$ is u.s.c. Therefore for every $\epsilon>0$, there exists $\delta>0$(depending on $\epsilon$ and $y$) such that for every $y'\in
 \mathbb{R}^{d_1}$, satisfying $\left\|y'-y\right\|<\delta$ we have $\lambda(y')\subseteq\lambda(y)+\epsilon B_1$. Since $\lambda(y)$ is compact, 
 $\lambda(y)+\epsilon B_1$ is compact. Since $y_n\to y$, there exists $N$ such that for every $n\geq N$, $\left\|y_n-y\right\|<\delta$. Then for 
 all $n\geq N$, $\lambda(y_n)\subseteq \lambda(y)+\epsilon B_1$. By the above we have that $\limsup_{n\to\infty}\mu^n_{\mathbb{R}^{d_1}}(\lambda(y)+\epsilon(B_1))
 =1$ for every $\epsilon>0$. Since $\mu^n\to\mu$, we have that $\mu^n_{\mathcal{R}^{d_1}}\to\mu_{\mathcal{R}^{d_1}}$ in $\mathcal{P}(\mathcal{R}^{d_1})$ and 
 by \cite[Thm.~2.1.1(iv)]{borkarap}, we have that for every $\epsilon>0$, $\mu_{\mathbb{R}^{d_1}}(\lambda(y)+\epsilon B_1)=1$. Since $\lambda(y)$ is compact, 
 $\lambda(y)=\cap_{n\geq1}(\lambda(y)+\frac{1}{n}B_1)$ and $\mu_{\mathbb{R}^{d_1}}(\lambda(y))=\lim_{n\to\infty}\mu_{\mathbb{R}^{d_1}}(\lambda(y)+\frac{1}{n}B_1)=1$. Therefore 
 $\mathrm{supp}(\mu_{\mathbb{R}^{d_1}})\subseteq\lambda(y)$. Let $\nu^n(d\tilde{s}^{(2)}):=\int_{\mathbb{R}^{d_1}\times\mathcal{S}^{(2)}}\Pi^{(2)}(x,y_n,s^{(2)})(d\tilde{
 s}^2)\mu^n(dx,ds^{(2)})\in\mathcal{P}(S^2)$ and $\nu(d\tilde{s}^{(2)}):=\int_{\mathbb{R}^{d_1}\times\mathcal{S}^{(2)}}\Pi^{(2)}(x,y,s^{(2)})(d\tilde{s}^{(2)})\mu(dx
 ,ds^{(2)})$. Then for any $f\in\mathcal{C}(\mathcal{S}^{(2)},\mathbb{R})$, for any $n\geq1$, $\int_{\mathcal{S}^{(2)}}f(\tilde{s}^{(2)})\nu^n(\tilde{s}^{(2)})=
 \int_{\mathbb{R}^{d_1}\times\mathcal{S}^{(2)}}\left[\int_{\mathcal{S}^{(2)}}f(\tilde{s}^{(2)})\Pi^{(2)}(x,y_n,s^{(2)})(d\tilde{s}^{(2)})\right]\mu^n(dx,ds^{(2)})$ and 
 $\int_{\mathcal{S}^{(2)}}f(\tilde{s}^{(2)})\nu(\tilde{s}^{(2)})=\int_{\mathbb{R}^{d_1}\times\mathcal{S}^{(2)}}\left[\int_{\mathcal{S}^{(2)}}f(\tilde{s}^{(2)})
 \Pi^{(2)}(x,y,s^{(2)})(d\tilde{s}^{(2)})\right]\mu(dx,ds^{(2)})$. Since for every $n\geq1$, $\mu^n\in D(y_n)$, we have that $\mathrm{supp}(\mu^n)\subseteq 
 \lambda(y_n)\times S^2$. By using the u.s.c. property of the map $\lambda(\cdot)$ and the fact that $y_n\to y$, we get that for any $\epsilon>0$, there 
 exists $N$ such that for every $n\geq N$, $\lambda(y_n)\times\mathcal{S}^{(2)}\subseteq \left(\lambda(y)+\epsilon B_1\right)\times\mathcal{S}^{(2)}$. 
 Therefore for every $f\in\mathcal{C}(\mathcal{S}^{(2)},\mathbb{R})$, for every $n\geq N$, $\int_{\mathcal{S}^{(2)}}f(\tilde{s}^{(2)})\nu^n(d\tilde{s}^{(2)})=
 \int_{\left(\lambda(y)+\epsilon B_1\right)\times\mathcal{S}^{(2)}}\left[\int_{\mathcal{S}^{(2)}}f(s^{(2)})\Pi^{(2)}(x,y_n,s^{(2)})(d\tilde{s}^{(2)})\right]\mu^n(dx,ds^{(2)})$.
 By assumption $(A4)$, the map $(x,y,s^{(2)})\rightarrow\int_{\mathcal{S}^{(2)}}f(tilde{s}^{(2)})\Pi^{(2)}(x,y,s^{(2)})(d\tilde{s}^{(2)})$ is continuous and hence its restriction 
 to the compact set $\left(\lambda(y)+\epsilon B_1\right)\times C\times\mathcal{S}^{(2)}$ is uniformly continuous where $C\subseteq \mathbb{R}^{d_2}$ is a 
 compact set such that for every $n\geq1$, $y_n\in C$. By the above we can conclude that for any $\tilde{\epsilon}>0$, there exists $N_1$ such that 
 for every $n\geq N_1$, for every $(x,s^{(2)})\in(\lambda(y)+\epsilon B_1)\times\mathcal{S}^{(2)}$, $\left|\int_{\mathcal{S}^{(2)}}f(\tilde{s}^{(2)})\Pi^{(2)}(x,y_n,s^{(2)})
 (d\tilde{s}^{(2)})-\int_{\mathcal{S}^{(2)}}f(\tilde{s}^{(2)})\Pi^{(2)}(x,y,s^{(2)})(d\tilde{s}^{(2)})\right|<\tilde{\epsilon}$. Therefore for every $f\in\mathcal{C}(
 \mathcal{S}^{(2)},\mathbb{R})$, there exists $\tilde{N}:=\max\left\{N,N_1\right\}$ such that for every $n\geq \tilde{N}$, 
 \begin{small}
 \begin{align*}
  \bigg|\int_{\mathcal{S}^{(2)}}f(\tilde{s}^{(2)})\nu^n(d\tilde{s}^{(2)})&-\int_{\mathcal{S}^{(2)}}f(\tilde{s}^{(2)})\nu(d\tilde{s}^{(2)})\bigg|\\
  &\leq\\
  \tilde{\epsilon}+
  \bigg|\int_{\mathbb{R}^{d_1}\times\mathcal{S}^{(2)}}\bigg[\int_{\mathcal{S}^{(2)}}f(\tilde{s}^{(2)})\Pi^{(2)}(x,y,s^{(2)})&(d\tilde{s}^{(2)})\bigg]\mu^n(ds,ds^{(2)})-
  \int_{\mathcal{S}^{(2)}}f(\tilde{s}^{(2)})\nu(d\tilde{s}^{(2)})\bigg|.
 \end{align*}
 \end{small}
 The second term in the R.H.S. of the above inequality goes to zero as $n\to\infty$ (use the definition of $\nu(d\tilde{s}^{(2)})$, assumption $(A4)$
 and \cite[Thm.~2.1.1(ii)]{borkarap}). Therefore taking limit in the above equation we get that for any $f\in\mathcal{C}(\mathcal{S}^{(2)},\mathbb{R}
 )$, for every $\tilde{\epsilon}>0$, $\lim_{n\to\infty}\left|\int_{\mathcal{S}^{(2)}}f(\tilde{s}^{(2)})\nu^n(d\tilde{s}^{(2)})-\int_{\mathcal{S}^{(2)}}f(\tilde{s}^{(2)})\nu
 (d\tilde{s}^{(2)})\right|\leq\tilde{\epsilon}$. Hence $\nu^n\to\nu$ in $\mathcal{P}(\mathcal{S}^{(2)})$ as $n\to\infty$. Clearly $\mu^n_{\mathcal{S}^{(2)}}
 \to\mu_{\mathcal{S}^{(2)}}$ as $n\to\infty$. Therefore $\nu=\mu_{\mathcal{S}^{(2)}}$ which gives us that $\mu\in D(y)$.
 \item [(iii)] Follows from part $(i)$ of this lemma.\qed
\end{itemize}
\end{proof}

Define the set-valued map $\hat{H}_2:\mathbb{R}^{d_2}\rightarrow\left\{\text{subsets of }\mathbb{R}^{d_1}\right\}$ such that for every $y\in
\mathbb{R}^{d_2}$, 
\begin{equation}\label{stsrmap}
 \hat{H}_2(y):=\cup_{\mu\in D(y)}\int_{\mathbb{R}^{d_1}\times\mathcal{S}^{(2)}}H_{2,y}(x,s^{(2)})\mu(dx,ds^{(2)}),
\end{equation}
where for every $y\in\mathbb{R}^{d_2}$, $H_{2,y}$ denotes the slice as in Definition \ref{slices}$(iv)$ of the set valued map $H_2$. Since for every 
$y\in\mathbb{R}^{d_2}$, for every $\mu\in D(y)$, $\mathrm{supp}(\mu)$ is compact, by Lemma \ref{msble2}$(iii)$ we know that the 
slices $H_{2,y}$ are $\mu$-integrable for every $\mu\in D(y)$. So the above set-valued map is well defined and we show later that the slower 
timescale iterates track the DI given by,
\begin{equation}\label{stsrdi}
 \frac{dy}{dt}\in\hat{H}_2(y).
\end{equation}

The above DI is guaranteed to have solutions as a consequence of the lemma below.
\begin{lemma}\label{march2}
 The set-valued map $\hat{H}_2:\mathbb{R}^{d_2}\rightarrow\left\{\text{subsets of }\mathbb{R}^{d_2}\right\}$ is a Marchaud map.
\end{lemma}

Proof of the above lemma is given in section \ref{limdiprop}. 

\begin{remark}
In order to understand the DI \eqref{stsrdi} better, we consider the cases where the map $\lambda(\cdot)$ is single-valued and the case where 
Markov noise terms are absent. These special cases also highlight the fact our results are a significant generalization of the results in \cite{arun2t} and \cite{prasan}.  
\begin{itemize}
\item [(1)] When the map $\lambda(\cdot)$ is single-valued, for any $\mu\in D(y)$, since $\mathrm{supp}(\mu_{\mathbb{R}^{d_1}})\subseteq
\lambda(y)$, we have that $\mu_{\mathbb{R}^{d_1}}=\delta_{\lambda(y)}$ where $\delta_{\lambda(y)}\in\mathcal{P}(\mathbb{R}^{d_1})$ denotes 
the Dirac measure at $\lambda(y)$. Therefore the measure $\mu=\delta_{\lambda(y)}\otimes\mu_{\mathcal{S}^{(2)}}$. Since $\mu\in D(y)$, we know that 
for every $A\in\mathscr{B}(\mathcal{S}^{(2)})$, $\mu_{\mathcal{S}^{(2)}}(A)=\int_{\mathbb{R}^{d_1}\times\mathcal{S}^{(2)}}\Pi^{(2)}(x,y,s^{(2)})(A)\mu(dx,ds^{(2)})=
\int_{\mathcal{S}^{(2)}}\left[\int_{\mathbb{R}^{d_1}}\Pi^{(2)}(x,y,s^{(2)})(A)\delta_{\lambda(y)}(dx)\right]\mu_{\mathcal{S}^{(2)}}(ds^{(2)})=\\\int_{\mathcal{S}^{(2)}}
\Pi^{(2)}(\lambda(y),y,s^{(2)})(A)\mu_{\mathcal{S}^{(2)}}(ds^{(2)})$. Thus $\mu_{\mathcal{S}^{(2)}}\in D^{(2)}(\lambda(y),y)$, where $D^{(2)}(\lambda(y),y)$ denotes the set 
of stationary measures of the Markov chain with transition kernel $\Pi^{(2)}(\lambda(y),y,\cdot)(\cdot)$. Therefore for every $y\in\mathbb{R}^{d_2}$, 
\begin{equation*}
 \hat{H}_2(y)=\cup_{\mu\in D(y)}\int_{\mathbb{R}^{d_1}\times\mathcal{S}^{(2)}}H_{2,y}(x,s^{(2)})\mu(dx,ds^{(2)})=\cup_{\nu\in D^{(2)}(\lambda(y),y)}\int_{\mathcal{S}^{(2)}}
 H_{2,(\lambda(y),y)}(s^{(2)})\nu(ds^{(2)}),
\end{equation*}
where $H_{2,(\lambda(y),y)}$ denotes the slice as in Definition \ref{slices}$(i)$ of the set-valued map $H_2$. Therefore DI \eqref{stsrdi} is nothing 
but the set-valued analogue of the slower timescale DI in \cite{prasan}.
\item [(2)] Suppose Markov noise terms are absent (for the analysis and definition of such a recursion see \cite{arun2t}). Then such a recursion 
can be rewritten in the form of recursion \eqref{2tr}, with Markov noise terms taking values in a dummy state space $\mathcal{S}^{(1)}=
\mathcal{S}^{(2)}=\{s^*\}$ with transition laws $\Pi^{(1)}(x,y,s^*)=\Pi^{(2)}(x,y,s^*)=\delta_{s^*}$ for every $(x,y)\in\mathbb{R}^d$. Then it is easy to deduce 
that the stationary distribution maps $D^{(1)}(x,y)=D^{(2)}(x,y)=\delta_{s^*}$ for every $(x,y)\in\mathbb{R}^d$. Then for every $y\in\mathbb{R}^{d_2}$, 
any $\mu\in D(y)$ is of the form $\mu=\nu\otimes\delta_{s^*}$ where $\nu\in\mathcal{P}(\mathbb{R}^{d_1})$ with $\mathrm{supp}(\nu)\subseteq
\lambda(y)$. Then for any $y\in\mathbb{R}^{d_2}$, 
\begin{align*}
 \hat{H}_2(y)=\cup_{\mu\in D(y)}\int_{\mathbb{R}^{d_1}\times\mathcal{S}^{(2)}}H_{2,y}(x,s^{(2)})\mu(dx,ds^{(2)})&=\cup_{\substack{\nu\in\mathcal{P}(\mathbb{R}^{d_1}) 
 \mathrm{supp}(\nu)\subseteq\lambda(y)}}\int_{\mathbb{R}^{d_1}}H_{2,y}(x,s^*)\nu(dx)\\
 &=\bar{co}\left(\cup_{x\in\lambda(y)}H_2(x,y,s^*)\right),
\end{align*}
which is exactly the same slower timescale DI as in \cite{arun2t}. 
\end{itemize}
\end{remark}

Suppose now that the following holds in addition:
\begin{itemize}
 \item [(A10)] DI \eqref{stsrdi} has a globally attracting set $\mathcal{Y}\subseteq\mathbb{R}^{d_2}$,
\end{itemize}
then the main result of this paper states that for almost every $\omega$, as $n\to\infty$,
\begin{equation*}
 \left(
 \begin{array}{c} 
 X_n(\omega)\\
 Y_n(\omega)
 \end{array}
 \right)\to\cup_{y\in\mathcal{Y}}
 \left(\lambda(y)\times\left\{ y\right\}\right).
\end{equation*}

%% file: lim_di_and_their_prop.tex
\section{Mean fields and their properties}
\label{limdiprop}

In this section we prove that for every $y\in\mathbb{R}^{d_2}$, the set-valued map $\hat{H}_1(\cdot,y)$ and the set-valued map $\hat{H}_2(\cdot)$ 
defined in equations \eqref{fstrmap} and \eqref{stsrmap} respectively, are Marchaud maps. 

Recall that by assumptions $(A1)$ and $(A2)$, the set-valued maps $H_1$ and $H_2$ are SAMs. For such set-valued maps, by Lemma \ref{ctem}, we 
know that there exist sequences of continuous set-valued maps, denoted by $\left\{H^{(l)}_1\right\}_{l\geq1}$ and $\left\{H_2^{(l)}\right\}_{l\geq1}$ 
which approximate $H_1$ and $H_2$ respectively. Further by Lemma \ref{param}, these appromating maps admit a continuous parametrization denoted by,
$h_1^{(l)}$ and $h_2^{(l)}$. Throughout this section $\left\{H_1^{(l)}\right\}_{l\geq1}$, $\left\{H_2^{(l)}\right\}$, $\left\{h_1^{(l)}\right\}_{l\geq1}$
and $\left\{h_2^{(l)}\right\}_{l\geq1}$ denote the maps as described above. 

Similar to the definition of the maps $\hat{H}_1$ and $\hat{H}_2$, we define the maps obtained by averaging the set-valued maps $H_1^{(l)}$ and $H_2^{(l)}$ for every $l\geq1$ with respect to measures 
given by the maps $(x,y)\rightarrow D^{(1)}(x,y)$ and $y\rightarrow D(y)$.

\begin{definition}\label{approxavg}
Let the maps $D^{(1)}:\mathbb{R}^d\rightarrow\left\{\text{subsets of }\mathcal{P}(\mathcal{S}^{(1)})\right\}$ and $D:\mathbb{R}^{d_2}\rightarrow\left\{
\text{subsets of }\mathcal{P}(\mathbb{R}^{d_1}\times\mathcal{S}^{(2)})\right\}$ be as in section \ref{recass}. For every $l\geq1$,
\begin{itemize}
  \item [(i)] for every $(x,y)\in\mathbb{R}^{d}$, define $\hat{H}_1^{(l)}:\mathbb{R}^{d}\rightarrow\left\{\text{subsets of }\mathbb{R}^{d_1}
  \right\}$ such that, 
  \begin{equation*}
   \hat{H}_1^{(l)}(x,y):=\cup_{\mu\in D(x,y)}\int_{\mathcal{S}^{(1)}}H_{1,(x,y)}^{(l)}(s^{(1)})\mu(ds^{(1)}),
  \end{equation*}
  where $H_{1,(x,y)}^{(l)}$ denotes the slice (as in Defn.\ref{slices}$(ii)$) of the set-valued map $H_1^{(l)}$,
  \item [(ii)] for every $y\in\mathbb{R}^{d_2}$, define $\hat{H}_2^{(l)}:\mathbb{R}^{d_2}\rightarrow\left\{\text{subsets of }\mathbb{R}^{d_2}
  \right\}$ such that,
  \begin{equation*}
   \hat{H}_2^{(l)}(y):=\cup_{\mu\in D(y)}\int_{\mathcal{R}^{d_1}\times\mathcal{S}^{(2)}}H_{2,y}^{(l)}(x,s^{(2)})\mu(dx,ds^{(2)}),
  \end{equation*}
  where $H_{2,y}^{(l)}$ denotes the slice (as in Defn.\ref{slices}$(v)$) of the set-valued map $H_2^{(l)}$.
\end{itemize}
\end{definition}

In the lemma below we prove that for every $y\in\mathbb{R}^{d_2}$, the maps $\hat{H}_1^{(l)}(\cdot,y)$ and the map $\hat{H}_2^{(l)}(\cdot)$ 
are Marchaud maps for every $l\geq1$.

\begin{lemma}\label{marchapprox1} For every $l\geq1$,
\begin{itemize}
 \item [(i)] the set-valued map $\hat{H}_1^{(l)}:\mathbb{R}^d\rightarrow\left\{\text{subsets of }\mathbb{R}^{d_1}\right\}$ is such that,
 \begin{itemize}
  \item [(a)] for every $(x,y)\in\mathbb{R}^d$, $\hat{H}_1^{(l)}(x,y)$ is a non-empty, convex and compact subset of $\mathbb{R}^{d_1}$,
  \item [(b)] for $K^{(l)}>0$ where $K^{(l)}$ is as in Lemma \ref{ctem}, for every $(x,y)\in\mathbb{R}^d$, $\sup_{x'\in\hat{H}_1^{(l)}(x,y)}
  \left\|x'\right\|\leq K^{(l)}(1+\left\|x\right\|+\left\|y\right\|)$,
  \item [(c)] for every $(x,y)\in\mathbb{R}^d$, for $\mathbb{R}^d$-valued sequence, $\left\{(x_n,y_n)\right\}_{n\geq1}$ converging to $(x,y)\in\mathbb{R}^d$, for every sequence 
  $\left\{x'_n\in\hat{H}_1^{(l)}(x_n,y_n)\right\}$ converging to $x'\in\mathbb{R}^{d_1}$, we have that $x'\in\hat{H}_1^{(l)}(x,y)$.
 \end{itemize}
 \item [(ii)] for every $y\in\mathbb{R}^{d_2}$, the map $\hat{H}_1^{(l)}(\cdot,y)$ is a Marchaud map,
 \item [(iii)] the map $\hat{H}_2^{(l)}(\cdot)$ is a Marchaud map.
 \end{itemize}
\end{lemma}
\begin{proof}
 Fix $l\geq1$.
 \begin{itemize}
  \item [(i)] For every $(x,y)\in\mathbb{R}^d$, by Lemma \ref{msble1}$(iv)$, $H_{1,(x,y)}^{(l)}$ is $\mu$-integrable for every $\mu\in D^{(1)}(x,y)$. 
  Hence for every $(x,y)\in\mathbb{R}^d$, $\hat{H}_1^{(l)}(x,y)$ is non-empty. Let $x^1,x^2\in \hat{H}_1^{(l)}(x,y)$ and $\alpha\in(0,1)$. Then by 
  Lemma \ref{chint}$(i)$, there exist $\nu^1,\nu^2\in\mathcal{P}(\mathcal{S}^{(1)}\times U)$, such that for $i\in\{1,2\}$, $\nu^i_{\mathcal{S}^{(1)}}\in D^{(1)}(x,y)$ and 
  $x^i=\int_{\mathcal{S}^{(1)}\times U}h_{1,(x,y)}^{(l)}(s^{(1)},u)\nu^i(ds^{(1)},du)$ where $U$ denotes the closed unit ball in $\mathbb{R}^{d_1}$. Then 
  $\alpha x^1+(1-\alpha)x^2=\int_{\mathcal{S}^{(1)}\times U}h_{1,(x,y)}^{(l)}(s^{(1)},u)(\alpha\nu^1+(1-\alpha)\nu^2)(ds^{(1)},du)$. Clearly $(\alpha\nu^1+
  (1-\alpha)\nu^2)_{\mathcal{S}^{(1)}}=\alpha\nu^1_{\mathcal{S}^{(1)}}+(1-\alpha)\nu^2_{\mathcal{S}^{(1)}}\in D^{(1)}(x,y)$ where the last inclusion follows from 
  the fact that $D^{(1)}(x,y)$ is a convex subset of $\mathcal{P}(\mathcal{S}^{(1)})$. By Lemma \ref{chint}$(i)$, we get that $\alpha x^1+(1-\alpha)x^2\in
  \hat{H}_1^{(l)}(x,y)$. Therefore $\hat{H}_1^{(l)}(x,y)$ is convex.
  
  By Lemma \ref{msble1}$(ii)$, for every $(x,y)\in\mathbb{R}^{d}$, the set-valued map $H_{1,(x,y)}^{(l)}$ is bounded by $C_{(x,y)}^{(l)}:=
  K^{(l)}(1+\left\|x\right\|+\left\|y\right\|)$. Therefore for every $f\in\mathscr{S}(H_{1,(x,y)})$, for every $s^{(1)}\in\mathcal{S}^{(1)}$, 
  $\left\|f(s^{(1)})\right\|\leq C_{(x,y)}^{(l)}$. Thus for every $x'\in\hat{H}_1^{(l)}(x,y)$, by definition, $x'=\int_{\mathcal{S}^{(1)}}f(s^{(1)})\mu(ds^{(1)})$ 
  for some $f\in\mathscr{S}(H_1^{(l)})$ and some $\mu\in D^{(1)}(x,y)$. Therefore for every $x'\in\hat{H}_1^{(l)}(x,y)$, $\left\| x'\right\|\leq
  \int_{\mathcal{S}^{(1)}}\left\|f(s^{(1)})\right\|\mu(ds^{(1)})\leq C_{(x,y)}^{(l)}=K^{(l)}(1+\left\|x\right\|+\left\|y\right\|)$. 
  
  As a consequence of the arguments in the preceding paragraph for some $(x,y)\in\mathbb{R}^{d}$ in order to show that $\hat{H}_1^{(l)}(x,y)$ is 
  compact, it is enough to show that it is closed. Consider a sequence $\left\{x^n\in\hat{H}_1^{l}(x,y)\right\}_{n\geq1}$ converging to $x^*\in\mathbb{R}
  ^{d_1}$. Then by definition of $\hat{H}_1^{(l)}(x,y)$ and by Lemma \ref{chint}$(i)$, for every $n\geq1$, there exists $\nu^n\in\mathcal{P}(\mathcal{S}
  ^1\times U)$, such that $\nu^n_{\mathcal{S}^{(1)}}\in D^{(1)}(x,y)$ and $x^n=\int_{\mathcal{S}^{(1)}\times U}h^{(l)}_{1,(x,y)}(s^{(1)},u)\nu^n(ds^{(1)},du)$. Since 
  $\mathcal{S}^{(1)}\times U$ is a compact metric space, $\mathcal{P}(\mathcal{S}^{(1)}\times U)$ is compact and hence there exists a subsequence $\left
  \{n_k\right\}_{k\geq1}$ such that $\left\{\nu^{n_k}\right\}_{k\geq1}$ converges to $\nu\in\mathcal{P}(\mathcal{S}^{(1)}\times U)$. Clearly $\left\{
  \nu^{n_k}_{\mathcal{S}^{(1)}}\right\}_{k\geq1}$ converges to $\nu_{\mathcal{S}^{(1)}}$ and by \cite[Thm.~2.1.1(ii)]{borkarap}, $x^{n_k}=\int_{\mathcal{S}^{(1)}\times
  U}h^{(l)}_{1,(x,y)}(s^{(1)},u)\nu^{n_k}(ds^{(1)},du)\to\int_{\mathcal{S}^{(1)}\times U}h^{(l)}_{1,(x,y)}(s^{(1)},u)\nu(ds^{(1)},du)=x^*$. Since for every $k$, $\nu^
  {n_k}_{\mathcal{S}^{(1)}}\in D^{(1)}(x,y)$ and by the fact that $D^{(1)}(x,y)$ is closed we get that, $\nu_{\mathcal{S}^{(1)}}\in D^{(1)}(x,y)$. Therefore 
  $x^*=\int_{\mathcal{S}^{(1)}\times U}h^{(l)}_{1,(x,y)}(s^{(1)},du)\nu(ds^{(1)},du)$ and $\nu_{\mathcal{S}^{(1)}}\in D^{(1)}(x,y)$. Thus $x^*\in\hat{H}_1^{(l)}(x,y)$
  which gives us that $\hat{H}_1^{(l)}(x,y)$ is closed. 
   
  Let $\left\{(x_n,y_n)\right\}_{n\geq1}$ be a sequence converging to $(x,y)$ and let $\left\{x'_n\in \hat{H}_1^{(l)}\right\}_{n\geq1}$ be a 
  sequence converging to $x'$. Then by Lemma \ref{chint}$(i)$, for every $n\geq1$, there exists $\nu^n\in\mathcal{P}(\mathcal{S}^{(1)}\times U)$ such 
  that $\nu^n_{\mathcal{S}^{(1)}}\in D^{(1)}(x_n,y_n)$ and $x'_n=\int_{\mathcal{S}^{(1)}\times U}h_{1,(x,y)}^{(l)}(s^{(1)},u)\nu^n(ds^{(1)},du)$. Since $
  \mathcal{S}^{(1)}\times U$ is a compact metric space, $\mathcal{P}(\mathcal{S}^{(1)}\times U)$ is a compact metric space and hence there exists a 
  subsequence say $\left\{n_k\right\}_{k\geq1}$ such that $\left\{\nu^{n_k}\right\}_{k\geq1}$ converges to $\nu\in \mathcal{P}(\mathcal{S}^{(1)}\times 
  U)$. Clearly $\nu^{n_k}_{\mathcal{S}^{(1)}}\to\nu_{\mathcal{S}^{(1)}}$ in $\mathcal{P}(\mathcal{S}^{(1)})$ and by closed graph property of the map $(x,y)
  \rightarrow D^{(1)}(x,y)$, we have that $\nu_{\mathcal{S}^{(1)}}\in D^{(1)}(x,y)$. Using the continuity of the map $h^{(l)}_1(\cdot)$ it is easy to show 
  that $\lim_{k\to\infty}\sup_{(s^{(1)},u)\in\mathcal{S}^{(1)}\times U}\left\|h^{(l)}_{1,(x_{n_k},y_{n_k})}(s^{(1)},u)-h^{(l)}_{1,(x,y)}(s^{(1)},u)\right\|=0$. 
  Then, $\|x'-\int_{\mathcal{S}^{(1)}\times U}h^{(l)}_{1,(x,y)}(s^{(1)},u)\nu(ds^{(1)},du)\|\leq\|x'-\int_{\mathcal{S}^{(1)}\times U}
  h^{(l)}_{1,(x_{n_k},y_{n_k})}(s^{(1)},u)\nu^{n_k}(ds^{(1)},du)\|+\\\int_{\mathcal{S}^{(1)}\times U}\|h^{(l)}_{1,(x_{n_k},y_{n_k})}(s^{(1)},u)-
  h^{(l)}_{1,(x,y)}(s^{(1)},u)\|\nu^{n_k}(ds^{(1)},du)+\|\int_{\mathcal{S}^{(1)}\times U}h^{(l)}_{1,(x,y)}(s^{(1)},u)\nu
  ^{n_k}(ds^{(1)},du)-\\\int_{\mathcal{\mathcal{S}^{(1)}}\times U}h^{(l)}_{1,(x,y)}(s^{(1)},u)\nu(ds^{(1)},du)\|$. Now use \cite[Thm.~2.1.1(ii)]{borkarap} in the 
  above inequality to obtain $\lim_{k\to\infty}\|x'-\int_{\mathcal{S}^{(1)}\times U}h^{(l)}_{1,(x,y)}(s^{(1)},u)\nu(ds^{(1)},du)\|=0$. Then Lemma \ref{chint}$(i)$ 
  gives us that $x'\in\hat{H}^{(l)}_1(x,y)$.
  \item [(ii)] Follows from part $(i)$ of this lemma.
  \item [(iii)] Proof is similar to part $(i)$ of this lemma with minor modifications. First modification is the use of Lemma \ref{chint}$(ii)$ 
  instead of Lemma \ref{chint}$(i)$. For example in order to show that $\hat{H}_2^{(l)}(y)$ is closed for some $y\in\mathbb{R}^{d_2}$, 
  fix sequence $\left\{y'_n\right\}_{n\geq1}\subseteq \hat{H}_2^{(l)}(y)$ converging to $y'$. Use Lemma \ref{chint}$(ii)$ and the definition of 
  $\hat{H}_2^{(l)}(y)$, to obtain $\left\{\nu^n\right\}_{n\geq1}\subseteq\mathcal{P}(\mathbb{R}^{d_1}\times\mathcal{S}^{(2)}\times U)$ where $U$ 
  denotes the closed unit ball in $\mathbb{R}^{d_2}$ and the sequence $\left\{\nu^n\right\}_{n\geq1}$ is such that for every $n\geq1$, 
  $\nu^n_{\mathbb{R}^{d_1}\times\mathcal{S}^{(2)}}\in D(y)$ and $y'_n=\int_{\mathbb{R}^{d_1}\times\mathcal{S}^{(2)}\times U}h^{(l)}_{2,y}(x,s^{(2)},u)\nu^n(dx,
  ds^{(2)},du)$. By definition of $D(y)$, for every $n\geq1$, $\mathrm{supp}(\nu^n_{\mathbb{R}^{d_1}\times\mathcal{S}^{(2)}})\subseteq\lambda(y)\times
  \mathcal{S}^{(2)}$ and hence $\mathrm{supp}(\nu^n)\subseteq\lambda(y)\times\mathcal{S}^{(2)}\times U$ which is a compact subset of $\mathbb{R}^{d_1}\times
  \mathcal{S}^{(2)}\times U$. Now by Prohorov's theorem the sequence $\left\{\nu^n\right\}_{n\geq1}$ is a relatively compact subset of 
  $\mathcal{P}(\mathbb{R}^{d_1}\times\mathcal{S}^{(2)}\times U)$ and hence has a convergent subsequence. By Lemma \ref{measmap}$(i)$, $D(y)$ is 
  compact and hence every limit point of $\left\{\nu^n_{\mathbb{R}^{d_1}\times\mathcal{S}^{(2)}}\right\}$ is in $D(y)$. The rest of the argument is 
  same as the corresponding in part $(i)$ of this lemma. 
  
  In order to show that $\hat{H}_2^{(l)}(\cdot)$ has a closed graph, fix sequences $\left\{y_n\right\}_{n\geq1}$ converging to $y$ and 
  $\left\{y'_n\in \hat{H}_2^{(l)}(y_n)\right\}_{n\geq1}$ converging to $y'$. Use Lemma \ref{chint}$(ii)$, to obtain $\left\{\nu^n\right\}_{n\geq1}
  \subseteq\mathcal{P}(\mathbb{R}^{d_1}\times\mathcal{S}^{(2)}\times U)$ such that for every $n\geq1$, $\nu^n_{\mathbb{R}^{d_1}\times\mathcal{S}^{(2)}}\in 
  D(y_n)$ and $y'_n:=\int_{\mathbb{R}^{d_1}\times\mathcal{S}^{(2)}\times U}h^{(l)}_{2,y}(x,s^{(2)},u)\nu^n(dx,ds^{(2)},du)$. Then for every $n\geq1$, 
  $\mathrm{supp}(\nu^n)\subseteq\lambda(y_n)\times\mathcal{S}^{(2)}\times U$. By assumption $(A9)$, for any $\delta>0$, the set 
  $L:=\left\{x\in\lambda(\tilde{y}): \|\tilde{y}-y\|\leq\delta \right\}$ is a compact subset of $\mathbb{R}^{d_1}$. Therefore there exists $N$ 
  large such that for every $n\geq N$, $\mathrm{supp}(\nu^n)\subseteq L\times\mathcal{S}^{(2)}\times U$. By Prohorov's theorem the sequence of measures $\left\{
  \nu^n\right\}_{n\geq N}$ is tight and has a convergent subsequence. Clearly by Lemma \ref{measmap}$(ii)$, every limit point of $\left\{
  \nu^n_{\mathbb{R}^{d_1}\times\mathcal{S}^{(2)}}\right\}_{n\geq N}$ is in $D(y)$. Now the rest of the argument is same as the corresponding in part 
  $(i)$ of this lemma.\qed
 \end{itemize}
\end{proof}

By Lemma \ref{ctem} we know that for every $l\geq1$, for every $(x,y)\in\mathbb{R}^d$, $H_1(x,y)\subseteq H_1^{(l+1)}(x,y)\subseteq 
H_1^{(l)}(x,y)$ (similarly $H_2(x,y)\subseteq H^{(l+1)}_2(x,y)\subseteq H^{(l)}_2(x,y)$). The next lemma states that the above is true for 
$\hat{H}_i$ and $\hat{H}_i^{(l)}$ as well, for every $i\in\left\{1,2\right\}$.

\begin{lemma}\label{rel}: 
 \begin{itemize}
  \item [(i)] For every $l\geq1$, for every $(x,y)\in\mathbb{R}^d$, $\hat{H}_1(x,y)\subseteq \hat{H}_1^{(l+1)}(x,y)\subseteq \hat{H}_1^{(l)}(x,y)$.
  \item [(ii)] For every $l\geq1$, for $y\in\mathbb{R}^{d_2}$, $\hat{H}_2(y)\subseteq \hat{H}_2^{(l+1)}(y)\subseteq \hat{H}_2^{(l)}(y)$.
  \item [(iii)] For every $(x,y)\in\mathbb{R}^d$, $\cap_{l\geq1}\hat{H}_1^{(l)}(x,y)=\cup_{\mu\in D^{(1)}(x,y)}\cap_{l\geq1}\int_{\mathcal{S}^{(1)}}
  H_{1,(x,y)}^{(l)}(s^{(1)})\mu(ds^{(1)})$.
  \item [(iv)] For every $y\in\mathbb{R}^{d_2}$, $\cap_{l\geq1}\hat{H}_2^{(l)}(y)=\cup_{\mu\in D(y)}\cap_{l\geq1}\int_{\mathbb{R}^{d_2}\times
  \mathcal{S}^{(2)}}H_{2,y}^{(l)}(x,s^{(2)})\mu(dx,ds^{(2)})$.
  \item [(v)] For every $(x,y)\in\mathbb{R}^d$, $\hat{H}_1(x,y)=\cap_{l\geq1}\hat{H}_1^{(l)}(x,y)$.
  \item [(vi)] For every $y\in\mathbb{R}^{d_2}$, $\hat{H}_2(y)=\cap_{l\geq1}\hat{H}_2^{(l)}(y)$.
 \end{itemize}
\end{lemma}
\begin{proof}
 The proofs of parts $(i)$ and $(ii)$ follow directly from the definition of $\hat{H}_i,\ \hat{H}_i^{(l)}$ for every $l\geq1$ and the fact that 
 for every $i\in\left\{1,2\right\}$, $H_i(x,y)\subseteq H_i^{(l+1)}(x,y)\subseteq H_i^{(l)}(x,y)$ for every $(x,y)\in\mathbb{R}^d$.
 The proof of part $(iii)$ is similar to part $(iv)$ and we present the proof of part $(iv)$ below (the proof of part $(iii)$ is in fact the same as that of 
 \cite[Lemma~4.4(ii)]{vin1t}).
 \begin{itemize}
  \item [(iv)] Fix $y\in\mathbb{R}^{d_2}$. Then by definition of $\hat{H}_2^{(l)}(y)$, we have that for every $l\geq1$, for any $\mu\in D(y)$, 
  $\int_{\mathbb{R}^{d_1}\times\mathcal{S}^{(2)}}H_{2,y}^{(l)}(x,s^{(2)})\mu(dx,ds^{(2)})\subseteq\hat{H}_2^{(l)}(y)$. Therefore, $\cup_{\mu\in D(y)}
  \cap_{l\geq1}\int_{\mathbb{R}^{d_1}\times\mathcal{S}^{(2)}}H_{2,y}^{(l)}(x,s^{(2)})\mu(dx,ds^{(2)})\subseteq \cap_{l\geq1}\hat{H}_2^{(l)}(y)$.
  
  Let $y'\in\cap_{l\geq1}\hat{H}_2^{(l)}(y)$. Then for every $l\geq1$, there exists $\mu^l\in D(y)$ such that \begin{small}$y'\in\int_{\mathbb{R}^{d_1}\times
  \mathcal{S}^{(2)}}H_{2,y}^{(l)}(x,s^{(2)})\mu^{l}(dx,ds^{(2)})$\end{small}. Since $\left\{\mu^l\right\}_{l\geq1}$ is a subset of $D(y)$, for every $l\geq1$, 
  $\mathrm{supp}(\mu^l)\subseteq \lambda(y)\times\mathcal{S}^{(2)}$. Hence the sequence of probability measures $\left\{\mu^l\right\}_{l\geq1}$ is 
  tight and by Prohorov's theorem has a limit say $\mu^*\in\mathcal{P}(\mathbb{R}^{d_1}\times\mathcal{S}^{(2)})$. Let $\left\{l_k\right\}_{k\geq1}$ 
  be a subsequence such that $\mu^{l_k}\to\mu^*$ as $k\to\infty$ and by Lemma \ref{measmap}$(i)$ we know that $D(y)$ is compact which gives us 
  $\mu^*\in D(y)$. Since for every $l\geq1$, for every $k$ such that $l_k\geq l$, $\mathscr{S}(H^{(l_k)}_{2,y})\subseteq\mathscr{S}(H^{(l)}_{2,y})$ 
  we get that for every $l\geq1$, for every $k$ such that $l_k\geq l$, $\int_{\mathbb{R}^{d_1}\times\mathcal{S}^{(2)}}H^{(l)}_{2,y}(x,s^{(2)})\mu^{l_k}
  (dx,ds^{(2)})=y'$. For every $l\geq1$, by Lemma \ref{chint}$(ii)$, we know that for every $k$ such that $l_k\geq l$, there exists $\nu^{(l,l_k)}
  \in\mathcal{P}(\mathbb{R}^{d_1}\times\mathcal{S}^{(2)}\times U)$ ($U$ denotes the closed unit ball in $\mathbb{R}^{d_2}$) such that 
  $y'=\int_{\mathbb{R}^{d_1}\times\mathcal{S}^{(2)}\times U}h^{(l)}_{2,y}(x,s^{(2)},u)\nu^{(l,l_k)}(dx,ds^{(2)},du)$ and $\nu^{(l,l_k)}_{\mathbb{R}^{d_1}\times
  \mathcal{S}^{(2)}}=\mu^{l_k}$. Further, for every $l\geq1$, for every $k$ such that $l_k\geq l$, $\mathrm{supp}(\nu^{(l,l_k)})\subseteq \lambda(y)
  \times\mathcal{S}^{(2)}\times U$ and hence $\left\{\nu^{(l,l_k)}\right\}_{k:l_k\geq l}$ is tight and by Prohorov's theorem has a convergent 
  subsequence. For every $l\geq 1$, let $\nu^{(l)}$ denote a limit point of the sequence $\left\{\nu^{(l,l_k)}\right\}_{k:l_k\geq l}$. Since for 
  every $l\geq1$, $\left\{\nu^{(l,l_k)}_{\mathbb{R}^{d_1}\times\mathcal{S}^{(2)}}=\mu^{l_k}\right\}_{k:l_k\geq l}$ and $\mu^{l_k}\to\mu^*$ as 
  $k\to\infty$ we have that $\nu^{(l)}_{\mathbb{R}^{d_1}\times\mathcal{S}^{(2)}}=\mu^*\in D(y)$. By \cite[Thm.~2.1.1(ii)]{borkarap}, for every $l\geq1$
  $y'=\int_{\mathbb{R}^{d_1}\times\mathcal{S}^{(2)}\times U}h^{(l)}_{2,y}(x,s^{(2)},u)\nu^{(l)}(dx,ds^{(2)},du)$ and hence by Lemma \ref{chint}$(ii)$, 
  $y'\in\int_{\mathbb{R}^{d_1}\times\mathcal{S}^{(2)}}H^{(l)}_{2,y}(x,s^{(2)})\mu^*(dx,ds^{(2)})$ where $\mu^*\in D(y)$. Therefore there exists $\mu^*\in 
  D(y)$ such that for every $l\geq1$, $y'\in\int_{\mathbb{R}^{d_1}\times\mathcal{S}^{(2)}}H^{(l)}_{2,y}(x,s^{(2)})\mu^*(dx,ds^{(2)})$. Hence 
  $y'\in\cup_{\mu\in D(y)}\cap_{l\geq1}\int_{\mathbb{R}^{d_1}\times\mathcal{S}^{(2)}}H^{(l)}_{2,y}(x,s^{(2)})\mu(dx,ds^{(2)})$.
 \end{itemize}
The proof of part $(v)$ is similar to the proof of part $(vi)$ and we present a proof of part $(vi)$ below (the proof of part $(v)$ is exactly the same 
as that of \cite[Lemma~4.4(iii)]{vin1t})
\begin{itemize}
 \item [(vi)] From part $(ii)$ of this lemma we have that, for every $y\in\mathbb{R}^{d_2}$, $\hat{H}_2(y)\subseteq \cap_{l\geq1}\hat{H}_2^{(l)}
 (y)$. 
 
 Fix $y\in\mathbb{R}^{d_2}$ and $\mu\in D(y)$. Let $y'\in\cap_{l\geq1}\int_{\mathbb{R}^{d_1}\times\mathcal{S}^{(2)}}H^{(l)}_{2,y}(x,s^{(2)})\mu(dx,ds^{(2)})$. Then for every $l\geq1$, 
 there exists $f^{(l)}\in\mathscr{S}(H^{(l)}_{2,y})$ such that $y'=\int_{\mathbb{R}^{d_1}\times\mathcal{S}^{(2)}}f^{(l)}(x,s^{(2)})\mu(dx,ds^{(2)})$. 
 Let $d(\tilde{y},A):=\inf\left\{\left\|\tilde{y}-z\right\|:
 z\in A\right\}$ for every $\tilde{y}\in\mathbb{R}^{d_2}$ and for every $A\subseteq \mathbb{R}^{d_2}$ compact. By Lemma \ref{cldint}, we have 
 that $\int_{\mathbb{R}^{d_1}\times\mathcal{S}^{(2)}}H_{2,y}(x,s^{(2)})\mu(dx,ds^{(2)})$ is compact and convex. Then, 
 \begin{align*}
  d(y',\int_{\mathbb{R}^{d_1}\times\mathcal{S}^{(2)}}H_{2,y}(x,s^{(2)})\mu(dx,ds^{(2)}))&=\inf\left\{\left\|y'-z\right\|:z\in\int_{\mathbb{R}^{d_1}\times
  \mathcal{S}^{(2)}}H_{2,y}(x,s^{(2)})\mu(dx,ds^{(2)})\right\}\\
  &=\inf_{f\in\mathscr{S}(H_{2,y})}\left\|y'-\int_{\mathbb{R}^{d_1}\times\mathcal{S}^{(2)}}f(x,s^{(2)})\mu(dx,ds^{(2)})\right\|\\
  &=\inf_{f\in\mathscr{S}(H_{2,y})}\left\|\int_{\mathbb{R}^{d_1}\times\mathcal{S}^{(2)}}\left(f^{(l)}(x,s^{(2)})-f(x,s^{(2)})\right)\mu(dx,ds^{(2)})\right\|\\
  &\leq\inf_{f\in\mathscr{S}(H_{2,y})}\int_{\mathbb{R}^{d_1}\times\mathcal{S}^{(2)}}\left\|f^{(l)}(x,s^{(2)})-f(x,s^{(2)})\right\|\mu(dx,ds^{(2)})\\
  &=\int_{\mathbb{R}^{d_1}\times\mathcal{S}^{(2)}}\inf\left\{\left\|f^{(l)}(x,s^{(2)})-\tilde{y}\right\|:\tilde{y}\in H_{2,y}(x,s^{(2)})\right\}\mu(dx,ds^{(2)}),
 \end{align*}
where the last equality follows from \cite[Lemma~1.3.12]{shoumei}. By \cite[Lemma~3.7]{vin1t}, we know that for every $l\geq1$, the map 
$(x,s^{(2)})\rightarrow d(f^{(l)}(x,s^{(2)}),H_{2,y}(x,s^{(2)}))$ is measurable and from the last equality it follows that for every $l\geq1$,
\begin{equation*}
d(y',\int_{\mathbb{R}^{d_1}\times\mathcal{S}^{(2)}}H_{2,y}(x,s^{(2)})\mu(dx,ds^{(2)}))\leq \int_{\mathbb{R}^{d_1}\times\mathcal{S}^{(2)}}d(f^{(l)}(x,s^{(2)}),
H_{2,y}(x,s^{(2)}))\mu(dx,ds^{(2)}).
\end{equation*}
By observation $(2)$ stated after Lemma \ref{ctem} we have that for every $(x,s^{(2)})\in\mathbb{R}^{d_2}\times\mathcal{S}^{(2)}$, 
$\\\lim_{l\to\infty}d(f^{(l)}(x,s^{(2)}),H_{2,y}(x,s^{(2)}))=0$. Since $\mu\in D(y)$, $\mathrm{supp}(\mu)\subseteq \lambda(y)\times\mathcal{S}^{(2)}$, for 
every $l\geq1$,
\begin{equation*}
 \int_{\mathbb{R}^{d_1}\times\mathcal{S}^{(2)}}d(f^{(l)}(x,s^{(2)}),H_{2,y}(x,s^{(2)}))\mu(dx,ds^{(2)})=\int_{\lambda(y)\times\mathcal{S}^{(2)}}d(f^{(l)}(x,s^{(2)}),
 H_{2,y}(x,s^{(2)})\mu(dx,ds^{(2)})).
\end{equation*}
Since $\lambda(y)$ is compact, there exists $M>0$ such that for every $x\in\lambda(y)$, $\|x\|\leq M$.
By Lemma \ref{msble2}$(ii)$, $(A2)$ and observation $(1)$ stated below Lemma \ref{ctem}, we have that for every $l\geq1$, for every $(x,s^{(2)})\in
\lambda(y)\times\mathcal{S}^{(2)}$,  $d(f^{(l)}(x,s^{(2)}),H_{2,y}(x,s^{(2)}))\leq (K_{y}^{(l)}+K)(1+\|x\|)\leq (\max\{\tilde{K},\tilde{K}\|y\|\}+K)(1+M)$.
By bounded convergence theorem we have,
\begin{equation*}
 d(y',\int_{\mathbb{R}^{d_1}\times\mathcal{S}^{(2)}}H_{2,y}(x,s^{(2)})\mu(dx,ds^{(2)}))\leq\lim_{l\to\infty}\int_{\mathbb{R}^{d_1}\times\mathcal{S}^{(2)}}
 d(f^{(l)}(x,s^{(2)}),H_{2,y}(x,s^{(2)}))\mu(dx,ds^{(2)})=0.
\end{equation*}
Therefore $d(y',\int_{\mathbb{R}^{d_1}\times\mathcal{S}^{(2)}}H_{2,y}(x,s^{(2)})\mu(dx,ds^{(2)}))=0$ and by Lemma \ref{cldint}, we know that $\\\int_{\mathbb{R}^{d_1}
\times\mathcal{S}^{(2)}}H_{2,y}(x,s^{(2)})\mu(dx,ds^{(2)})$ is a closed subset of $\mathbb{R}^{d_2}$. Hence $y'\in\int_{\mathbb{R}^{d_1}\times\mathcal{S}^{(2)}}H_{2,y}
(x,s^{(2)})\mu(dx,ds^{(2)})$.

From the arguments in the preceding paragraph, we have that for every $y\in\mathbb{R}^{d_2}$, for every $\mu\in D(y)$, 
$\cap_{l\geq1}\int_{\mathbb{R}^{d_1}\times\mathcal{S}^{(2)}}H^{(l)}_{2,y}(x,s^{(2)})\mu(dx,ds^{(2)})\subseteq\int_{\mathbb{R}^{d_1}\times\mathcal{S}^{(2)}}H_{2,y}
(x,s^{(2)})\mu(dx,ds^{(2)})$. Thus for every $y\in\mathbb{R}^{d_2}$, 
\begin{equation*}
\cup_{\mu\in D(y)}\cap_{l\geq1}\int_{\mathbb{R}^{d_1}\times\mathcal{S}^{(2)}}H^{(l)}_{2,y}(x,s^{(2)})\mu(dx,ds^{(2)})\subseteq
\cup_{\mu\in D(y)}\int_{\mathbb{R}^{d_1}\times\mathcal{S}^{(2)}}H_{2,y}(x,s^{(2)})\mu(dx,ds^{(2)})=\hat{H}_2(y).
\end{equation*}
By part $(iv)$ of this lemma we get that for every $y\in\mathbb{R}^{d_2}$, $\cap_{l\geq1}\hat{H}_2^{(l)}(y)\subseteq \hat{H}_2(y)$.\qed
\end{itemize}
\end{proof}

\begin{lemma}\label{march3}
 The set-valued map $\hat{H}_1:\mathbb{R}^d\rightarrow\left\{\text{subsets of }\mathbb{R}^{d_1}\right\}$ as defined in eqn.\eqref{fstrmap} is such 
 that,
 \begin{itemize}
  \item [(i)] for every $(x,y)\in\mathbb{R}^d$, $\hat{H}_1(x,y)$ is a non-empty, convex and compact subset of $\mathbb{R}^{d_1}$,
  \item [(ii)] there exists $K>0$ (same as in $(A1)(ii)$), such that for every $(x,y)\in\mathbb{R}^d$, $\sup_{x'\in\hat{H}_1(x,y)}\|x'\|\leq K(1+\|x\|+
  \|y\|)$,
  \item [(iii)] for every $(x,y)\in\mathbb{R}^{d}$, for every $\mathbb{R}^d$-valued sequence, $\{(x_n,y_n)\}_{n\geq1}$ converging to $(x,y)$ and for every $\{x'_n\in \hat{H}_1
  (x_n,y_n)\}_{n\geq1}$ converging to $x'$, we have $x'\in\hat{H}_1(x,y)$. 
 \end{itemize}
\end{lemma}
\begin{proof}
 \begin{itemize}
 \item [(i)] Fix $(x,y)\in\mathbb{R}^d$. By Lemma \ref{msble1}(iii) $H_{1,(x,y)}$ is $\mu$-integrable for every $\mu\in D^{(1)}(x,y)$. 
 Hence $\hat{H}_{1}(x,y)$ is non-empty. For every $l\geq1$, by Lemma \ref{marchapprox1}$(i)(a)$ we know that $\hat{H}^{(l)}_1(x,y)$ is 
 convex and compact subset of $\mathbb{R}^{d_1}$. By Lemma \ref{rel}$(v)$ we have that $\hat{H}_1(x,y)=\cap_{l\geq1}\hat{H}^{(l)}_1
 (x,y)$ and hence $\hat{H}_1(x,y)$ is convex and compact.
 
 \item [(ii)] Fix $(x,y)\in\mathbb{R}^d$. For any $x'\in\hat{H}_1(x,y)$, there exists $\mu\in D^{(1)}(x,y)$ and $f\in\mathscr{S}(H_{1,(x,y)})$ 
 such that $x'=\int_{\mathcal{S}^{(1)}}f(s^{(1)})\mu(ds^{(1)})$. By Lemma \ref{msble1}$(i)$, we know that for every $s^{(1)}\in\mathcal{S}^{(1)}$, $\|f(s^{(1)})\|
 \leq C_{(x,y)}=K(1+\|x\|+\|y\|)$. Therefore $\|x'\|=\|\int_{\mathcal{S}^{(1)}}f(s^{(1)})\mu(ds^{(1)})\|\leq C_{(x,y)}=K(1+\|x\|+\|y\|)$.
 
 \item [(iii)] Let $\{(x_n,y_n)\}_{n\geq1}$ be a sequence converging to $(x,y)$ and $\{x'_n\in\hat{H}_1(x_n,y_n)\}_{n\geq1}$ be a sequence 
 converging to $x'$. Then by Lemma \ref{rel}$(v)$, we have that for every $l\geq1$, for every $n\geq1$, $x'_n\in\hat{H}^{(l)}_1(x_n,y_n)$. 
 By Lemma \ref{marchapprox1}$(i)(c)$, we have that for every $l\geq1$, $x'\in\hat{H}^{(l)}_1(x,y)$. Thus by Lemma \ref{rel}$(v)$, we have $x'
 \in\hat{H}_1(x,y)$.\qed
 \end{itemize}
\end{proof}

Lemma \ref{march1} is an immediate consequence of the above lemma. Similarly, the proof of Lemma \ref{march2} follows from the fact that $\{\hat{H}^{(l)}_2\}_{l\geq1}$ are Marchaud maps (see Lemma \ref{marchapprox1}$(iii)$) which approximate $\hat{H}_2$ (see Lemma \ref{rel}$(vi)$) and the linear growth property of the map $\lambda(\cdot)$ (that is $(A9(i))$).

%% file: recursion_analysis.tex
\section{Recursion analysis}
\label{recanal}
In this section we present the analysis of recursion \eqref{2tr}. The analysis comprises of two parts.

The first part deals with the analysis of the faster timescale recursion where we show that the faster timescale iterates $\{X_n\}_{n\geq1}$ converge 
almost surely to $\lambda(y)$ (as in $(A9)$) for some $y\in\mathbb{R}^{d_2}$.

The second part deals with the slower timescale recursion analysis where we show that the slower timescale iterates $\{Y_n\}_{n\geq1}$ track the flow 
of DI \eqref{stsrdi}.

Throughout this section we assume that assumptions $(A1)-(A9)$ are satisfied.

\input{faster_timescale_rec_anal}

\input{slower_timescale_rec_anal}

%% file: faster_timescale_rec_anal.tex
\subsection{Faster timescale recursion analysis}
\label{fstranal}
For every $\omega\in\Omega$, for every $n\geq0$, the two timescale recursion \eqref{2tr} can be written as,
\begin{subequations}\label{2tr1}
\begin{align}
 Y_{n+1}(\omega)-Y_n(\omega)-b(n)M^{(2)}_{n+1}(\omega)&=b(n)V^2_n(\omega),\\
 X_{n+1}(\omega)-X_n(\omega)-a(n)M^{(1)}_{n+1}(\omega)&=a(n)V^1_n(\omega),
\end{align}
\end{subequations}
where for every $n\geq0$, $V^1_n$ and $V^2_n$ are such that, for every $\omega\in\Omega$,
\begin{align*}
 V^1_n(\omega)&\in H_1(X_n(\omega),Y_n(\omega),S^{(1)}_n(\omega)),\\
 V^2_n(\omega)&\in H_2(X_n(\omega),Y_n(\omega),S^{(2)}_n(\omega)).
\end{align*}
The recursion \eqref{2tr1} can be rewritten as,  
\begin{subequations}
 \begin{align*}
  Y_{n+1}(\omega)-Y_n(\omega)&=a(n)\left(\frac{b(n)}{a(n)}V^2_n(\omega)+\frac{b(n)}{a(n)}M^{(2)}_{n+1}(\omega)\right),\\
  X_{n+1}(\omega)-X_n(\omega)&=a(n)\left(V^1_n(\omega)+M^{(1)}_{n+1}(\omega)\right),
 \end{align*}
\end{subequations}
for every $\omega\in\Omega$ and for every $n\geq0$. The above can be now written in the form of the single timescale recursion (that is \eqref{str}):
\begin{equation}\label{fstrstsr}
 Z_{n+1}-Z_{n}-a(n)M_{n+1}\in a(n)F(Z_n,S^{(1)}_n),
\end{equation}
where,
\begin{itemize}
 \item [(1)] for every $n\geq0$, $Z_n=\left(X_n,Y_n\right)$,
 \item [(2)] for $n\geq0$, $M_{n+1}=\left(M^{(1)}_{n+1},\frac{b(n)}{a(n)}\left(V^2_n+M^{(2)}_{n+1}\right)\right)$,
 \item [(3)] $F:\mathbb{R}^d\times\mathcal{S}^{(1)}\rightarrow\left\{\text{subsets of }\mathbb{R}^d\right\}$ such that for every $(x,y,s^{(1)})\in
 \mathbb{R}^d\times\mathcal{S}^{(1)}$, $F(x,y,s^{(1)})=(H_1(x,y,s^{(1)}),0)$.
\end{itemize}

We now show that the quantities defined above satisfy the assumptions associated with the single timescale recursion as in section 
\ref{stsri}. Clearly by assumption $(A5)$, the step size sequence $\left\{a(n)\right\}_{n\geq0}$ satisfies assumption $S(A3)$ and by 
assumption $(A3)$ the Markov noise terms, $\left\{S^{(1)}_n\right\}_{n\geq0}$ satisfy assumption $S(A2)$. As a consequence of the stability 
assumption $(A8)$, we have that $\mathbb{P}(\sup_{n\geq0}\left\|Z_n:=(X_n,Y_n)\right\|<\infty)=1$ and hence assumption $S(A5)$ is satisfied.

Consider the set-valued map $F$ defined above. Clearly by assumption $(A1)(i)$, for every $(x,y,s^{(1)})\in\mathbb{R}^d\times\mathcal{S}^{(1)}$, 
$F(x,y,s^{(1)})$ is a non-empty, convex and compact subset of $\mathbb{R}^d$. Further by assumption $(A1)(ii)$, we have that for every $(x,y,s^{(1)})
\in\mathbb{R}^d\times\mathcal{S}^{(1)}$, $\sup_{(x',y')\in F(x,y,s^{(1)})}\left\|(x',y')\right\|=\sup_{x'\in H_1(x,y,s^{(1)})}\left\|x'\right\|\leq K(1+
\left\|x\right\|+\left\|y\right\|)\leq \max\left\{K,KC\right\}(1+\left\|(x,y)\right\|)$ where $C>0$ is such that $\|x\|+\|y\|\leq C\|(x,y)\|$ 
for every $(x,y)\in\mathbb{R}^d$ (see \cite[Thm.~4.3.26]{kumar}). By assumption $(A1)(iii)$, the map $H_1$ has a closed graph and hence the map 
$F$ also has a closed graph. Therefore the set-valued map satisfies assumption $S(A1)$.

Recall that for every $T>0$, for every $n\geq0$, $\tau^1(n,T):=\min\left\{m>n:\sum_{k=n}^{m-1}a(k)\geq T\right\}$. Let, 
\begin{equation}\label{omega1}
\Omega_1:=\left\{\omega\in\Omega: (A6), (A7)\ and\ (A8)\ hold\right\}.
\end{equation}
It is clear that $\mathbb{P}(\Omega_1)=1$. Let $\omega\in\Omega_1$ and fix $T>0$. For any $n\geq0$, 
\begin{small}
\begin{equation*}
 \sup_{n\leq k\leq \tau^1(n,T)}\left\|\sum_{m=n}^ka(m)M_{m+1}(\omega)\right\|\leq \sup_{n\leq k\leq\tau^1(n,T)}\left\|\sum_{m=n}^k
                    a(m)M^{(1)}_{m+1}(\omega)\right\|+\sup_{n\leq k\leq\tau^1(n,T)}\left\|\sum_{m=n}^ka(m)\frac{b(m)}{a(m)}\left(V^2_n(\omega)+M^{(2)}_{m+1}(\omega)
                    \right)\right\|.
\end{equation*}
\end{small}
 By assumption $(A8)$, for every $\omega\in\Omega_1$, 
there exists $r>0$, such that $\sup_{n\geq0}\left(\|X_n(\omega)\|+\|Y_n(\omega)\|\right)\leq r$. Since for every $n\geq0$, for every $\omega\in
\Omega_1$, $V^2_n(\omega)\in H_2(X_n(\omega),Y_n(\omega),S^{(2)}_n(\omega))$, by assumption $(A2)(ii)$ we have that, $\|V^2_n(\omega)\|\leq 
K(1+\|X_n(\omega)\|+\|Y_n(\omega)\|)\leq K(1+r)=:R<\infty$. Further by assumption $(A5)(iii)$ for every $0<\epsilon<T$, there exists $N$ such that 
for every $n\geq N$, $b(n)\leq \frac{\epsilon}{T+1}a(n)$. Therefore for every $n\geq N$, for every $m>n$, $\sum_{k=n}^{m-1}b(k)\leq\frac{\epsilon}
{T+1}\sum_{k=n}^{m-1}a(k)$. Thus for every $n\geq N$, $\sum_{k=n}^{\tau^1(n,T)-1}b(k)\leq\epsilon$ and $\tau^1(n,T)\leq\tau^2(n,T)$. 
Therefore for every $0<\epsilon<T$, for every $n\geq N$, 
\begin{small}
\begin{equation*}
 \sup_{n\leq k\leq \tau^1(n,T)}\left\|\sum_{m=n}^ka(m)M_{m+1}(\omega)\right\|\leq \sup_{n\leq k\leq\tau^1(n,T)}\left\|\sum_{m=n}^k
                    a(m)M^{(1)}_{m+1}(\omega)\right\|+\epsilon R+\sup_{n\leq k\leq\tau^2(n,T)}\left\|\sum_{m=n}^k b(m)M^{(2)}_{m+1}(\omega)
                    \right\|.
\end{equation*}
\end{small}
Taking limit in the above equation and using assumptions $(A6)$ and $(A7)$ gives us that, for every $0<\epsilon<T$,
\begin{equation*}
 \lim_{n\to\infty}\sup_{n\leq k\leq \tau^1(n,T)}\left\|\sum_{m=n}^ka(m)M_{m+1}(\omega)\right\|\leq R\epsilon.
\end{equation*}
Therefore for every $\omega\in\Omega_1$, for every $T>0$, $\lim_{n\to\infty}\sup_{n\leq k\leq \tau^1(n,T)}\left\|\sum_{m=n}^ka(m)M_{m+1}
(\omega)\right\|=0$. Thus the additive noise terms $\{M_{n}\}_{n\geq1}$ satisfy assumption $S(A4)$.

Therefore quantities in recursion \eqref{fstrstsr} satisfy assumptions $S(A1)-S(A5)$ and we apply the main result of the single timescale 
recursion (see Theorem \ref{strapt}(iii)) to conclude that,

\begin{lemma}\label{fstrtmp}
 Under assumptions $(A1)-(A8)$, for almost every $\omega$, there exists a non-empty compact set $L\subseteq \mathbb{R}^d$ (depending on $\omega$) 
 such that,
 \begin{itemize}
  \item [(i)] $(X_n(\omega),Y_n(\omega))\to L$ as $n\to\infty$, where $\left\{(X_n,Y_n)\right\}_{n\geq0}$ is as in recursion \eqref{2tr}.
  \item [(ii)] the set $L$ is internally chain transitive for the flow of the DI,
  \begin{equation}\label{fstrdi1}
   \left(\begin{array}{c}\frac{dx}{dt}\\\frac{dy}{dt}\end{array}\right)\in \left(\begin{array}{c} \hat{H}_1(x,y)\\0\end{array}\right).
  \end{equation}
 \end{itemize}
\end{lemma}

The set-valued map associated with the DI \eqref{fstrdi1} is clearly a Marchaud map (use Lemma \ref{march3}). Further any solution 
$(\bm{\mathrm{x}}(\cdot),\bm{\mathrm{y}}(\cdot))$ of DI \eqref{fstrdi1} is such that for every $t\in\mathbb{R}$, $\bm{\mathrm{y}}(t)=\bm{\mathrm{y}}(0)$ and 
$\bm{\mathrm{x}}(\cdot)$ is a solution to DI \eqref{fstrdi} with $y_0=\bm{\mathrm{y}}(0)$. 

Fix $\omega\in\Omega$ such that Lemma \ref{fstrtmp} holds. Let $L\subseteq\mathbb{R}^d$ be as in Lemma \ref{fstrtmp}. Let 
\begin{equation}\label{settmp}
 A:=\left\{(x,y)\in\mathbb{R}^d: x\in\lambda(y),\ y\in\mathbb{R}^{d_2}\right\},
\end{equation}
where for every $y\in\mathbb{R}^{d_2}$, $\lambda(y)$ is as in assumption $(A9)$. Since $L$ is internally chain transitive for the flow of DI 
\eqref{fstrdi1}, by \cite[Lemma 3.5]{benaim1} we know that it is invariant. Let $(x^*,y^*)\in L$ and $(\bm{\mathrm{x}}(\cdot),\bm{\mathrm{y}}(\cdot))$ be a 
solution to DI \eqref{fstrdi1} with initial condition $(x^*,y^*)$ and for all $t\in\mathbb{R}$, $(\bm{\mathrm{x}}(t),\bm{\mathrm{y}}(t))\in
L$. Then for every $t\in\mathbb{R}$, $\bm{\mathrm{y}}(t)=y^*$ and $\bm{\mathrm{x}}(\cdot)$ is a solution of DI \eqref{fstrdi} with $y_0=y^*$. By assumption $(A9)$, there 
exists a compact subset $\lambda(y^*)\subseteq\mathbb{R}^{d_1}$, which is a globally attracting set for the flow of DI \eqref{fstrdi} with 
$y_0=y^*$. By definition of a globally attracting set we have that $\cap_{t\geq0}\overline{\left\{\bm{\mathrm{x}}(q+t):q\geq0\right\}}\subseteq 
\lambda(y^*)$. Therefore $(\bm{\mathrm{x}}(t),\bm{\mathrm{y}}(t))\to \lambda(y^*)\times\{y^*\}\subseteq A$ and since for every $t\in\mathbb{R}$, 
$(\bm{\mathrm{x}}(t),\bm{\mathrm{y}}(t))\in L$ we get that $L\cap A\neq\emptyset$. In fact for any closed set $C\subseteq\mathbb{R}^d$ 
invariant for the flow of DI \eqref{fstrdi1} the above argument gives us that $C\cap A\neq\emptyset$. If we are able to show that $L\cap A=L$, 
then by Lemma \ref{fstrtmp}$(i)$ we obtain that $(X_n(\omega),Y_n(\omega))\to L\subseteq A$ as $n\to\infty$. In this regard we need to impose the 
following assumption.

\begin{itemize}
 \item [(A11)] For any compact set $C\subseteq\mathbb{R}^d$, invariant for the flow of DI \eqref{fstrdi1}, for any open neighborhood 
 $\mathcal{O}$ of $C\cap A$, there exists an open neighborhood $\mathcal{O}'$ of $C\cap A$, such that, 
 \begin{equation*}
  \Phi^C(\mathcal{O}'\cap C,[0,\infty))\subseteq \mathcal{O}\cap C,
 \end{equation*}
 where $\Phi^C:C\times\mathbb{R}\rightarrow\left\{\text{subsets of }\mathbb{R}^d\right\}$ denotes the flow of DI \eqref{fstrdi1} restricted to 
 the invariant set $C$ (see section \ref{diatls} for definition).  
\end{itemize}

\begin{remark}
 The above assumption is a weaker form of assumption $(A1)$ imposed in \cite[Ch.~6]{borkartxt} (that implies assumption $(A11)$ above). The 
 above assumption is basically the Lyapunov stability condition (see \cite[Defn.~IX(ii)]{benaim1}) for the flow restricted to the invariant set 
 $C$. We shall see that in the application studied later the above assumption is satisfied. 
\end{remark}

\begin{lemma}\label{fstrtmp1}
 Under assumptions $(A1)-(A9)$ and $(A11)$, for almost every $\omega$, $L\subseteq \mathbb{R}^d$ as in Lemma \ref{fstrtmp}, is such that 
 \begin{equation*}
  L\subseteq \left\{(x,y)\in\mathbb{R}^d:x\in\lambda(y),\ y\in\mathbb{R}^{d_2}\right\}.
 \end{equation*}
  Therefore $(X_n(\omega),Y_n(\omega))\to L\subseteq\left\{(x,y)\in\mathbb{R}^d:x\in\lambda(y),\ y\in\mathbb{R}^{d_2}\right\}$ as $n\to\infty$, 
  where $\left\{(X_n,Y_n)\right\}_{n\geq0}$ is as in recursion \eqref{2tr}.
\end{lemma}
\begin{proof}
 We present a brief outline here to highlight where assumption $(A11)$ is used.
 
 Let $A\subseteq\mathbb{R}^d$ be as in equation \eqref{settmp}. Fix $\omega\in\Omega$ and obtain $L$ as in Lemma \ref{fstrtmp}. We know that $L$ 
 is internally chain transitive for the flow of DI \eqref{fstrdi1} and since it is also invariant $L\cap A\neq\emptyset$. By assumption $(A9)$, 
 for every $(x^*,y^*)\in L$, $\omega_{\Phi^L}((x^*,y^*))\subseteq L\cap A$. Thus for every $(x^*,y^*)\in L$, for every solution of DI 
 \eqref{fstrdi1}, $(\bf{x}\rm(\cdot),\bf{y}\rm(\cdot))$ such that $(\bm{\mathrm{x}}(0),\ \bm{\mathrm{y}}(0))=(x^*,y^*)$ and for every $t\in\mathbb{R}$, 
 $(\bm{\mathrm{x}}(t),\bm{\mathrm{y}}(t))\in L$, for every open neighborhood $\mathcal{O}$ of $L\cap A$, there exists $t>0$ such that 
 $(\bm{\mathrm{x}}(t),\bm{\mathrm{y}}(t))\in \mathcal{O}\cap L$. By \cite[Lemma~3.13]{benaim1}, we get that for every open neighborhood 
 $\mathcal{O}$ of $L\cap A$, there exists $T>0$, for every $(x^*,y^*)\in L$, for every solution of DI \eqref{fstrdi1}, 
 $(\bm{\mathrm{x}}(\cdot),\bm{\mathrm{y}}(\cdot))$ such that $(\bm{\mathrm{x}}(0),\ \bm{\mathrm{y}}(0))=(x^*,y^*)$ and for every $t\in\mathbb{R}$, 
 $(\bm{\mathrm{x}}(t),\bm{\mathrm{y}}(t))\in L$, for some $t\in[0,T]$, $(\bm{\mathrm{x}}(t),\bm{\mathrm{y}}(t))\in \mathcal{O}\cap L$. 
 
 Fix $\epsilon>0$. Then by assumption $(A11)$ there exists on open neighborhood of $L\cap A$, $\mathcal{O}$ such that, 
 $\Phi^L(\mathcal{O}\cap L,[0,\infty))\subseteq N^{\epsilon}(L\cap A)\cap L$. From arguments in the previous paragraph we can find $T>0$ such that 
 for every $(x^*,y^*)\in L$, for every solution of DI \eqref{fstrdi1} with $(\bm{\mathrm{x}}(0),\ \bm{\mathrm{y}}(0))=(x^*,y^*)$ and $(\bm{\mathrm{x}}(t),
 \bm{\mathrm{y}}(t))\in L$ for every $t\in\mathbb{R}$, there exists $t\in[0,T]$, such that $(\bm{\mathrm{x}}(t),\bm{\mathrm{y}}(t))\in
 \mathcal{O}\cap L$. Therefore $\Phi^L(L,[T,\infty))\subseteq N^{\epsilon}(L\cap A)\cap L$. Thus $L\cap A$ is an attracting set for $\Phi^L$. 
 Now the claim follows from \cite[Prop.~3.20]{benaim1}.\qed
\end{proof}

%% file: slower_timescale_rec_anal.tex
\subsection{Slower timescale recursion analysis}
\label{stsranal}

Before we present the analysis of slower timescale recursion we present some preliminaries where we shall define various quantities needed later.
Throughout this section let $\left\{H_2^{(l)}\right\}_{l\geq1}$ and $\left\{h^{(l)}_2\right\}_{l\geq1}$ denote maps as in section \ref{limdiprop}.
Further we shall allow assumptions $(A1)-(A9)$ and $(A11)$ to be satisfied. The slower timescale recursion analysis is similar to the 
analysis of single timescale inclusion in \cite{vin1t} with minor modifications arising due to the presence of faster timescale 
iterates. Throughout this section $U$ denotes the closed unit ball in $\mathbb{R}^{d_2}$ and $B_{r}$ denotes the closed ball of radius $r>0$ 
in $\mathbb{R}^{d_1}$ centered at the origin.

\subsubsection{Preliminaries}
\label{stsranalprelim}
Let $t^{s}(0):=0$ and for every $n\geq1$, $t^{s}(n):=\sum_{m=0}^{n-1}b(m)$. Define $\bar{Y}:\Omega\times[0,\infty)\rightarrow\mathbb{R}^{d_2}$, 
such that for every $(\omega,t)\in\Omega\times[0,\infty)$, 
\begin{equation*}
 \bar{Y}(\omega,t):=\left(\frac{t-t^s(n)}{t^s(n+1)-t^s(n)}\right)Y_{n+1}(\omega)+\left(\frac{t^s(n+1)-t}{t^s(n+1)-t^s(n)}\right)Y_n(\omega),
\end{equation*}
where $n$ is such that $t\in[t^s(n),t^s(n+1))$. 

Consider the slower timescale recursion \eqref{stsr} given by,
\begin{equation*}
 Y_{n+1}-Y_n-b(n)M^{(2)}_{n+1}\in b(n)H_2(X_n,Y_n,S^{(2)}_n),
\end{equation*}
for every $n\geq0$. By Lemma \ref{ctem}, we have that for every $l\geq1$, for every $n\geq0$, $H_2(X_n,Y_n,S^{(2)}_n)\subseteq 
H_2^{(l)}(X_n,Y_n,S^{(2)}_n)$. Therefore, for every $l\geq1$, the following recursion follows from the above (that is, \eqref{stsr} ): 
\begin{equation*}
 Y_{n+1}-Y_n-b(n)M^{(2)}_{n+1}\in b(n)H_2^{(l)}(X_n,Y_n,S^{(2)}_n).
\end{equation*}
 By lemma \ref{param}, we know that for every $l\geq1$ the set-valued map $H_2^{(l)}$ admits a 
continuous single valued parametrization, $h^{(l)}_2$. The next lemma allows us to write the slower timescale inclusion in terms of the 
parametrization of $H^{(l)}_2$ and result follows from \cite[Lemma~6.1]{vin1t} . 

\begin{lemma}\label{paramuse}
 For every $l\geq1$, for every $n\geq0$, there exists a $U$-valued random variable on 
 $\Omega$, $U^{(l)}_n$, such that for every $\omega\in\Omega$, 
 \begin{equation*}
  Y_{n+1}(\omega)-Y_n(\omega)-b(n)M^{(2)}_{n+1}(\omega)= b(n)h^{(l)}_2(X_n(\omega),Y_n(\omega),S^{(2)}_n(\omega),U^{(l)}_n(\omega)).
 \end{equation*}
\end{lemma}

For every $l\geq1$, define $\Gamma^{(l)}:\Omega\times[0,\infty)\rightarrow\mathcal{P}(U\times\mathbb{R}^{d_1}\times\mathcal{S}^{(2)})$, such that 
for every $(\omega,t)\in\Omega\times[0,\infty)$,
\begin{equation}\label{measproc}
 \Gamma^{(l)}(\omega,t):=\delta_{U^{(l)}_n(\omega)}\otimes\delta_{X_n(\omega)}\otimes\delta_{S^{(2)}_n(\omega)},
\end{equation}
where $\delta_{U^{(l)}_n(\omega)}\in\mathcal{P}(U)$ denotes the Dirac measure at $U^{(l)}_n(\omega)\in U$ (for every $A\in\mathscr{B}
(U)$, $\delta_{U^{(l)}_n(\omega)}(A)=1$ if $U^{(l)}_n(\omega)\in A$, $0$ otherwise), $\delta_{X_n(\omega)}\in\mathcal{P}(\mathbb{R}
^{d_1})$ denotes the Dirac measure at $X_n(\omega)\in\mathbb{R}^{d_1}$ and similarly $\delta_{S^{(2)}_n(\omega)}\in\mathcal{P}(\mathcal{S}^{(2)})$ 
denotes the Dirac measure at $S^{(2)}_n(\omega)\in\mathcal{S}^{(2)}$.

The lemma below provides an equicontinuity result used later. 
\begin{lemma}\label{ecr}
 For every $l\geq1$, for every $r>0$, the family of maps 
 \begin{equation*}
\left\{y\in rU\rightarrow\int_{U\times B_r\times\mathcal{S}^{(2)}}h^{(l)}_2(x,y,s^{(2)},u)\nu(du,dx,ds^{(2)}):\nu\in\mathcal{P}(U\times B_r\times
\mathcal{S}^{(2)})\right\}  
 \end{equation*}
is equicontinuous.
\end{lemma}
\begin{proof}
 Fix $l\geq1$. By Lemma \ref{param}, we know that the map $h^{(l)}_2(\cdot)$ is continuous. Hence the map $h^{(l)}_2(\cdot)$ restricted to the 
 compact set $B_r\times rU\times\mathcal{S}^{(2)}\times U$ is uniformly continuous. Therefore for every $\epsilon>0$, there exists $\delta>0$, 
 such that for every $(x,y,s^{(2)},u),(x,y',s^{(2)},u)\in B_r\times rU\times\mathcal{S}^{(2)}\times U$, satisfying $\left\|y-y'\right\|<\delta$, 
 $\left\|h^{(l)}_2(x,y,s^{(2)},u)-h^{(l)}_2(x,y',s^{(2)},u)\right\|<\epsilon$. Therefore for $\delta>0$ as above, with $\left\|y-y'\right\|<\delta$, 
 for any $\nu\in\mathcal{P}(U\times B_r\times\mathcal{S}^{(2)})$, $\left\|\int_{U\times B_r\times\mathcal{S}^{(2)}}h^{(l)}_2(x,y,s^{(2)},u)\nu(du,dx,ds^{(2)})-
 \int_{U\times B_r\times\mathcal{S}^{(2)}}h^{(l)}_2(x,y',s^{(2)},u)\nu(du,dx,ds^{(2)})\right\|\leq\int_{U\times B_r\times\mathcal{S}^{(2)}}\left\|
 h^{(l)}_2(x,y,s^{(2)},u)-h^{(l)}_2(x,y',s^{(2)},u)\right\|\nu(du,dx,ds^{(2)})\leq\epsilon$.\qed
\end{proof}

For every $l\geq1$, define $G^{(l)}:\Omega\times[0,\infty)\rightarrow\mathbb{R}^{d_2}$ such that, for every $(\omega,t)\in\Omega\times[0,\infty)$, 
\begin{equation}\label{psdode}
 G^{(l)}(\omega,t):=h^{(l)}_2(X_n(\omega),Y_n(\omega),S^{(2)}_n(\omega),U^{(l)}_n(\omega)),
\end{equation}
where $n$ is such that $t\in[t^s(n),t^s(n+1))$. 

In what follows most of the arguments are sample path wise and we use smaller case symbols to denote the above defined quantities along a particular 
sample path. For example $x_n,\ y_n,\ u^{(l)}_n,\ m^{(2)}_{n+1},\ s^{(2)}_n,\ \bar{y}(t),\ \gamma^{(l)}_n(t)$ and $g^{(l)}(t)$ denote $X_n(\omega),\ 
Y_n(\omega),\ U^{(l)}_n(\omega),\ M^{(2)}_{n+1}(\omega),\ S^{(2)}_n(\omega),\ \bar{Y}(\omega,t),\ \Gamma^{(l)}_n(\omega,t)$ and $G^{(l)}(\omega,t)$ 
respectively for some $\omega$ fixed.

\subsubsection{Main result-Asymptotic pseudotrajectory}
\label{stsranalapt}
For every $\omega\in\Omega$, for every $l\geq1$, for every $\tilde{t}\geq0$, let $\tilde{y}^{(l)}(\cdot;\tilde{t})$ denote the solution of the 
o.d.e.
\begin{equation}\label{psdode1}
 \dot{\tilde{y}}^{(l)}(t;\tilde{t})=g^{(l)}(t+\tilde{t}),
\end{equation}
for every $t\geq0$ with initial condition $\tilde{y}^{(l)}(0;\tilde{t})=\bar{y}(\tilde{t})$.

Let $\Omega_1$ be as in \eqref{omega1}. Then by assumptions (A6)-(A8), we have that $\mathbb{P}(\Omega_1)=1$. First we shall get rid of the 
additive noise terms. In this regard we prove the lemma below which states that for every $\omega\in\Omega_1$, the family of functions 
$\left\{\bar{y}(\cdot+t)\right\}_{t\geq0}$ and $\left\{\tilde{y}^{(l)}(\cdot;t)\right\}_{t\geq0}$ have the same limit points in $\mathcal{C}(
[0,\infty),\mathbb{R}^{d_2})$ for every $l\geq1$. Proof of the lemma below is similar to \cite[Lemma~6.3]{vin1t} and is given in appendix 
\ref{prflim1}. 

\begin{lemma}\label{lim1}
 For almost every $\omega$, for every $l\geq1$, for every $T>0$,
 \begin{equation*}
  \lim_{t\to\infty}\sup_{0\leq q\leq T}\left\|\bar{y}(q+t)-\tilde{y}^{(l)}(q;t)\right\|=0.
 \end{equation*}
\end{lemma}

The lemma below guarantees the existence of limit points for $\left\{\tilde{y}^{(l)}(\cdot;t)\right\}_{t\geq0}$ in $\mathcal{C}([0,\infty),
\mathbb{R}^{d_2})$. The proof is similar to \cite[Lemma~6.4]{vin1t} and is given in appendix \ref{prflim2}.

\begin{lemma}\label{lim2}
 For almost every $\omega$, for every $l\geq1$, the family of functions $\left\{\tilde{y}^{(l)}(\cdot;t)\right\}_{t\geq0}$ is relatively 
 compact in $\mathcal{C}([0,\infty),\mathbb{R}^{d_2})$.
\end{lemma}

As a consequence of Lemmas \ref{lim1} and \ref{lim2} we get that, for almost every $\omega$,
\begin{itemize}
 \item [(i)] the family of functions $\left\{\bar{y}(\cdot+t)\right\}_{t\geq0}$ is relatively compact in $\mathcal{C}([0,\infty),\mathbb{R}^d)$,
 \item [(ii)] the linearly interpolated trajectory of the slower timescale iterates, $\bar{y}(\cdot)$, is uniformly continuous on $[0,\infty)$. 
\end{itemize}

The next proposition states that every limit point of $\left\{\bar{y}(\cdot+t)\right\}_{t\geq0}$ is a solution of DI \eqref{stsrdi} on 
$[0,\infty)$. The proof is along the lines of \cite[Prop.~6.5]{vin1t} but with modifications arising due to the presence of faster timescale 
iterates.

\begin{proposition}\label{main1}
 For almost every $\omega$, every limit point $y^*(\cdot)$ of $\left\{\bar{y}(\cdot+t)\right\}_{t\geq0}$ in $\mathcal{C}([0,\infty),
 \mathbb{R}^{d_2})$, satisfies the following.
 \begin{itemize}
  \item [(i)] For some $r>0$, for every $l\geq1$, there exists $\tilde{\gamma}^{(l)}\in\mathcal{M}(U\times B_r\times\mathcal{S}^{(2)})$ such that for 
  every $t\geq0$,
  \begin{equation*}
   y^*(t)=y^*(0)+\int_{0}^{t}\left[\int_{U\times B_r\times\mathcal{S}^{(2)}}h^{(l)}_2(x,y^*(q),s^{(2)},u)\tilde{\gamma}^{(l)}(q)(du,dx,ds^{(2)})\right]dq.
  \end{equation*}
  \item [(ii)] For every $l\geq1$, $\tilde{\gamma}^{(l)}$ as in part $(i)$ of this proposition is such that for almost every $t\geq0$,
  \begin{equation*}
   \Theta_1(\tilde{\gamma}^{(l)})(t)\in D(y^*(t)).
  \end{equation*}
  \item [(iii)] $y^*(\cdot)$ is absolutely continuous and for almost every $t\in[0,\infty)$,
  \begin{equation*}
   \frac{dy^*(t)}{dt}\in\hat{H}_2(y^*(t)).
  \end{equation*}
 \end{itemize}
\end{proposition}
\begin{proof}
Let $\Omega_2:=\left\{\omega\in \Omega:\ \text{Lemma \ref{fstrtmp1} holds}\right\}$. From the proof of \cite[Thm.~6.6 ]{vin1t} it is clear that 
$\Omega_1\subseteq\Omega_2$ and $\mathbb{P}(\Omega_2)=1$. Fix $\omega\in\Omega_2$ and let $t_n\to\infty$ be such that 
$\bar{y}(\cdot+t_n)\to y^*(\cdot)$ in $\mathcal{C}([0,\infty),\mathbb{R}^{d_2})$.
\begin{itemize}
 \item [(i)] Fix $l\geq1$. By assumption $(A8)$ there exists $r>0$, such that $\sup_{n\geq0}\left(\|x_n\|+\|y_n\|\right)\leq r$. Then for any 
 $t\geq0$, $\gamma^{(l)}(t):=\delta_{u^{(l)}_n}\otimes\delta_{x_n}\otimes\delta_{s^{(2)}_n}\in\mathcal{P}(U\times B_r\times\mathcal{S}^{(2)})$ where 
 $n$ is such that $t\in[t^s(n),t^s(n+1))$. Therefore $\gamma^{(l)}\in\mathcal{M}(U\times B_r\times\mathcal{S}^{(2)})$ and by Lemma \ref{mtriz}$(i)$ 
 we get that the sequence $\left\{\gamma^{(l)}(\cdot+t_n)\right\}_{n\geq1}$ has a convergent subsequence in $\mathcal{M}(U\times B_r\times
 \mathcal{S}^{(2)})$. Let $\tilde{\gamma}^{(l)}$ be a limit point of $\left\{\gamma^{(l)}(\cdot+t_n)\right\}_{n\geq1}$ and without loss of generality assume 
 $\gamma^{(l)}(\cdot+t_n)\to\tilde{\gamma}^{(l)}$ as $n\to\infty$.  By definition of $\tilde{y}^{(l)}(\cdot;t_n)$, we have that for every 
 $n\geq1$, for every $t\geq0$,
 \begin{align*}
  \tilde{y}^{(l)}(t;t_n)&=\bar{y}(t_n)+\int_{0}^{t}g^{(l)}(q+t_n)dq\\
                        &=\bar{y}(t_n)+\int_{0}^{t}h^{(l)}_2(x_{[t_n+q]},y_{[t_n+q]},s^{(2)}_{[t_n+q]},u^{(l)}_{[t_n+q]})dq,
 \end{align*}
 where for any $t\geq0$, $[t]:=\max\left\{n\geq0:t\geq t^s(n)\right\}$. Using the defintion of $\gamma^{(l)}(\cdot+t_n)$(see \eqref{measproc} 
 and recall that $\gamma^{(l)}(\cdot+t_n):=\Gamma^{(l)}(\omega,\cdot+t_n)$) in the above we get that for every $n\geq1$, for every $t\geq0$,
 \begin{align*}
  \tilde{y}^{(l)}(t;t_n)&=\bar{y}(t_n)+\int_{0}^{t}\left[\int_{U\times B_r\times\mathcal{S}^{(2)}}h^{(l)}_2(x,y_{[t_n+q]},s^{(2)},u)\gamma^{(l)}
  (q+t_n)(du,dx,ds^{(2)})\right]dq.
 \end{align*}
 Since $\bar{y}(\cdot+t_n)\to y^*(\cdot)$ in $\mathcal{C}([0,\infty),\mathbb{R}^{d_2})$, by Lemma \ref{lim1} we have that 
 $\tilde{y}^{(l)}(\cdot;t_n)\to y^*(\cdot)$ as $n\to\infty$. By taking limit in the above equation we get that for every $t\geq0$, 
 \begin{align}\label{tmp_main}
  \lim_{n\to\infty}\left[\tilde{y}^{(l)}(t;t_n)-\bar{y}(t_n)\right]&=\lim_{n\to\infty}\int_{0}^{t}\left[\int_{U\times B_r\times\mathcal{S}^{(2)}}
                                                               h^{(l)}_2(x,y_{[t_n+q]},s^{(2)},u)\gamma^{(l)}(q+t_n)(du,dx,ds^{(2)})\right]dq,\nonumber\\
      y^*(t)-y^*(0)&=\lim_{n\to\infty}\int_{0}^{t}\left[\int_{U\times B_r\times\mathcal{S}^{(2)}}h^{(l)}_2(x,y_{[t_n+q]},s^{(2)},u)\gamma^{(l)}(q+t_n)
                      (du,dx,ds^{(2)})\right]dq.                            
 \end{align}
 Since $\gamma^{(l)}(\cdot+t_n)\to\tilde{\gamma}^{(l)}(\cdot)$ and by our choice of the topology for $\mathcal{M}(U\times B_r\times\mathcal{S}^{(2)})$, 
  we have,
  \begin{small}
  \begin{equation*}
   \int_{0}^{t}\left[\int_{U\times B_r\times\mathcal{S}^{(2)}}\tilde{f}(q,u,x,s^{(2)})\gamma^{(l)}(q+t_n)(du,dx,ds^{(2)})\right]dq-
   \int_{0}^{t}\left[\int_{U\times B_r\times\mathcal{S}^{(2)}}\tilde{f}(q,u,x,s^{(2)})\tilde{\gamma}^{(l)}(q)(du,dx,ds^{(2)})\right]dq\to0,
  \end{equation*}
  \end{small}
 for all  bounded continuous $\tilde{f}:[0,t]\times U\times B_r\times\mathcal{S}^{(2)}\rightarrow\mathbb{R}$ of the form, 
 \begin{equation*}
  \tilde{f}(q,u,x,s^{(2)})=\sum_{m=1}^{N}a_mg_m(q)f_m(u,x,s^{(2)}),
 \end{equation*}
 for some $N\geq1$, scalars $a_m$ and bounded continuous functions $g_m,\ f_m$ on $[0,t],\ U\times B_r\times\mathcal{S}^{(2)}$ respectively, 
 for $1\leq m\leq N$. By the Stone-Weierstrass theorem, such functions can uniformly approximate any function in 
 $\mathcal{C}([0,t]\times U\times B_r\times\mathcal{S}^{(2)},\mathbb{R})$. Thus the above convergence holds true for all real valued continuous functions 
 on $[0,t]\times U\times B_r\times\mathcal{S}^{(2)}$, implying that $\\t^{-1}\gamma^{(l)}(q+t_n)(du,dx,ds^{(2)})dq\to t^{-1}\tilde{\gamma}^{(l)}(q)(du,
 dx,ds^{(2)})dq$ in $\mathcal{P}([0,t]\times U\times B_r\times\mathcal{S}^{(2)})$. Thus,
 \begin{small}
 \begin{align}\label{tmp0}
  \|\int_{0}^{t}\left[\int_{U\times B_r\times\mathcal{S}^{(2)}}\!\!\!\!\!h^{(l)}_2(x,y^*(q),s^{(2)},u)\gamma^{(l)}(q+t_n)(du,dx,ds^{(2)})\right]&dq-\nonumber\\
  \int_{0}^{t}&\left[\int_{U\times B_r\times\mathcal{S}^{(2)}}\!\!\!\!\!\!\!\!h^{(l)}_2(x,y^*(q),s^{(2)},u)\tilde{\gamma}^{(l)}(q)(du,dx,ds^{(2)})\right]dq\|\to0
 \end{align}
 \end{small}
 as $n\to\infty$. Since $\left\{\bar{y}(\cdot+t_n)|_{[0,t]}\right\}_{n\geq1}$ converges uniformly to $y^*(\cdot)|_{[0,t]}$ we have that, the 
 function $q\rightarrow y_{[t_n+q]}$ converges uniformly to $y^*(\cdot)|_{[0,t]}$ on $[0,t]$. Using the above and Lemma \ref{ecr} we have that 
 for every $\epsilon>0$, there exists $N$(depending on $\epsilon$) such that for every $n\geq N$, for every $q\in[0,t]$,
 \begin{small}
 \begin{equation}\label{tmp1}
  \|\int_{U\times B_r\times\mathcal{S}^{(2)}}\!\!\!\!\!h^{(l)}_2(x,y_{[t_n+q]},s^{(2)},u)\gamma^{(l)}(q+t_n)(du,dx,ds^{(2)})-
  \int_{U\times B_r\times\mathcal{S}^{(2)}}\!\!\!\!\!h^{(l)}_2(x,y^*(q),s^{(2)},u)\gamma^{(l)}(q+t_n)(du,dx,ds^{(2)})\|<\epsilon.
 \end{equation}
 \end{small}
 Now,
 \begin{small}
  \begin{align*}
   \|\int_{0}^{t}\!\!\left[\int_{U\times B_r\times\mathcal{S}^{(2)}}\!\!\!\!\!\!\!\!\!\!\!h^{(l)}_2(x,y_{[t_n+q]},s^{(2)},u)\gamma^{(l)}(q+t_n)(du,dx,ds^{(2)})\right]dq&-\!\!
  \int_{0}^{t}\!\!\left[\int_{U\times B_r\times\mathcal{S}^{(2)}}\!\!\!\!\!\!\!\!\!\!\!h^{(l)}_2(x,y^*(q),s^{(2)},u)\tilde{\gamma}^{(l)}(q)(du,dx,ds^{(2)})\right]\!dq\|\\
  &\leq\\
  \|\int_{0}^{t}\!\!\int_{U\times B_r\times\mathcal{S}^{(2)}}\!\!\!\!\!\!\!\!\!\!\!h^{(l)}_2(x,y_{[t_n+q]},s^{(2)},u)\gamma^{(l)}(q+t_n)(du,dx,ds^{(2)})\!dq&-\!\!
  \int_{0}^{t}\!\!\int_{U\times B_r\times\mathcal{S}^{(2)}}\!\!\!\!\!\!\!\!\!\!\!h^{(l)}_2(x,y^*(q),s^{(2)},u)\gamma^{(l)}(q+t_n)(du,dx,ds^{(2)})\!dq\|\\
  &+\\
  \|\int_{0}^{t}\!\!\left[\int_{U\times B_r\times\mathcal{S}^{(2)}}\!\!\!\!\!\!\!\!\!\!\!h^{(l)}_2(x,y^*(q),s^{(2)},u)\gamma^{(l)}(q+t_n)(du,dx,ds^{(2)})\right]\!dq&-\!\!
  \int_{0}^{t}\!\!\left[\int_{U\times B_r\times\mathcal{S}^{(2)}}\!\!\!\!\!\!\!\!\!\!\!h^{(l)}_2(x,y^*(q),s^{(2)},u)\tilde{\gamma}^{(l)}(q)(du,dx,ds^{(2)})\right]\!dq\|.
  \end{align*}
 \end{small}
 Taking limit as $n\to\infty$ in the above equation and using \eqref{tmp0} and \eqref{tmp1} we get,
 \begin{small}
 \begin{align*}
  \lim_{n\to\infty}\|\int_{0}^{t}\left[\int_{U\times B_r\times\mathcal{S}^{(2)}}\!\!\!\!\!\!\!\!\!\!\!\!\!\!h^{(l)}_2(x,y_{[t_n+q]},s^{(2)},u)\gamma^{(l)}(q+t_n)(du,dx,ds^{(2)})\right]&dq-\\
  \int_{0}^{t}&\left[\int_{U\times B_r\times\mathcal{S}^{(2)}}\!\!\!\!\!\!\!\!\!\!\!\!\!\!h^{(l)}_2(x,y^*(q),s^{(2)},u)\tilde{\gamma}^{(l)}(q)(du,dx,ds^{(2)})\right]dq\|\leq\epsilon t,
 \end{align*}
 \end{small}
 for every $\epsilon>0$. Therefore, for every $t\geq0$,
 \begin{small}
  \begin{equation*}
   \lim_{n\to\infty}\int_{0}^{t}\left[\int_{U\times B_r\times\mathcal{S}^{(2)}}\!\!\!\!\!\!\!\!\!\!\!\!\!\!\!h^{(l)}_2(x,y_{[t_n+q]},s^{(2)},u)\gamma^{(l)}(q+t_n)(du,dx,ds^{(2)})\right]dq=
   \int_{0}^{t}\left[\int_{U\times B_r\times\mathcal{S}^{(2)}}\!\!\!\!\!\!\!\!\!\!\!\!\!\!\!h^{(l)}_2(x,y^*(q),s^{(2)},u)\tilde{\gamma}^{(l)}(q)(du,dx,ds^{(2)})\right]dq.
  \end{equation*}
 \end{small}
Substituting the above limit in equation \eqref{tmp_main} we get that for every $t\geq0$,
\begin{equation*}
 y^*(t)-y^*(0)=\int_{0}^{t}\left[\int_{U\times B_r\times\mathcal{S}^{(2)}}h^{(l)}_2(x,y^*(q),s^{(2)},u)\tilde{\gamma}^{(l)}(q)(du,dx,ds^{(2)})\right]dq.
\end{equation*}
\item [(ii)] Fix $l\geq1$. Let $\mu^{(l)}:=\Theta_1(\gamma^{(l)})$. Then for every $n\geq1$, $\mu^{(l)}(\cdot+t_n)=\Theta_1(\gamma^{(l)}
(\cdot+t_n))$ and since $\gamma^{(l)}(\cdot+t_n)\to\tilde{\gamma}^{(l)}$ in $\mathcal{M}(U\times B_r\times\mathcal{S}^{(2)})$, by Lemma 
\ref{contfun}$(v)$ we get that $\mu^{(l)}(\cdot+t_n)\to\tilde{\mu}^{(l)}(\cdot)=:\Theta_1(\tilde{\gamma}^{(l)})$ in $\mathcal{M}(B_r\times
\mathcal{S}^{(2)})$ as $n\to\infty$. In order to prove that for almost every $t\geq0$, $\tilde{\mu}^{(l)}(t)\in D(y^*(t))$ we need to show that for 
almost every $t\geq0$,
\begin{itemize}
\item [(1)] $\mathrm{supp}(\tilde{\mu}^{(l)}_{B_r}(t))\subseteq \lambda(y^*(t))$,
\item [(2)] for every $A\in\mathscr{B}(\mathcal{S}^{(2)})$, $\tilde{\mu}^{(l)}_{\mathcal{S}^{(2)}}(t)(A)=
\int_{\mathbb{R}^{d_1}\times\mathcal{S}^{(2)}}\Pi^{(2)}(x,y^*(t),s^{(2)})(A)\tilde{\mu}^{(l)}(t)(dx,ds^{(2)})$.
\end{itemize}

First we present a proof of the claim in (1) above. $\mu^{(l)}_{B_r}:=\Theta_2(\mu^{(l)})\in\mathcal{M}(B_r)$ and since as $n\to\infty$, 
$\mu^{(l)}(\cdot+t_n)\to\tilde{\mu}^{(l)}$, by Lemma \ref{contfun}$(vi)$, we have that $\mu^{(l)}_{B_r}(\cdot+t_n)\to\tilde{\mu}^{(l)}_{B_r}$ in 
$\mathcal{M}(B_r)$. Since as $n\to\infty$, $\mu^{(l)}_{B_r}(\cdot+t_n)\to\tilde{\mu}^{(l)}_{B_r}$ in $\mathcal{M}(B_r)$, by proof of 
\cite[Ch.~6, Lemma~3]{borkartxt}, we have that for almost every $t\geq0$, there exists a subsequence $\left\{n_k\right\}_{k\geq1}$ and a 
subsequence of natural numbers $\left\{c_p\right\}_{p\geq1}$ such that 
\begin{equation}\label{tmp3}
\frac{1}{c_p}\sum_{k=1}^{c_p}\mu^{(l)}_{B_r}(t+t_{n_k})\to\tilde{\mu}^{(l)}_{B_r}(t)
\end{equation}
in $\mathcal{P}(B_r)$ as $p\to\infty$. Fix $t\geq0$ such that the above holds. By definition of $\mu^{(l)}_{B_r}$, we have that 
for every $k\geq1$, $\mu^{(l)}_{B_r}(t+t_{n_k})=\delta_{x_{[t+t_{n_k}]}}$ where $[t+t_{n_k}]:=\max\left\{m>n:t+t_{n_k}\geq t^s(m)\right\}$. 
Since $\bar{y}(\cdot+t_{n_k})\to y^*(t)$, using the uniform continuity of $\bar{y}(\cdot)$, we have that the function 
$\tilde{t}\in[0,\infty)\rightarrow y_{[\tilde{t}+t_{n_k}]}$ converges uniformly on compacts to the function $y^*(\cdot)$. Therefore 
$y_{[t+t_{n_k}]}\to y^*(t)$ as $k\to\infty$ and by Lemma \ref{fstrtmp1} we have that $x_{[t+t_{n_k}]}\to\lambda(y^*(t))$. Further by definition 
of $r>0$ as in part $(i)$  of this proposition, we get that $\sup_{n\geq1}\|x_n\|\leq r$. Hence $\lambda(y^*(t))\cap B_r\neq\emptyset$ and 
$x_{[t+t_{n_k}]}\to\lambda(y^*(t))\cap B_r$ as $k\to\infty$. For every $\epsilon>0$, clearly $\left(\lambda(y^*(t))+B_{\epsilon}\right)\cap B_r$ 
is compact and there exists $K$ large such that for every $k\geq K$, $x_{[t+t_{n_k}]}\in \left(\lambda(y^*(t))+B_{\epsilon}\right)\cap B_r$ and 
hence $\delta_{x_{[t+t_{n_k}]}}(\left(\lambda(y^*(t))+B_{\epsilon}\right)\cap B_r)=1$ for every $k\geq K$. 
Since $\frac{1}{c_p}\sum_{k=1}^{c_p}\mu^{(l)}_{B_r}(t+t_{n_k})\to\tilde{\mu}^{(l)}_{B_r}(t)$ in $\mathcal{P}(B_r)$ as $p\to\infty$, by 
\cite[Thm.~2.1.1(iv)]{borkarap} we have that, for every $\epsilon>0$,
\begin{align*}
\limsup_{p\to\infty}\frac{1}{c_p}\sum_{k=1}^{c_p}\mu^{(l)}_{B_r}(t+t_{n_k})(\left(\lambda(y^*(t))+B_{\epsilon}\right)\cap B_r)&=
\limsup_{p\to\infty}\frac{1}{c_p}\sum_{k=1}^{c_p}\delta_{x_{[t+t_{n_k}]}}(\left(\lambda(y^*(t))+B_{\epsilon}\right)\cap B_r)\\&=
1\\&\leq\tilde{\mu}^{(l)}_{B_r}(t)(\left(\lambda(y^*(t))+B_{\epsilon}\right)\cap B_r)\leq1.
\end{align*}
Therefore for every $\epsilon>0$, $\tilde{\mu}^{(l)}_{B_r}(t)(\left(\lambda(y^*(t))+B_{\epsilon}\right)\cap B_r)=1$ which gives us that 
$\tilde{\mu}^{(l)}_{B_r}(t)(\lambda(y^*(t))\cap B_r)=1$ and hence $\mathrm{supp}(\tilde{\mu}^{(l)}_{B_r}(t))\subseteq\lambda(y^*(t))$. 
Since \eqref{tmp3} holds for almost every $t\geq0$, we have that for almost every $t\geq0$, $\mathrm{supp}(\tilde{\mu}^{(l)}_{B_r}(t))\subseteq
\lambda(y^*(t))$.

The proof of the claimin (2) above is similar to the proof of \cite[Prop.~6.5(ii)]{vin1t} and we provide a brief outline for the sake of completeness. Let 
$\{f_i\}_{i\geq1}\subseteq\mathcal{C}(\mathcal{S}^{(2)},\mathbb{R})$ be a convergence determining class for $\mathcal{P}(\mathcal{S}^{(2)})$. By an 
appropriate affine transformation we can ensure that for every $i\geq1$, for every $s^{(2)}\in\mathcal{S}^{(2)}$, $0\leq f_i(s^{(2)})\leq1$. Define,
\begin{equation*}
 \zeta^i_n:=\sum_{k=0}^{n-1}b(k)\left(f_i(S^2_{k+1})-\int_{\mathcal{S}^{(2)}}f_i(\tilde{s}^{(2)})\Pi^{(2)}(X_k,Y_k,S^2_k)(d\tilde{s}^{(2)})\right),
\end{equation*}
for every $n\geq1$, for every $i\geq1$. For every $i\geq1$, $\left\{\zeta^i_n\right\}_{n\geq1}$ is a square integrable martingale w.r.t. the 
filtration $\left\{\mathscr{F}_{n}:=\sigma\left(X_k,Y_k,S^2_k:0\leq k\leq n\right)\right\}_{n\geq1}$ and further $\sum_{n=0}^{\infty}
\mathbb{E}\left[(\zeta^i_{n+1}-\zeta_n^i)^2|\mathscr{F}_n\right]\leq2\sum_{n=1}^{\infty}(b(n))^2<\infty$. By Martingale convergence theorem 
(see \cite[Appendix~C, Thm.~11]{borkartxt}) we get that for almost every $\omega$, for every $i\geq1$, $\left\{\zeta^i_n\right\}_{n\geq1}$ 
converges. Let $\Omega_m:=\left\{\omega\in\Omega:\ \forall i\geq1,\ \{\zeta^i_n\}\ \text{converges}\right\}$. Define, $\Omega^*:=
\Omega_m\cap\Omega_2$ and from the arguments above we get that $\mathbb{P}(\Omega^*)=1$. Therefore for every $\omega\in\Omega^*$, for every 
$i\geq1$ for every $T>0$,
\begin{equation*}
 \sum_{k=n}^{\tau^2(n,T)}b(k)\left(f_i(s^{(2)}_{k+1})-\int_{\mathcal{S}^{(2)}}f_i(\tilde{s}^{(2)})\Pi^{(2)}(x_k,y_k,s^{(2)}_k)(d\tilde{s}^{(2)})\right)\to0,
\end{equation*}
as $n\to\infty$. By our choice of $\left\{f_i\right\}_{i\geq1}$, the fact that the step size sequence $\left\{b(n)\right\}_{n\geq0}$ is non- 
increasing and the definition of $\mu^{(l)}$, we get that for every $\omega\in\Omega^*$, for every $i\geq1$, for every $T>0$,
\begin{equation*}
 \lim_{t\to\infty}\int_{0}^{T}\int_{B_r\times\mathcal{S}^{(2)}}\left[f_i(s^{(2)})-\int_{\mathcal{S}^{(2)}}f_i(\tilde{s}^{(2)})\Pi^{(2)}(x,y_{[q+t]},s^{(2)})
 (d\tilde{s}^{(2)})\right]\mu^{(l)}(q+t)(dx,ds^{(2)})dq=0,
\end{equation*}
where $[q+t]:=\max\left\{n\geq0:t\geq t^s(n)\right\}$. By assumption $(A4)$, we have that for every $i\geq1$, the function $(x,y,s^{(2)})\rightarrow 
f_i(s^{(2)})-\int_{\mathcal{S}^{(2)}}f_i(\tilde{s}^{(2)})\Pi^{(2)}(x,y,s^{(2)})(d\tilde{s}^{(2)})$ is continuous and hence the restriction of the above function to 
the compact set $B_r\times rU\times \mathcal{S}^{(2)}$ is uniformly continuous where $r>0$ is as in part $(i)$ of this proposition. Using the 
uniform continuity above and the fact that $\lim_{t\to\infty}\sup_{0\leq q\leq T}\left\|\bar{y}(q+t)-y_{[q+t]}\right\|=0$ (which follows 
from definition of $\bar{y}(\cdot)$ and uniform continuity of $\bar{y}(\cdot)$) we get that for every $\omega\in\Omega^*$, for every $i\geq1$, 
for every $T>0$,
\begin{equation}\label{tmp4}
 \lim_{t\to\infty}\int_{0}^{T}\int_{B_r\times\mathcal{S}^{(2)}}\left[f_i(s^{(2)})-\int_{\mathcal{S}^{(2)}}f_i(\tilde{s}^{(2)})\Pi^{(2)}(x,\bar{y}(q+t),s^{(2)})
 (d\tilde{s}^{(2)})\right]\mu^{(l)}(q+t)(dx,ds^{(2)})dq=0.
\end{equation}

We know that as $n\to\infty$, $\bar{y}(\cdot+t_n)\to y^*(\cdot)$ in $\mathcal{C}([0,\infty),\mathbb{R}^{d_2})$ and $\mu^{(l)}(\cdot+t_n)\to
\tilde{\mu}^{(l)}(\cdot)$ in $\mathcal{M}(B_r\times\mathcal{S}^{(2)})$. Further by arguments similar to Lemma \ref{ecr}, the family of functions 
\begin{equation*}
\left\{y\in rU\rightarrow \int_{B_r\times\mathcal{S}^{(2)}}\left[f_i(s^{(2)})-\int_{\mathcal{S}^{(2)}}f_i(\tilde{s}^{(2)})\Pi^{(2)}(x,y,s^{(2)})(d\tilde{s}^{(2)})\right]
\nu(dx,ds^{(2)}):\nu\in\mathcal{P}(B_r\times\mathcal{S}^{(2)})\right\} 
\end{equation*}
is equicontinuous. Therefore, for every $\omega\in\Omega^*$, for every $T>0$,
\begin{small}
\begin{equation*}
 \lim_{t\to\infty}\left\|\int_0^T\!\!\int_{B_r\times\mathcal{S}^{(2)}}\!\!\!\left[\int_{\mathcal{S}^{(2)}}\!\!\!f_i(\tilde{s}^{(2)})\Pi^{(2)}(x,y^*(q),s^{(2)})(d\tilde{s}^{(2)})\!-\!\!
 \int_{\mathcal{S}^{(2)}}\!\!\!f_i(\tilde{s}^{(2)})\Pi^{(2)}(x,\bar{y}(q+t),s^{(2)})(d\tilde{s}^{(2)})\right]\mu^{(l)}(q+t)(dx,ds^{(2)})dq\right\|=0,
\end{equation*}
\begin{align*}
\lim_{t\to\infty}\bigg\|\int_{0}^{T}\int_{B_r\times\mathcal{S}^{(2)}}\bigg[f_i(s^{(2)})-\int_{\mathcal{S}^{(2)}}f_i(\tilde{s}^{(2)})\Pi^{(2)}(x,y^*(q),s^{(2)})
 (d\tilde{s}^{(2)})\bigg]&\mu^{(l)}(q+t)(dx,ds^{(2)})dq\\-\int_{0}^{T}\int_{B_r\times\mathcal{S}^{(2)}}\bigg[f_i(s^{(2)})-\int_{\mathcal{S}^{(2)}}f_i(\tilde{s}^{(2)})
 \Pi^{(2)}(x,y^*(q),s^{(2)})&(d\tilde{s}^{(2)})\bigg]\tilde{\mu}^{(l)}(q)(dx,ds^{(2)})dq\bigg\|=0.
\end{align*}
\end{small}
Using the above and equation \eqref{tmp4}, we get that for every $\omega\in\Omega^*$, for every $T>0$,
\begin{equation*}
 \int_{0}^{T}\int_{B_r\times\mathcal{S}^{(2)}}\left[f_i(s^{(2)})-\int_{\mathcal{S}^{(2)}}f_i(\tilde{s}^{(2)})\Pi^{(2)}(x,y^*(q),s^{(2)})
 (d\tilde{s}^{(2)})\right]\tilde{\mu}^{(l)}(q)(dx,ds^{(2)})dq=0.
\end{equation*}
By applying Lebesgue's differentiation theorem (see \cite[Ch.~11.1.3]{borkartxt}), we get that for every $\omega\in\Omega^*$, for every $i\geq1$, 
for almost every $t\geq0$, 
\begin{equation*}
 \int_{B_r\times\mathcal{S}^{(2)}}\left[f_i(s^{(2)})-\int_{\mathcal{S}^{(2)}}f_i(\tilde{s}^{(2)})\Pi^{(2)}(x,y^*(t),s^{(2)})
 (d\tilde{s}^{(2)})\right]\tilde{\mu}^{(l)}(t)(dx,ds^{(2)})=0.
\end{equation*}
Since for every $t\geq0$, $\tilde{\mu}^{(l)}(t)\in\mathcal{P}(B_r\times\mathcal{S}^{(2)})$ it is also an element of $\mathcal{P}(\mathbb{R}^{d_1}\times\
\mathcal{S}^{(2)})$ with $\mathrm{supp}(\tilde{\mu}^{(l)}(t))\subseteq B_r\times\mathcal{S}^{(2)}$. Therefore, for every $\omega\in\Omega^*$, for every 
$i\geq1$, for almost every $t\geq0$,
\begin{equation*}
 \int_{\mathbb{R}^{d_1}\times\mathcal{S}^{(2)}}\left[f_i(s^{(2)})-\int_{\mathcal{S}^{(2)}}f_i(\tilde{s}^{(2)})\Pi^{(2)}(x,y^*(t),s^{(2)})
 (d\tilde{s}^{(2)})\right]\tilde{\mu}^{(l)}(t)(dx,ds^{(2)})=0.
\end{equation*}
Since $\left\{f_i\right\}_{i\geq1}$ is a convergence determining class for $\mathcal{P}(\mathcal{S}^{(2)})$, from the above it follows that 
for every $\omega\in\Omega^*$, for almost every $t\geq0$,
\begin{equation*}
 \tilde{\mu}^{(l)}_{\mathcal{S}^{(2)}}(t)(d\tilde{s}^{(2)})=\int_{\mathbb{R}^{d_1}\times\mathcal{S}^{(2)}}\Pi^{(2)}(x,y^*(t),s^{(2)})(d\tilde{s}^{(2)})\tilde{\mu}^{(l)}(t)
 (dx,ds^{(2)}).
\end{equation*}
\item [(iii)] Fix $l\geq1$. Then by part $(i)$ of this proposition, $y^*(\cdot)$ is clearly absolutely continuous and for almost every $t\geq0$,
\begin{equation*}
 \frac{dy^*(t)}{dt}=\int_{U\times\mathbb{R}^{d_1}\times\mathcal{S}^{(2)}}h^{(l)}_2(x,y^*(t),s^{(2)},u)\tilde{\gamma}^{(l)}(t)(du,dx,ds^{(2)}).
\end{equation*}
By part $(ii)$ of this lemma we know that for almost every $t\geq0$, $\Theta_1(\tilde{\gamma}^{(l)})(t)=
\tilde{\gamma}^{(l)}_{B_r\times\mathcal{S}^{(2)}}(t)=\tilde{\gamma}^{(l)}_{\mathbb{R}^{d_1}\times\mathcal{S}^{(2)}}(t)\in D(y^*(t))$. Hence, by Lemma 
\ref{chint}$(ii)$, for almost every $t\geq0$,
\begin{align*}
 \frac{dy^*(t)}{dt}&=\int_{U\times\mathbb{R}^{d_1}\times\mathcal{S}^{(2)}}h^{(l)}_2(x,y^*(t),s^{(2)},u)\tilde{\gamma}^{(l)}(t)(du,dx,ds^{(2)})\\
                   &\in\cup_{\mu\in D(y^*(t))}\int_{\mathbb{R}^{d_1}\times\mathcal{S}^{(2)}}H^{(l)}_{2,y^*(t)}(x,s^{(2)})\mu(dx,ds^{(2)})\\
                   &=\hat{H}^{(l)}_2(y^*(t)).
\end{align*}
Since the above holds for every $l\geq1$, we get that for almost every $t\geq0$,
\begin{equation*}
 \frac{dy^*(t)}{dt}\in\cap_{l\geq1}\hat{H}^{(l)}_2(y^*(t))=\hat{H}_2(y^*(t)),
\end{equation*}
where the last equality follows from Lemma \ref{rel}$(vi)$.\qed
\end{itemize} 
\end{proof}

A continuous function $\bf{y}\rm:\mathbb{R}\rightarrow\mathbb{R}^{d_2}$ is said to be an asymptotic pseudotrajectory for the flow of DI 
\eqref{stsrdi} if $\lim_{t\to\infty}\bf{D}\rm(\bf{y}\rm(\cdot+t),\Sigma^2)=0$ where $\Sigma^2\subseteq \mathcal{C}(\mathbb{R},\mathbb{R}^{d_2})$ 
denotes the set of solutions of DI \eqref{stsrdi}. Fix $\omega\in\Omega^*$. Extend $\bar{y}(\cdot)$ to the whole of $\mathbb{R}$ by defining 
$\bar{y}(t)=\bar{y}(0)$ for every $t<0$. Then by assumption $(A8)$ and uniform continuity of $\bar{y}(\cdot)$ we have that the family of functions 
$\left\{\bar{y}(\cdot+t)\right\}_{t\geq0}$ is relatively compact in $\mathcal{C}(\mathbb{R},\mathbb{R}^{d_2})$. Let 
$y^*(\cdot)$ be a limit point of the above family of functions. Then by Proposition \ref{main1}$(iii)$, $y^*(\cdot)|_{[0,\infty)}$ is a solution of DI \eqref{stsrdi} on 
$[0,\infty)$. Usually the negative time argument is omitted since it follows from the positive time argument as follows:

Fix $T>0$. Let $t_n\to\infty$ be such that $\bar{y}(\cdot+t_n)\to y^*(\cdot)$ in $\mathcal{C}(\mathbb{R},\mathbb{R}^{d_2})$. Then 
$\bar{y}(\cdot+t_n-T)\to y^*(\cdot-T)$. By Proposition \ref{main1}$(iii)$, $y^*(\cdot-T)|_{[0,\infty)}$ is a solution of DI \eqref{stsrdi} on 
$[0,\infty)$. Therefore $y^*(\cdot)|_{[-T,0]}$ is absolutely continuous and for $a.e.\ t\in[-T,0]$, $\frac{dy^*(t)}{dt}\in \hat{H}_2(y^*(t))$. 
Since $T$ was arbitrary we have that the $y^*(\cdot)|_{(-\infty,0]}$ is solution of DI \eqref{stsrdi} on $(-\infty,0]$. Therefore $y^*(\cdot)\in
\Sigma^2$ and by \cite[Thm.~4.1]{benaim1} we get the following result.

\begin{theorem}\emph{[APT]}\label{aptstsr}
 Under assumptions $(A1)-(A9)$ and $(A11)$, for almost every $\omega$, the linearly interpolated trajectory of the slower timescale recursion 
 \eqref{stsr}, $\bar{y}(\cdot)$, is an asymptotic pseudotrajectory of DI \eqref{stsrdi}.
\end{theorem}

\subsubsection{Characterization of limit sets}
\label{stsranalcls}
As a consequence of Theorem \ref{aptstsr} for almost every $\omega$, the limit sets of the slower timescale recursion, $L(\bar{y})$, defined as,
\begin{equation}\label{limsetstsr}
 L(\bar{y}):=\cap_{t\geq0}\overline{\left\{\bar{y}(q+t): q\geq0\right\}},
\end{equation}
can be characterized in terms of the dynamics of DI \eqref{stsrdi}. Further using Lemma \ref{fstrtmp1} we get the main result of this paper 
stated below.

\begin{theorem}\emph{[Limit set]}\label{lsthm}
 Under assumptions $(A1)-(A9)$ and $(A11)$, for almost every $\omega$,
 \begin{itemize}
  \item [(i)] $L(\bar{y})$ is a non-empty, compact subset of $\mathbb{R}^{d_2}$ and is internally chain transitive for the flow of DI 
  \eqref{stsrdi},
  \item [(ii)] if assumption $(A10)$ is satisfied then $L(\bar{y})\subseteq \mathcal{Y}$ and as $n\to\infty$,
  \begin{equation*}
   \left(\begin{array}{c}x_n\\y_n\end{array}\right)\to \cup_{y\in\mathcal{Y}}\left(\lambda(y)\times\left\{y\right\}\right).
  \end{equation*}
 \end{itemize}
\end{theorem}
\begin{proof}
 Fix $\omega\in\Omega^*$.
 \begin{itemize}
  \item [(i)] By Theorem \ref{aptstsr} we know that $\bar{y}(\cdot)$ is an asymptotic pseudotrajectory for the flow of DI \eqref{stsrdi}. Now the 
  claim follows from \cite[Thm.~4.3]{benaim1}.
  \item [(ii)] By part $(i)$ of this theorem we know that $L(\bar{y})$ is internally chain transitive for the flow of DI\eqref{stsrdi}. Since 
  $\mathcal{Y}$ is a globally attracting set for DI\eqref{stsrdi}, by \cite[Cor.~3.24]{benaim1} we get that $L(\bar{y})\subseteq\mathcal{Y}$. 
  Therefore $y_n\to\mathcal{Y}$ as $n\to\infty$ and by Lemma \ref{fstrtmp1}, we get that, as $n\to\infty$,
  \begin{equation*}
   \left(\begin{array}{c}x_n\\y_n\end{array}\right)\to \cup_{y\in\mathcal{Y}}\left(\lambda(y)\times\left\{y\right\}\right).\qed
  \end{equation*}
 \end{itemize}

\end{proof}

%% file: application.tex
\section{Application: Constrained convex optimization}
\label{appl}

In this section we consideran application of the theory to a problem of constrained convex optimization. Throughout this section we 
assume that $\mathcal{S}^{(1)}=\mathcal{S}^{(2)}=\mathcal{S}$ and $|\mathcal{S}|<\infty$.

Let the objective function $J:\mathbb{R}^{d_1}\times\mathcal{S}\rightarrow\mathbb{R}$ be such that  $J(\cdot)$ is continuous and for every 
$s\in\mathcal{S}$, $J(\cdot,s)$ is convex and coercive (that is for any $M>0$, there exists $r>0$, such that for any $x\in\mathbb{R}^{d_1}$ with 
$\|x\|\geq r$, we have that $J(x,s)\geq M$). The functions describing the constraints are given by $C:\mathcal{S}\rightarrow\mathbb{R}^{d_2\times 
d_1}$ and $w:\mathcal{S}\rightarrow\mathbb{R}^{d_2}$. We assume that for any $s\in\mathcal{S}$, the set $\mathcal{X}(s):=\left\{x\in\mathbb{R}
^{d_1}:C(s)x=d(s)\right\}$ is non empty. The law of the Markov noise terms is given by $\Pi:\mathbb{R}^{d_1}\times\mathcal{S}\rightarrow
\mathcal{P}(\mathcal{S})$ such that $\Pi$ is continuous and let $\mu\in\mathcal{P}(\mathcal{S})$ denote the unique stationary distribution of 
the Markov chain given by the transition kernel $\Pi(x,\cdot)(\cdot)$, for every $x\in\mathbb{R}^{d_1}$. 

Let $\partial J(x,s)$ denote the set of subgradients of the convex function $J(\cdot,s)$ at the point $x\in\mathbb{R}^{d_1}$. Formally,
\begin{equation*}
 \partial J(x,s):=\left\{g\in\mathbb{R}^{d_1}:\ \forall x'\in\mathbb{R}^{d_1},\ J(x',s)\geq J(x,s)+\left\langle g,x'-x\right\rangle\right\}.
\end{equation*}
Then it is easy to show that for every $(x,s)\in\mathbb{R}^{d_1}\times\mathcal{S}$, $\partial J(x,s)$ is convex and compact. Further the map 
$(x,s)\rightarrow \partial J(x,s)$ possesses the closed graph property. We assume that the map $(x,s)\rightarrow\partial J(x,s)$ satisfies the linear 
growth property, that is, there exists $K>0$ such that $\sup_{x'\in\partial J(x,s)}\|x'\|\leq K(1+\|x\|)$.

Let $J_{\mu}:\mathbb{R}^{d_1}\rightarrow\mathbb{R}$ be defined such that for every $x\in\mathbb{R}^{d_1}$, $J_{\mu}(x):=\int_{\mathcal{S}}J(x,s)
\mu(ds)$. Similarly define $C_{\mu}:=\int_{\mathcal{S}}C(s)\mu(ds)\in\mathbb{R}^{d_2\times d_1}$ and $w_{\mu}:=\int_{\mathcal{S}}w(s)\mu(ds)\in
\mathbb{R}^{d_2}$. The optimization problem that we wish to solve is given by,
\begin{align*}
 OP(\mu):\ \  &\min_{x\in\mathbb{R}^{d_1}}J_{\mu}(x),\ \mathrm{subject\ to:}\\
              &C_{\mu}x=w_{\mu}.
\end{align*}
 
The standard approach in solving the optimization problem $OP(\mu)$ is the projected subgradient descent algorithm whose recursion is given by,
 \begin{equation*}
  X_{n+1}=P_{\mu}\left(X_n-a(n)(g_n+M_{n+1})\right),
 \end{equation*}
where $g_n\in\partial J_{\mu}(X_n)$, $M_{n+1}$ is the subgradient estimation error and $P_{\mu}$ denotes the projection operation onto the 
affine subspace $\mathcal{X}_{\mu}:=\left\{x\in\mathbb{R}^{d_1}:C_{\mu}x=w_{\mu}\right\}$. Such a scheme cannot be implemented when $\mu$ is 
not known. Such is the case in problems arising in optimal control. 

The feasible set of the optimization problem $OP(\mu)$, given by $\mathcal{X}_{\mu}$ is non empty since for every $s\in\mathcal{S}$, 
$\mathcal{X}(s)$ is non-empty. Further, since for every $s\in\mathcal{S}$, $J(\cdot,s)$ is 
coercive, the function $J_{\mu}(\cdot)$ is coercive and hence bounded below. Therefore the optimization problem $OP(\mu)$ has at least one 
solution.  Let the solution set of the optimization problem $OP(\mu)$, be denoted by $Z$.

For any $r>0$, let $B_r$ denote the closed ball of radius $r$ in $\mathbb{R}^{d_1}$ centered at the origin. For every $s\in\mathcal{S}$, pick $x_s\in\mathcal{X}(s)$, 
and compute $M_1:=\max\{J(x_s,s'):s,s'\in\mathcal{S}\}$. Then $x_{\mu}:=\sum_{s\in\mathcal{S}}\mu(s)x_s\in\mathcal{X}_{\mu}$ and 
$J_{\mu}(x_{\mu})\leq M_1$. Since $|\mathcal{S}|<\infty$ and the functions $J(\cdot,s)$ are coercive, for some $M>\max\{0,M_1\}$, there exists 
$r>\max\{\|x_s\|:s\in\mathcal{S}\}$, such that for every $s\in\mathcal{S}$, for every $x\in B_r^c$, $J(x,s)\geq M$ and for every $s\in\mathcal{S}$, 
$B_r\cap\mathcal{X}(s)\neq\emptyset$. Then $Z\subseteq B_r$. Instead of $OP(\mu)$ we shall solve the following penalized/regularized optimization 
problem given by, 
\begin{align*}
 \tilde{OP}(\mu):\ \  &\min_{x\in\mathbb{R}^{d_1}}J_{\mu}(x)+\frac{\epsilon}{2r^2}\left\|x\right\|^2+\frac{K+1}{2}
                      \max\left\{\left\|x\right\|^2-r^2,0\right\},\\
                      &\mathrm{subject\ to:}\ C_{\mu}x=w_{\mu},
\end{align*}
where, $r>0$ is as determined above, $K$ is the constant associated with the linear growth property of the subgradient map $\partial J(\cdot,s)$ and 
$\epsilon>0$ is an arbitrary constant small in value. Then it is easy to show that $\tilde{OP}(\mu)$ has at least one solution and 
the set of solutions of $\tilde{OP}(\mu)$, denoted by $\tilde{Z}$ is such that $\tilde{Z}\subseteq B_r$. Further for any $\tilde{x}\in\tilde{Z}$,
for any $x^*\in Z$, $J_{\mu}(\tilde{x})-J_{\mu}(x^*)\leq \epsilon$.

Consider the Lagrangian $L:\mathbb{R}^{d}\rightarrow\mathbb{R}$ associated with optimization problem $\tilde{OP}(\mu)$ defined such that for 
every $(x,y)\in\mathbb{R}^{d}$,
\begin{equation*}
 L(x,y):=J_{\mu}(x)+\frac{\epsilon}{2r^2}\left\|x\right\|^2+\frac{K+1}{2}\max\left\{\left\|x\right\|^2-r^2,0\right\}+\left\langle y, 
 C_{\mu}x-w_{\mu}\right\rangle.
\end{equation*}
Let $\hat{J}:\mathbb{R}^{d_1}\times\mathcal{S}\rightarrow\mathbb{R}$, be defined such that for every $(x,s)\in\mathbb{R}^{d_1}\times\mathcal{S}$, 
$\hat{J}(x,s):=J(x,s)+\frac{\epsilon}{2r^2}\left\|x\right\|^2+\frac{K+1}{2}\max\left\{\left\|x\right\|^2-r^2,0\right\}$. Then, for every $(x,s)
\in\mathbb{R}^{d_1}\times\mathcal{S}$,
\begin{equation*}
 \hat{J}_{\mu}(x):=J_{\mu}(x)+\frac{\epsilon}{2r^2}\left\|x\right\|^2+\frac{K+1}{2}\max\left\{\left\|x\right\|^2-r^2,0\right\}
 =\int_{\mathcal{S}}\hat{J}(x,s)\mu(ds).
\end{equation*}

When the transition law $\Pi$ and hence $\mu$ is not known we propose the following recursion which performs primal descent along the faster 
time scale (that is minimization of $L(\cdot,y)$ w.r.t. $x$) and dual ascent on the slower timescale (that is maximization of $L(x,\cdot)$ 
w.r.t. $y$). The recursion is given by,
\begin{subequations}\label{algo}
\begin{align}\label{algo1}
 Y_{n+1}-Y_{n}&=b(n)(C(S_n)X_n-w(S_n)),\\
 \label{algo2}
 X_{n+1}-X_{n}-a(n)M^1_{n+1}&\in -a(n)\left(\partial \hat{J}(X_n,S_n)+C(S_n)^TY_n\right),
\end{align}
 \end{subequations}
 
where the step size sequences $\left\{a(n)\right\}_{n\geq0}$ and $\left\{b(n)\right\}_{n\geq0}$ are chosen such that they satisfy assumption 
$(A5)$ and $\left\{M^1_{n}\right\}_{n\geq1}$ denotes the subgradient estimation error which is assumed to satisfy assumption $(A6)$ ( for 
example, when $\left\{M^1_{n}\right\}_{n\geq1}$ is i.i.d. zero mean with finite variance, assumption $(A6)$ is satisfied. More generally $(A6)$ is 
satisfied if $\left\{M^1_{n}\right\}_{n\geq1}$ is a martingale difference sequence satisfying assumption $(A3)$ in \cite[ch.~2.1]{borkartxt}). 

It is easy to see that the maps $(x,y,s)\rightarrow-\left(\partial \hat{J}(x,s)+C(s)^Ty\right)$ and $(x,y,s)\rightarrow C(s)x-w(s)$ satisfy 
assumptions $(A1)$ and $(A2)$ respectively. The linear growth property of the map $(x,y,s)\rightarrow-\left(\partial \hat{J}(x,s)+C(s)^Ty\right)$, 
follows from the linear growth property of $x\rightarrow\partial J(x,s)$. Further by \cite[Prop.~5.4.6]{bertsekas}, we get that for every 
$(x,y)\in\mathbb{R}^{d}$, $-\int_{\mathcal{S}}\left(\partial \hat{J}(x,s)+C(s)^Ty\right)\mu(ds)=-\left(\partial \hat{J}_{\mu}(s)+C^T_{\mu}y\right
)=-\partial L(x,y)$. 

For every $y\in\mathbb{R}^{d_2}$, let $\lambda(y):=\left\{ x\in\mathbb{R}^{d_1}:-C_{\mu}^Ty\in\partial \hat{J}_{\mu}(x)\right\}$. Then for 
every $y\in\mathbb{R}^{d_2}$, $\lambda(y)$ is non-empty since $L(\cdot,y)$ is convex and coercive. Further $|\lambda(y)|=1$, that is 
$\lambda(y)$ is a singleton since, $L(\cdot,y)$ is strictly convex. For any $y\in\mathbb{R}^{d_2}$, $x'\in\lambda(y)$ if and only 
if there exists $\tilde{g}\in\mathbb{R}^{d_1}$ in the set of subgradients of the function 
$J_{\mu}(\cdot)+\frac{K+1}{2}\max\{\|\cdot\|^2-r^2,0\}$ at $x'$, such that $\tilde{g}+\frac{\epsilon}{r^2}x'+C_{\mu}^Ty=0$. So either 
$\|x'\|\leq r$ or if $\|x'\|>r$, then $\max\{\|x'\|^2-r^2,0\}=\|x'\|^2-r^2$ and hence for some $g\in\partial J_{\mu}(x')$, $\tilde{g}=g+(K+1)x'$, 
from which we get that,
\begin{align*}
 (K+1+\frac{\epsilon}{2r^2})\|x'\|&=\|g+C_{\mu}^Ty\|\\
                                 &\leq K+K\|x'\|+\|C_{\mu}^T\|\|y\|.
\end{align*}
Thus for $K':=\max\{K,r,\|C_{\mu}^T\|\}$, we get that for every $y\in\mathbb{R}^{d_2}$, $\|\lambda(y)\|\leq K'(1+\|y\|)$. The set $\lambda(y)$ 
is clearly globally attracting for the flow of DI $\frac{dx}{dt}\in-\left(\partial \hat{J}_{\mu}(x)+C^T_{\mu}y\right)$ and by 
\cite[Thm.~6]{aubindi} the map $y\in\mathbb{R}^{d_2}\rightarrow\lambda(y)$ is u.s.c. (since $\lambda(\cdot)$ is also single valued, 
it is continuous). Hence the map $\lambda(\cdot)$ satisfies assumption $(A9)$.  

If the iterates are stable for a.e. $\omega$ (that is, $(A8)$ is satisfied), the result in section \ref{fstranal} gives us that for almost every 
$\omega$, there exists a non-empty compact set $A\subseteq \mathbb{R}^{d}$, such that $(X_n(\omega),Y_n(\omega))\to A$ as $n\to\infty$ and $A$ 
is internally chain transitive for the flow of DI,
\begin{equation}\label{appfstrdi}
 \left(\begin{array}{c}\frac{dx}{dt}\\\frac{dy}{dt}\end{array}\right)\in\left(\begin{array}{c}-\partial \hat{J}_{\mu}(x)-C^T_{\mu}y\\0\end{array}\right).
\end{equation}

By arguments in section \ref{fstranal}, we have that $A\cap\mathcal{G}(\lambda)\neq\emptyset$, where $\mathcal{G}(\lambda):=\left\{(\lambda(y),y): 
y\in\mathbb{R}^{d_2}\right\}$. Let $\mathcal{O}\subseteq\mathbb{R}^{d}$ be an open neighborhood of $\mathcal{G}(\lambda)$. Let 
$\mathcal{O}'(\delta):=\{(x,y)\in\mathbb{R}^{d}: L(x,y)-L(\lambda(y),y)<\delta\}$. By \cite[ch.~1.2, Thm.~6]{aubindi},
the map $y\in\mathbb{R}^{d_2}\rightarrow L(\lambda(y),y)$ is continuous and hence $\mathcal{O}'(\delta)$ is an open neighborhood 
of $\mathcal{G}(\lambda)$. Further it is easy to show that $\cap_{\delta>0}\mathcal{O}'(\delta)=\mathcal{G}(\lambda)$ and hence $\cap_{\delta>0}
\left(\mathcal{O}'(\delta)\cap A\right)=\mathcal{G}(\lambda)\cap A$. Since $A\subseteq\mathbb{R}^d$ is compact, there exists $\delta^*>0$, such 
that $\mathcal{O}'(\delta^*)\cap A\subseteq \mathcal{O}\cap A$. Consider any solution of DI \eqref{appfstrdi}, $(\bm{\mathrm{x}}(\cdot),\bm{\mathrm{y}}(\cdot))$ 
starting at $(x^*,y^*)\in \mathcal{O}'(\delta^*)\cap A$ and satisfying for every $t\in\mathbb{R}$, $(\bm{\mathrm{x}}(t),\bm{\mathrm{y}}(t))
\in A$. Recall from section\ref{fstranal} that $(\bm{\mathrm{x}}(\cdot),\bm{\mathrm{y}}(\cdot))$ as above is such that for every $t\geq0$, $\bm{\mathrm{y}}(t
)=y^*$ and $\bm{\mathrm{x}}(\cdot)$ is a solution of DI, $\frac{dx}{dt}\in -(\partial \hat{J}_{\mu}(x)+C_{\mu}^Ty^*)=-\partial L(x,y^*)$ and 
hence descends along the potential $L(x,y^*)$. Therefore the solution $(\bm{\mathrm{x}}(\cdot),\bm{\mathrm{y}}(\cdot))$ remains within 
$\mathcal{O}'(\delta^*)\cap A$ which gives us that $\Phi^A(\mathcal{O}'(\delta^*)\cap A,[0,\infty))\subseteq \mathcal{O}\cap A$, where $\Phi^A$ 
denotes the flow of DI \eqref{appfstrdi} restricted to the set $A$. Thus assumption $(A11)$ is satisfied and from Lemma \ref{fstrtmp1} we get the
following result.

\begin{lemma}\emph{[Faster timescale convergence]}\label{idontknow}
For almost every $\omega$, $(X_{n}(\omega),Y_n(\omega))\to \mathcal{G}(\lambda)$ as $n\to\infty$.  
\end{lemma}

Theorem \ref{aptstsr} gives us that the iterates $\left\{Y_n\right\}_{n\geq0}$ in recursion \eqref{algo1} track the flow of DI,
\begin{equation}\label{algotmp1}
 \frac{dy}{dt}\in \cup_{\nu\in D(y)}\int_{\mathbb{R}^{d_1}\times \mathcal{S}}\left(C(s)x-w(s)\right)\nu(dx,ds),
\end{equation}
where, for every $y\in\mathbb{R}^{d_2}$, $D(y)$ is as in equation \eqref{measmapdefn}. Since for every $y\in\mathbb{R}^{d_2}$, $\lambda(y)$ is a 
singleton and since $\mu$ is the unique stationary distribution of the Markov chain given by transition kernel $\Pi(\cdot)(\cdot)$, we get that 
for every $y\in\mathbb{R}^{d_2}$, $D(y)=\delta_{\lambda(y)}\otimes\mu$. Therefore DI \eqref{algotmp1} takes the form,
\begin{equation}\label{algotmp2}
 \frac{dy}{dt}=C_{\mu}\lambda(y)-w_{\mu}.
\end{equation}

In order to analyze the asymptotic behavior of o.d.e. \eqref{algotmp2}, we need the following version of the envelope theorem. The 
proof of the envelope theorem below is similar to that in \cite{milgrom}.

\begin{lemma}\emph{[envelope theorem]}\label{et}
 Let $\bm{\mathrm{y}}:[0,T]\rightarrow\mathbb{R}^{d_2}$ be an absolutely continuous function. Let $\tilde{L}:\mathbb{R}^{d_1}\times[0,T]\rightarrow\mathbb{R}$ 
 be defined such that for every $(x,t)\in\mathbb{R}^{d_1}\times[0,T]$, $\tilde{L}(x,t):=L(x,\bm{\mathrm{y}}(t))$. Then,
 \begin{itemize}
  \item [(i)] for every $x\in\mathbb{R}^{d_1}$, $\tilde{L}(x,\cdot)$ is absolutely continuous and there exists $\mathcal{D}\subseteq [0,T]$ with 
  Lebesgue measure $T$ such that for every $t\in\mathcal{D}$, for every $x\in\mathbb{R}^{d_1}$, $\frac{\partial\tilde{L}(x,t)}{\partial t}$ 
  exists and $\frac{\partial \tilde{L}(x,t)}{\partial t}=\left\langle \frac{d\bm{\mathrm{y}}(\it{t})}{dt},C_{\mu}x-w_{\mu}\rm\right\rangle$.
  \item [(ii)] the function $V:[0,T]\rightarrow\mathbb{R}$ where for every $0\leq t\leq T$, $V(t):=\inf_{x\in\mathbb{R}^{d_1}}\tilde{L}(x,t)$ is 
  absolutely continuous. Further for any $0<t\leq T$, 
  \begin{equation*}
     V(t)=V(0)+\int_0^t\frac{\partial\tilde{L}(x,q)}{\partial q}\bigg{|}_{x=\lambda(\bm{\mathrm{y}}(q))}dq.
  \end{equation*}
\end{itemize}
\end{lemma}
\begin{proof}
 \begin{itemize}
  \item [(i)] Since $\bm{\mathrm{y}}(\cdot)$ is absolutely continuous, it is differentiable almost everywhere and let $\mathcal{D}\subseteq [0,T]$ be 
  the set of $t\in[0,T]$, such that $\frac{d\bm{\mathrm{y}}(t)}{dt}$ exists. Then clearly the Lebesgue measure of $\mathcal{D}$ is $T$.
  
  Fix $x\in\mathbb{R}^{d_1}$. Then $L(x,\cdot)$ is a Lipschitz continuous function since for any $y',y''\in\mathbb{R}^{d_2}$, 
  \begin{align*}
   \left|L(x,y')-L(x,y'')\right|&= \left|\left\langle y'-y'',C_{\mu}x-w_{\mu}\right\rangle\right|\\
                                &\leq\left\|y'-y''\right\|\ \left\|C_{\mu}x-w_{\mu}\right\|\\
                                &=\beta_{x}\left\|y'-y''\right\|,
  \end{align*}
where $\beta_{x}:=\left\|C_{\mu}x-w_{\mu}\right\|$. Further $L(x,\cdot)$ is differentiable (i.e. totally differentiable since it is linear in $y$) 
and the total derivative is given by $\nabla_{y}L(x,y')=\left(C_{\mu}x-w_{\mu}\right)$ for every $y'\in\mathbb{R}^{d_2}$. Since 
$\tilde{L}(x,\cdot)$, is the composition of absolutely continuous function $\bm{\mathrm{y}}(\cdot)$ and a Lipschitz continuous function $L(x,\cdot)$, 
we have that $\tilde{L}(x,\cdot)$ is absolutely continuous. By \cite[Thm.~9.15]{rudin}, we have that for every $t\in\mathcal{D}$, 
$\frac{\partial\tilde{L}(x,t)}{dt}$ exists and $\frac{\partial\tilde{L}(x,t)}{dt}=\left\langle\frac{d\bm{\mathrm{y}}(\it{t'}\rm)}{dt}, 
C_{\mu}x-w_{\mu}\right\rangle$.

  \item [(ii)] Since $\bm{\mathrm{y}}(\cdot)$ is absolutely continuous, there exists $\alpha>0$ such that $\sup_{t\in[0,T]}\left\|\bm{\mathrm{y}}(t)
  \right\|\leq \alpha$. Further by assumption $(A9)$, for every $t\in [0,T]$, 
  \begin{equation*}
  V(t)=\inf_{\substack{x\in\mathbb{R}^{d_1}:\\ \|x\|\leq K'(1+\|\bm{\mathrm{y}}(t)\|)}}\tilde{L}(x,t)=\inf_{\substack{x\in\mathbb{R}^{d_1}:\\ \|x\|\leq K'
  (1+\alpha)}}\tilde{L}(x,t).
  \end{equation*}
  Therefore for every $0\leq t<t'\leq T$, 
  \begin{align*}
   \left|V(t')-V(t)\right|&\leq \sup_{\substack{x\in\mathbb{R}^{d_1}:\\ \|x\|\leq K'(1+\alpha)}}\left|\tilde{L}(x,t')-\tilde{L}(x,t)\right|\\
   &\leq\sup_{\substack{x\in\mathbb{R}^{d_1}:\\ \|x\|\leq K'(1+\alpha)}}\left|\int_t^{t'}\frac{\partial\tilde{L}(x,q)}{\partial t}dq\right|\\
   &=\sup_{\substack{x\in\mathbb{R}^{d_1}:\\ \|x\|\leq K'(1+\alpha)}}\left|\left\langle \bm{\mathrm{y}}(t')-\bm{\mathrm{y}}(t), \it{C_{\mu}x-w_{\mu}}\right
   \rangle\right|\\
   &\leq\left(\sup_{\substack{x\in\mathbb{R}^{d_1}:\\ \|x\|\leq K'(1+\alpha)}}\left\|C_{\mu}x-w_{\mu}\right\|\right)\left\|\bm{\mathrm{y}}(t')-\bm{\mathrm{y}}(t)\right\|.
  \end{align*}
  Now the absolute continuity of $V(\cdot)$ follows from absolute continuity of $\bm{\mathrm{y}}(\cdot)$. Since $V(\cdot)$ is absolutely continuous, 
  $\frac{dV(q)}{dq}$ exists for a.e. $q\in[0,T]$ and for any $0<t\leq T$, $V(t)=V(0)+\int_{0}^t\frac{dV(q)}{dq}dq$. Let $q\in(0,T)$ be such that 
  $\frac{dV(q)}{dq}$ exists and $q\in \mathcal{D}$. Then for $q'>q$, $V(q')-V(q)\leq \tilde{L}(\lambda(\bm{\mathrm{y}}(\it{q})),q')-\tilde{L}(\lambda(
  \bm{\mathrm{y}}(\it{q})),q)$. Therefore 
  the right hand derivative of $V(\cdot)$ at $q$ which is the same as $\frac{dV(q)}{dq}$ satisfies, $\frac{dV(q)}{dq}\leq 
  \frac{\partial\tilde{L}(x,q)}{\partial q}|_{x=\lambda(\bm{\mathrm{y}}(q))}$. Considering $q<q'$ and repeating the above argument gives us, 
  $\frac{\partial\tilde{L}(x,q)}{\partial q}|_{x=\lambda(\bm{\mathrm{y}}(q))}\leq\frac{dV(q)}{dq}$. Thus for a.e. $q\in[0,T]$, 
  $\frac{\partial\tilde{L}(x,q)}{\partial q}|_{x=\lambda(\bm{\mathrm{y}}(q))}=\frac{dV(q)}{dq}$ and since $V(\cdot)$ is absolutely continuous, 
  for any $0<t\leq T$,
  \begin{equation*}
   V(t)=V(0)+\int_{0}^t\frac{\partial\tilde{L}(x,q)}{\partial q}\bigg{|}_{x=\lambda(\bm{\mathrm{y}}(q))}dq.\qed
  \end{equation*}
 \end{itemize}
\end{proof}

Let $Q_{\mu}:\mathbb{R}^{d_2}\rightarrow\mathbb{R}$ be defined such that for $y\in\mathbb{R}^{d_2}$, 
$Q_{\mu}(y):=\inf_{x\in\mathbb{R}^{d_1}}L(x,y)=L(\lambda(y),y)$. The function $Q_{\mu}(\cdot)$ is the objective function of the dual of 
the optimization problem $\tilde{OP}(\mu)$ and is a concave function. By the strong duality theorem (see \cite[Prop.~5.3.3]{bertsekas}), the 
dual optimization problem given by, $\max_{y\in\mathbb{R}^{d_2}}Q_{\mu}(y)$ has at least one solution and let the set of solutions of the dual 
optimization problem be denoted by $\mathcal{Y}$. Further the strong duality theorem also gives us that for any $y\in\mathcal{Y}$ and for any 
$x\in\tilde{Z}$, $Q_{\mu}(y)=\hat{J}_{\mu}(x)$.

Let $\bm{\mathrm{y}}:\mathbb{R}\rightarrow \mathbb{R}^{d_2}$ be a solution of the o.d.e. \eqref{algotmp2} with initial condition $y\in\mathcal{Y}$. 
Then $\bm{\mathrm{y}}(\cdot)$ is absolutely continuous and for a.e. $t\in[0,\infty)$, $\frac{d\bm{\mathrm{y}}(t)}{dt}=C_{\mu}\lambda(\bm{\mathrm{y}}(t))
-w_{\mu}$. By Lemma \ref{et}$(ii)$ we have that for any $t\geq0$,
\begin{align}
 V(t)&=V(0)+\int_{0}^t\frac{\partial\tilde{L}(x,q)}{\partial q}\bigg{|}_{x=\lambda(\bm{\mathrm{y}}(q))}dq,\nonumber\\
     &=V(0)+\int_{0}^t\left\langle\frac{d\bm{\mathrm{y}}(q)}{dq},C_{\mu}\lambda(\bm{\mathrm{y}}(q))-w_{\mu}\right\rangle dq,\nonumber\\
     &=V(0)+\int_0^t\left\|C_{\mu}\lambda(\bm{\mathrm{y}}(\it{q}))-w_{\mu}\right\|^2dq.
\end{align}
Since $V(t)=Q_{\mu}(\bm{\mathrm{y}}(t))$ and $V(0)=Q_{\mu}(y)$, where $y\in\mathcal{Y}$, we get that $V(t)-V(0)\leq 0$. Hence for every $t\geq0$, 
$\int_0^t\left\|C_{\mu}\lambda(\bm{\mathrm{y}}(q))-w_{\mu}\right\|^2dq\leq 0$ which gives us that $\|C_{\mu}\lambda(\bm{\mathrm{y}}(t))-w_{\mu}\|=0$ for 
a.e. $t\in[0,\infty)$. Thus for any 
solution of o.d.e. \eqref{algotmp2}, $\bm{\mathrm{y}}(\cdot)$, with initial condition $y\in\mathcal{Y}$, we have that 
$C_{\mu}\lambda(y)-w_{\mu}=0$ and for every $t\geq0$, $\bm{\mathrm{y}}(t)=y$. Therefore $\mathcal{Y}\subseteq \left\{y\in
\mathbb{R}^{d_2}:C_{\mu}\lambda(y)-w_{\mu}=0\right\}$. Further by \cite[Prop.~5.3.3(ii)]{bertsekas}, any $y\in\mathbb{R}^{d_2}$, such that 
$C_{\mu}\lambda(y)-w_{\mu}=0$, is a solution of the dual optimization problem and hence $\mathcal{Y}=\left\{y\in\mathbb{R}^{d_2}:C_{\mu}
\lambda(y)-w_{\mu}=0\right\}$ (from this it also follows that $\mathcal{Y}$ is closed). 

In the theorem below we summarize the main convergence result associated with the recursion \eqref{algo}.

\begin{theorem}\emph{[Convergence to Lagrangian saddle points]}
\begin{itemize}
 \item [(i)] For any solution $\bm{\mathrm{y}}(\cdot)$ of the o.d.e. \eqref{algotmp2} with any initial condition $y_0\in\mathbb{R}^{d_2}$ which is 
 bounded for $t\geq0$ (that is $\sup_{t\geq0}\|\bm{\mathrm{y}}(\it{t}\rm)\|<\infty$), we have that as $t\to\infty$, $\inf_{y\in\mathcal{Y}}\|\bm{\mathrm{y}}
 (t)-y\|\to 0$.
 \item [(ii)] For any $y\in\mathcal{Y}$, $\lambda(y)$ is a solution of the optimization problem $\tilde{OP}(\mu)$ (that is 
 $\lambda(y)\in \tilde{Z})$.  
 \item [(iii)] If the iterates remain stable for almost every $\omega$ (that is $(A8)$ is satisfied), then, for almost every $\omega$, 
 \begin{itemize}
 \item [(a)] $Y_n(\omega)\to\mathcal{Y}$ as $n\to\infty$,
 \item [(b)] $\left(\begin{array}{c}X_n(\omega)\\Y_n(\omega)\end{array}\right)\to\cup_{y\in\mathcal{Y}}
 \left\{\left(\begin{array}{c}\lambda(y)\\y\end{array}\right)\right\}\subseteq\mathbb{R}^{d}$.
 \end{itemize}
\end{itemize}
\end{theorem}
\begin{proof}$\ $
\begin{itemize}
 \item [(i)] Let $\bm{\mathrm{y}}(\cdot)$ be a solution of the o.d.e. \eqref{algotmp2} with initial condition $y_0\in\mathbb{R}^{d_2}$ (assume $y_0\notin\mathcal{Y}$ 
 since, otherwise we know that for every $t\geq0$, $\bm{\mathrm{y}}(t)=y_0$ and hence the claim follows) such that $\sup_{t\geq0}\|\bm{\mathrm{y}}(\it{t}\rm)\|\leq M$ 
 for some $M>0$. Then $\bm{\mathrm{y}}(\cdot)|_{[0,\infty)}$ is uniformly continuous since for any $0\leq t<t'<\infty$,
 \begin{align*}
  \left\|\bm{\mathrm{y}}(t')-\bm{\mathrm{y}}(t)\right\|&=\left\|\int_{t}^{t'}\left(C_{\mu}\bm{\mathrm{y}}(q)-w_{\mu}\right)dq\right\|\\
                                                           &\leq\int_t^{t'}\left(\left\|C_{\mu}\right\|\left\|\bm{\mathrm{y}}(q)\right\|+
                                                           \left\|w_{\mu}\right\|\right)dq\\
                                                           &\leq\left(\left\|C_{\mu}\right\|M+\left\|w_{\mu}\right\|\right)(t'-t).
 \end{align*}
 The function $y\in\mathbb{R}^{d_2}\rightarrow\|C_{\mu}y-w_{\mu}\|$ is uniformly continuous and hence the function $t\in[0,\infty)\rightarrow 
 \|C_{\mu}\bm{\mathrm{y}}(\it{t})-w_{\mu}\|$ is uniformly continuous. Further by Lemma \ref{et}$(ii)$, for any $t>0$, $0\leq V(t)-V(0)\leq Q_{\mu}(y)-V(0)
 <\infty$ where $y\in\mathcal{Y}$. The claim that as $t\to\infty$, $\bm{\mathrm{y}}(t)\to\mathcal{Y}$ is equivalent to the claim that as 
 $t\to\infty$, $\|C_{\mu}\bm{\mathrm{y}}(t)-w_{\mu}\|\to0$. 
 
 Suppose there exists $\epsilon>0$, such that for every $T>0$, there exists $t\geq T$, such that $\|C_{\mu}\bm{\mathrm{y}}(t)-w_{\mu}\|>\epsilon$.
 From the uniform continuity of $t\to\|C_{\mu}\bm{\mathrm{y}}(t)-w_{\mu}\|$, there exists $\delta>0$ such that for every $t,t'\in[0,\infty)$ 
 satisfying $|t-t'|<\delta$, $\left|\left\|C_{\mu}\bm{\mathrm{y}}(t)-w_{\mu}\right\|-\left\|C_{\mu}\bm{\mathrm{y}}(t')-w_{\mu}\right\|\right|<
 \frac{\epsilon}{2}$. Therefore we can obtain a sequence $\{t_n\}_{n\geq1}$ such that for every $n\geq1$, $\delta<t_n<t_{n+1}-2\delta$ and for every 
 $t\in(t_n-\delta,t_n+\delta)$, $\|C_{\mu}\bm{\mathrm{y}}(t)-w_{\mu}\|>\frac{\epsilon}{2}$. Let $N$ be such that $\frac{2(Q_{\mu}(y)-V(0))}
 {\epsilon^2\delta}<N$ where $y\in\mathcal{Y}$. Then by Lemma \ref{et}$(ii)$, we get that,
 \begin{align*}
  Q_{\mu}(\bm{\mathrm{y}}(\it{t_{N+1}}))-V(\rm0)&=V(t_{N+1})-V(0)\\
                                 &=\int_0^{t_{N+1}}\left\langle\frac{d\bm{\mathrm{y}}(q)}{dq},C_{\mu}\lambda(\bm{\mathrm{y}}(q))-w_{\mu}\right\rangle dq\\
                                 &=\int_0^{t_{N+1}}\|C_{\mu}\bm{\mathrm{y}}(q)-w_{\mu}\|^2dq\\
                                 &\geq\sum_{n=1}^N \int_{t_n-\delta}^{t_n+\delta}\|C_{\mu}\bm{\mathrm{y}}(\it{q})-w_{\mu}\|^2dq\\
                                 &>N\left(\frac{\epsilon^2\delta}{2}\right)\\
                                 &> Q_{\mu}(y)-V(0)
 \end{align*}
which contradicts the fact that $V(t)-V(0)\leq Q_{\mu}(y)-V(0)$. Therefore $\lim_{t\to\infty}\|C_{\mu}\bm{\mathrm{y}}(\it{t})-w_{\mu}\|=0$.

\item [(ii)] Let $y\in\mathcal{Y}$. Then we know that $C_{\mu}\lambda(y)-w_{\mu}=0$ and hence $\lambda(y)$ is feasible for $\tilde{OP}(\mu)$. By definition of 
$\lambda(y)$, we have that for every $x\in\mathbb{R}^{d_1}$, $L(\lambda(y),y)\leq L(x,y)$. Now the claim follows from 
\cite[Prop.~5.3.3(ii)]{bertsekas}.
\item [(iii)] Let $\omega$ be such that Theorem \ref{lsthm} holds. 
\begin{itemize}
\item [(a)] Then by Theorem \ref{lsthm}(i) we know that there exists a non empty, compact set $A\subseteq \mathbb{R}^{d_2}$ such that as $n\to\infty$, 
$Y_n(\omega)\to A$. Further $A$ is internally chain transitive for the flow of o.d.e. \eqref{algotmp2} and hence is invariant. Let 
$\bm{\mathrm{y}}(\cdot)$ be a solution to o.d.e. \eqref{algotmp2} with initial condition in $A$ and for every $t\in\mathbb{R}$, $\bm{\mathrm{y}}(\it{t})\in A$.
Since $A$ is compact, $\sup_{t\geq0}\|\bm{\mathrm{y}}(\it{t})\|<\infty$ and hence by part $(i)$ of this theorem we get that $\bm{\mathrm{y}}(t)\to\mathcal{Y}$ 
as $t\to\infty$. Since for every $t\geq0$, $\bm{\mathrm{y}}(\it{t})\in A$, we get that $\mathcal{Y}\cap A\neq\emptyset$. Further for some $y\in 
\mathcal{Y}$, $\left(\cap_{\delta>0}\left\{y'\in\mathbb{R}^{d_2}:Q_{\mu}(y)-Q_{\mu}(y')<\delta\right\}\right)\cap A=\mathcal{Y}\cap A$. 
For any $\epsilon>0$, there exists $\delta_{\epsilon}>0$ such that $\left\{y'\in\mathbb{R}^{d_2}:Q_{\mu}(y)-Q_{\mu}(y')<\delta\right\}
\cap A\subseteq N^{\epsilon}(\mathcal{Y}\cap A)$ where $N^{\epsilon}(\cdot)$ denotes the $\epsilon$-neighborhood of a set. By using 
Lemma\ref{et}$(ii)$, it is easy to show that $\Phi^A\left(\left\{y'\in\mathbb{R}^{d_2}:Q_{\mu}(y)-Q_{\mu}(y')<\delta\right\}\cap A,[0,\infty)\right)
\subseteq N^{\epsilon}(\mathcal{Y}\cap A)$ where $\Phi^A$ denotes the flow of o.d.e. \eqref{algotmp2} restricted to set $A$ 
(see section \ref{diatls}). Therefore $\mathcal{Y}\cap A$ is an attracting set for the flow $\Phi^A$. From 
\cite[Prop.~3.20]{benaim1} we get that $\mathcal{Y}\cap A=A$. Therefore as $n\to\infty$, $Y_n(\omega)\to\mathcal{Y}$.
\item [(b)] Follows from part $(iii)(a)$ of this theorem and Lemma \ref{idontknow}.\qed
\end{itemize}
\end{itemize}
\end{proof}

%% file: conc_and_direc_for_future_work.tex
\section{Conclusions and directions for future work}
\label{cadffw}
We have presented a detailed analysis of a two timescale stochastic recursive inclusion with set-valued drift functions and in the presence of 
non-additive iterate dependent Markov noise with non-unique stationary distributions. Analysis in section \ref{recanal} shows us that the 
asymptotic behavior of the two timescale recursion \eqref{2tr} is such that the faster timescale iterates in recursion 
\eqref{fstr}, track the flow of DI \eqref{fstrdi} for some fixed value of the slower timescale variable and the slower timescale iterates track 
the flow of DI \eqref{stsrdi}. The assumptions under which the two timescale recursion is studied in this paper is weaker than those in current 
literature. Recursions with such behavior are often required to solve nested minimization problems which arise in machine learning and optimization. 
A special case of constrained convex optimization with linear constraints is considered as an application where the objective function is not assumed 
to be differentiable and further the objective function and constraints are averaged with respect to stationary distribution of an underlying 
Markov chain. When the transition law and hence the stationary distribution is not known in advance, a primal descent-dual ascent algorithm as 
in recursion \eqref{algo} can be implemented with the knowledge of the sample paths of the underlying Markov chain and the analysis presented in 
this paper guarantees convergence to an $\epsilon$-optimal solution for a user specified choice of $\epsilon$.

We outline a few important directions for future work.
\begin{itemize}
 \item [(1)] For two timescale stochastic approximation schemes with set-valued mean fields, to the best of our knowledge there are no sufficient conditions for stability in 
 current literature. We believe extensions of the stability result for single timescale stochastic approximation as in \cite{borkarmeyn,arunstab}, 
 can be made to the case of two timescale recursions. Another approach to stability could be along the lines of \cite{mouli}.
 \item [(2)] In many applications the iterates are projected at each time step and are ensured to remain within a compact, convex set. Such 
 projections often arise due to inherent need of the application or is used to ensure stability. Such projected schemes have a tendency to introduce 
 spurious equilibrium points at the boundary of the feasible set. Further complications arise due to the presence of Markov noise terms since 
 the projection map is most of the time not differentiable but only directional derivatives are known to exist. Such projected stochastic approximation schemes for single-valued case without 
 Markov noise component are analyzed in \cite{dupuis} and should serve as a basis for analyzing more general frameworks with projection.
 \item [(3)] In some applications arising in reinforcement learning,  the noise terms are not Markov by themselves, but their lack of Markov property comes through the dependence on a control sequence. Under such controlled Markov noise assumption, two timescale stochastic approximation scheme has been analyzed in \cite{prasan} but with single-valued, Lipschitz continuous 
 drift functions. Extending the analysis presented in this paper to the case with set-valued drfit function and controlled Markov noise assumption is straightforward and requires no major change in the overall flow of the analysis. This extension allows one to analyse the asymptotic behavior of a larger class of reinforcement learning algorithms (see \cite{leslep}).
 \item [(4)] Several other applications, such as two timescale controlled stochastic approximation, two timescale approximate drift problem also can be analyzed with the help of the results presented in this  paper (see \cite[ch.~5.3]{borkartxt} for definitions of the above).
\end{itemize}

%% file: apnd.tex
\appendix

\section{Proof of Lemma \ref{lim1}} 
\label{prflim1}
Fix $\omega\in\Omega_1$, $l\geq1$ and $T>0$. We prove the claim along the sequence $\{t^s(n)\}_{n\geq1}$ from which the claim of Lemma \ref{lim1}
easily follows.
 
 Fix $n\geq0$. Let $\tau^2(n,T):=\min\left\{m>n:t^s(m)\geq t^s(n)+T\right\}$. Let $q\in[0,T]$. Then, there exists $k$ such that $t^s(n)+q\in
 [t^s(k),t^s(k+1))$ and $n\leq k\leq \tau^2(n,T)-1$. By definition of $\bar{y}(\cdot)$ and $\tilde{y}^{(l)}(\cdot;t^s(n))$, we have that, 
 $\bar{y}(t^s(n)+q)=\alpha y_k+(1-\alpha)y_{k+1}$ and $\tilde{y}^{(l)}(q;t^s(n))=\alpha\tilde{y}^{(l)}(t^s(k)-t^s(n);t^s(n))+(1-\alpha)
 \tilde{y}^{(l)}(t^s(k+1)-t^s(n);t^s(n))$ where $\alpha=\frac{t^s(k+1)-t^s(n)-q}{t^s(k+1)-t^s(k)}$. Since $\tilde{y}^{(l)}(\cdot;t^s(n))$ is a 
 solution of the o.d.e. \eqref{psdode1}, we have that, for every $k\geq n$, $\tilde{y}^{(l)}(t^s(k)-t^s(n);t^s(n))=y_n+\sum_{j=n}^{k-1}b(j)
 h^{(l)}_2(x_j,y_j,s^{(2)}_j,u^{(l)}_j)$ and by Lemma \ref{paramuse}, we have that, for every $k\geq n$, $y_k=\bar{y}(t^s(k))=y_n+\sum_{j=n}^{k-1}b(j)
 h^{(l)}_2(x_j,y_j,s^{(2)}_j,u^{(l)}_j)+\sum_{j=n}^{k-1}b(j)m^{(2)}_{j+1}$. Thus, 
 \begin{align*}
 \|\bar{y}(t^s(n)+q)-\tilde{y}^{(l)}(q;t^s(n))\|&\leq \|\alpha\sum_{j=n}^{k-1}b(j)m^{(2)}_{j+1}+(1-\alpha)
                              \sum_{j=n}^{k}b(j)m^{(2)}_{j+1}\|\\
                              &\leq\alpha\|\sum_{j=n}^{k-1}b(j)m^{(2)}_{j+1}\|+(1-\alpha)\|\sum_{j=n}^{k}b(j)m^{(2)}_{j+1}
                              \|\\
                              &\leq \sup_{n\leq k\leq \tau(n,T)}\|\sum_{j=n}^{k}b(j)m^{(2)}_{j+1}\|.
 \end{align*}
 Since the r.h.s. of the above inequality is independent of $q\in[0,T]$, we have, $\\\sup_{0\leq q\leq T}\|\bar{y}(t^s(n)+q)-
 \tilde{y}^{(l)}(q;t^s(n))\|\leq\sup_{n\leq k\leq \tau^2(n,T)}\|\sum_{j=n}^{k}b(j)m^{(2)}_{j+1}\|$. Therefore, 
 $\\\lim_{n\to\infty}\sup_{0\leq q\leq T}\|\bar{y}(t^s(n)+q)-\tilde{y}^{(l)}(q;t^s(n))\|\leq\lim_{n\to\infty}
 \sup_{n\leq k\leq \tau^2(n,T)}\|\sum_{j=n}^{k}b(j)m^{(2)}_{j+1}\|$. Now the claim follows follows from assumption $(A7)$.
 
 \section{Proof of Lemma \ref{lim2}}
 \label{prflim2}
 Fix $l\geq1$, $\omega\in\Omega_1$. By assumption $(A8)$, we know that there exists $r>0$ such that 
 $\sup_{n\geq0}(\| x_n\|+\|y_n\|)\leq r$ and hence $\sup_{t\geq0}\tilde{y}^{(l)}(0;t)=\sup_{t\geq0}\bar{y}(t)\leq r$. 
 
 For any $t\geq0$, let $[t]:=\max\left\{n\geq0:t^s(n)\leq t\right\}$. For every $t\geq0$ and $q_1,q_2\in[0,\infty)$ (w.l.o.g. assume $q_1<q_2$) 
 we have, 
 \begin{align*}
  \|\tilde{y}^{(l)}(q_1;t)-\tilde{y}^{(l)}(q_2;t)\|&=\|\int_{q_1}^{q_2}h^{(l)}_2(x_{[t+q]},y_{[t+q]},s^{(2)}_{[t+q]},u^{(l)}_{[t+q]})dq\|\\
                                         &\leq\int_{q_1}^{q_2}\| h^{(l)}_2(x_{[t+q]},y_{[t+q]}s^{(2)}_{[t+q]},u^{(l)}_{[t+q]})\| dq\\
                                         &\leq\int_{q_1}^{q_2}K^{(l)}(1+\| x_{[t+q]}\|+\|y_{[t+q]}\|)dq\\
                                         &\leq C^{(l)}(q_2-q_1),
 \end{align*}
where $C^{(l)}:=K^{(l)}(1+r)$ and $r>0$ is such that, $\sup_{n\geq0}(\| x_n\|+\|y_n\|)\leq r$. Thus 
$\left\{\tilde{y}^{(l)}(\cdot;t)\right\}_{t\geq0}$ is an equicontinuous family. Now the claim follows from Arzella-Ascoli theorem.